\documentclass[11pt]{llncs}
\usepackage{a4wide}

\usepackage{graphicx}

\usepackage{amsmath,amsfonts,amssymb,epsfig}
\usepackage{subfigure}
\usepackage{cite}

\newcommand{{\HC}}{{{Push \& Pull}}}

\begin{document}

\title{Push \& Pull: autonomous deployment of mobile sensors for a complete coverage
\thanks{Animations and the complete code of the proposed algorithm are available for download at the address
http://www.dsi.uniroma1.it/\small{$\sim$}novella/mobile\_sensors/}\\
\bigskip \emph{TECHNICAL REPORT}}

\date{}

\author{Novella Bartolini         \and
        Tiziana Calamoneri \and
       Emanuele Guido Fusco \and
    Annalisa Massini \and
      Simone Silvestri}
\institute{Department of Computer Science\\
"Sapienza" University of Rome, Italy\\
              \email{\normalsize \{bartolini, calamo, fusco, massini, simone.silvestri\}@di.uniroma1.it}}
 \maketitle

\pagestyle{plain} \pagenumbering{arabic}

\begin{abstract}
Mobile sensor networks are important for several strategic applications
devoted to monitoring critical areas. In such hostile scenarios, sensors cannot be deployed manually and are either sent from a safe location
or dropped from an aircraft.
Mobile devices permit a dynamic deployment reconfiguration that improves the coverage in terms of completeness and uniformity.

In this paper we propose a distributed algorithm for the autonomous deployment
of mobile sensors called {\HC}. According to our proposal, movement decisions are made by each
sensor on the basis of
locally available information
and do not require any prior knowledge of the operating conditions
or any manual tuning of key parameters.

We formally prove that, when a sufficient number of sensors are available, our approach guarantees
a complete and uniform
coverage.
Furthermore, we
demonstrate that the algorithm execution always
terminates preventing movement oscillations.

Numerous simulations show that our
algorithm reaches a complete coverage
within reasonable time with moderate energy consumption, even when the
target area has irregular shapes.
Performance comparisons between {\HC} and one of the most acknowledged algorithms
show how the former one can efficiently reach a more uniform and complete coverage
under a wide range of working scenarios.
\end{abstract}

\section{Introduction}

Research in the field of mobile wireless sensor networks is motivated by
the need to monitor hostile environments such as
wild fires, disaster areas, toxic regions or battlefields, where
static sensor deployment cannot be performed manually.

In these  working settings, sensors may
be dropped from an aircraft or
sent from a safe location.
Mobile sensors can dynamically adjust their position to improve the coverage with
respect to their initial deployment.

This paper addresses the problem of coordinating sensor movements to reach a more
satisfactory deployment in terms of coverage extension and uniformity.

Centralized solutions to this problem are  inefficient because they  require either a
prior assignment of sensors to positions, or
a starting topology that ensures the connectivity
of all sensors (for global coordination purposes).
On the one hand, a prior assignment is
inapplicable because it  requires an excessive
amount of movements
to deploy sensors independently  of their initial position.
On the other hand, connectivity  cannot be guaranteed in
any starting scenario.
Therefore, feasible and scalable solutions should employ a distributed
scheme according to which sensors make local
decisions to meet global objectives.

When designing solutions to the deployment problem, energy consumption is an important issue.
Indeed, due to the limited power available, each sensor should coordinate with others
with very few messages and should reach its position traversing small distances.
Energy consumption should also be controlled by uniformly placing redundant sensors when available.
In fact, a uniformly redundant coverage of the AoI
allows to prolong the network lifetime,
for example by allowing an alternative activation of sensors without any loss of coverage.
A redundant sensor placement has also several benefits
as it allows a better target sensing,
stronger environmental monitoring, and fault tolerance capabilities.

The main contribution of this paper is an original fully distributed algorithm for mobile sensor
deployment called {\HC}, which is radically different from any previous one.
Most of the existing approaches fall into one of two main categories,
as they are  either inspired by molecular physics \cite{Howard2002,Zou2003,Heo2005,Chen03,Suckme2004,Pac2006,Kerr2004,Chiasserini07} or
by computational geometry \cite{LaPorta06,Yang2007,LaPorta04,Tan08,LaPorta_Relocation}.
In general, they aim at reaching a final deployment which is similar to the one
targeted by our algorithm. Nevertheless,  the solutions inspired by physical models usually tend to
non-stable deployment, due to the dynamicity of the equilibrium that characterizes molecular systems.
Hence such solutions necessitate proper countermeasures to ensure a gradual decrease of movements.
On the other hand the approaches inspired by computational geometry are often unable to handle concave AoIs and lead to non-uniform deployments.

The design of our solution follows the grassroots approach \cite{Babaoglu2005} to autonomic computing.
Self-organization emerges without the need of external coordination or  human intervention
as the sensors autonomously adapt their position on the basis of a
local view of the surrounding scenario. This way our algorithm
shows the basic self-* properties of autonomic computing,
i.e. self-configuration, self-adaptation and self-healing.

This algorithm produces a hexagonal tiling by spreading
sensors out of high density regions and attracting them towards coverage holes.
Decisions regarding the behavior of each sensor are
based on locally available information and do not
require any prior knowledge of the operative
scenario or any manual tuning of key parameters.
Location awareness is only necessary in the case
of sensor deployment over a specific target area, whereas this capability is not required when sensors
are to be
deployed in an open environment.

We formally prove that our algorithm terminates and provides a complete coverage regardless of the particular shape of the AoI;
moreover, we propose a variant that exploits redundant sensors to produce a $k$-coverage, where $k$ depends on
the number of the available sensors and on the shape and extension of the AoI.

We ran numerous simulations to evaluate the performance of our algorithm
and compare it to existing solutions. Experimental results show that our algorithm
reaches a complete and stable coverage
within reasonable time with moderate energy consumption, even when the
target area has an irregular shape.
It turns out that our proposal provides better performance than one of the most acknowledged and
cited algorithms \cite{LaPorta06}.
Furthermore, our solution also outperforms  previous approaches producing a redundant coverage with guaranteed uniformity.

This paper is organized as follows.
In Section \ref{sec:algorithm} we
describe the \HC\ algorithm.
We devote Section \ref{sec:discussion} to a discussion on the implications of coverage uniformity on
fault tolerance and network lifetime. In this section we also propose an algorithm variant which privileges uniformity over other performance requirements.
In Section \ref{sec:properties} we formally prove some important properties of the final deployment, namely
termination, coverage completeness and uniformity.
The simulation analysis is shown in Section \ref{sec:exp_results}.
Section \ref{sec:related_work}
describes the state of the art,  while
Section
\ref{sec:conclusions} concludes the paper.

\section{The Push \& Pull algorithm} \label{sec:algorithm}

In order to make the exposition clearer, we  outline
the algorithm, before giving deeper details.
\subsection{The idea}
\label{sec:idea}

Sensors aim at realizing a complete and uniform coverage of the AoI by means of a hexagonal tiling.
Notice that the hexagonal tiling corresponds to a
triangular lattice arrangement, that is the one that guarantees
optimal coverage and density, as discussed in \cite{Brass2007}, and connectivity, as we detail in section \ref{sec:lattice}.
The algorithm starts with the concurrent creation of several tiling portions.
Every sensor not yet involved in the creation of a tiling portion gives start to
its own portion in an instant which is randomly selected in a given time interval.

In the following, when we talk about $s_\texttt{init}$
we refer, more in general, to any starter.

The algorithm mandates that four main activities are carried out in
an interleaved manner. The combination of the described activities expands the
tiling and, at the same time, does its best to
uniformly distribute redundant sensors over the
tiled area, preventing oscillations.

\paragraph{Snap activity.}

The sensor $s_\texttt{init}$  elects its position as the center of the first hexagon
of its tiling portion.
It selects at most six sensors
among those located within its {\em transmission
radius} $R_\texttt{tx}$ and makes them snap to the
center of adjacent hexagons.
Such deployed sensors, in turn, give start
to their own selection and snap activity,
thus expanding the boundary of the current tiling portion.
The sensors that are positioned in the center of a hexagon according to the snap activity,
are hereafter referred to as {\em snapped sensors}.
This activity continues until no other snaps are
possible, because either the whole AoI is covered,
or  the boundary tiles do not contain any unsnapped sensors.

\paragraph{Push activity.}
After the completion of their snapping activity, snapped sensors
may still be surrounded by non-snapped sensors located inside their hexagon, hereafter referred to as
their {\em slaves}.
In this case, they proactively push such
slaves towards lower density areas
located
within their transmission range.
Consequently, slaves being in overcrowded
areas migrate to low density zones, thus
accelerating the coverage process and enhancing its uniformity.
A snapped sensor stops the push activity when the maximum detected density difference does not
exceed one sensor.

\begin{figure*}[t]
\begin{center}
\begin{tabular}{c c c c c c}
{\scalebox{0.30}{
\includegraphics[]{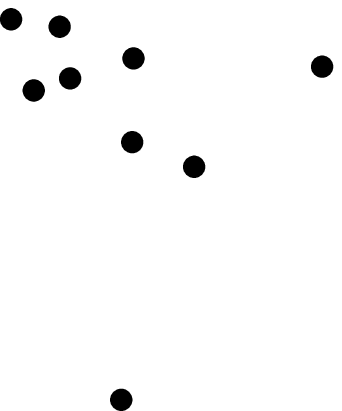}}}
&
{\scalebox{0.30}{
\includegraphics[]{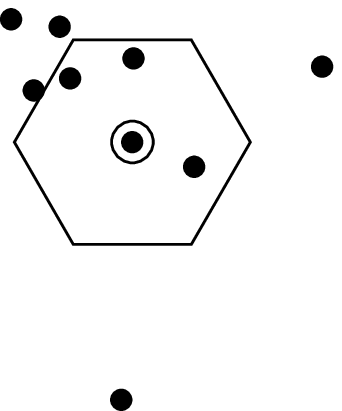}}}
&
{\scalebox{0.30}{
\includegraphics[]{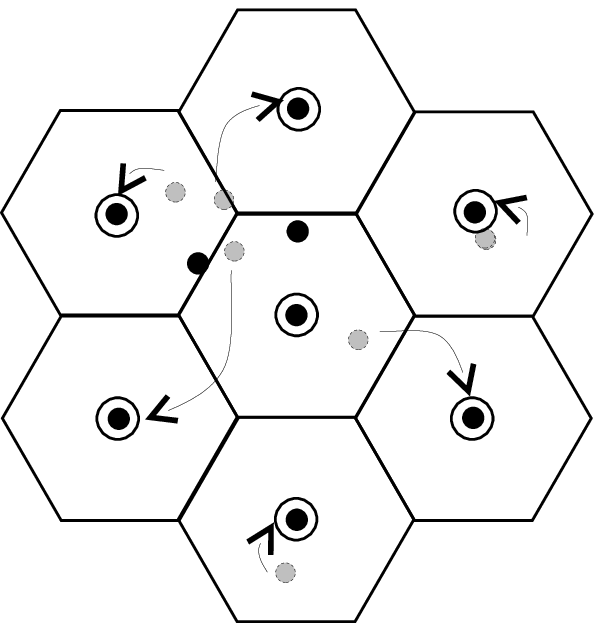}}}
&
{\scalebox{0.30}{
\includegraphics[]{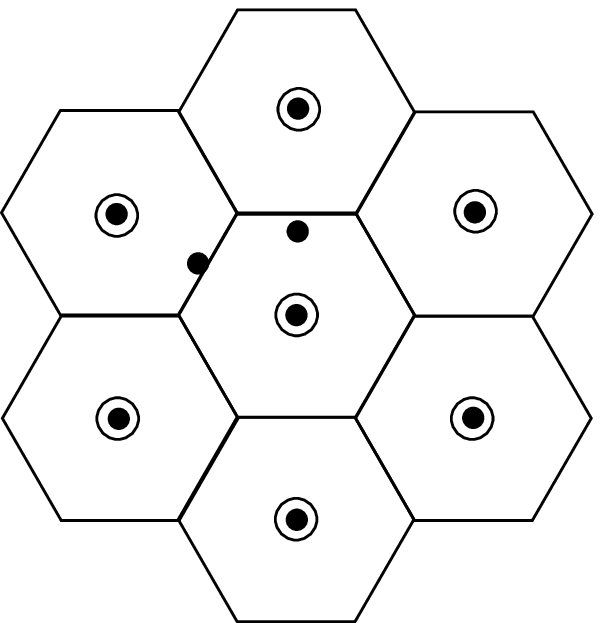}}}
&
{\scalebox{0.30}{
\includegraphics[]{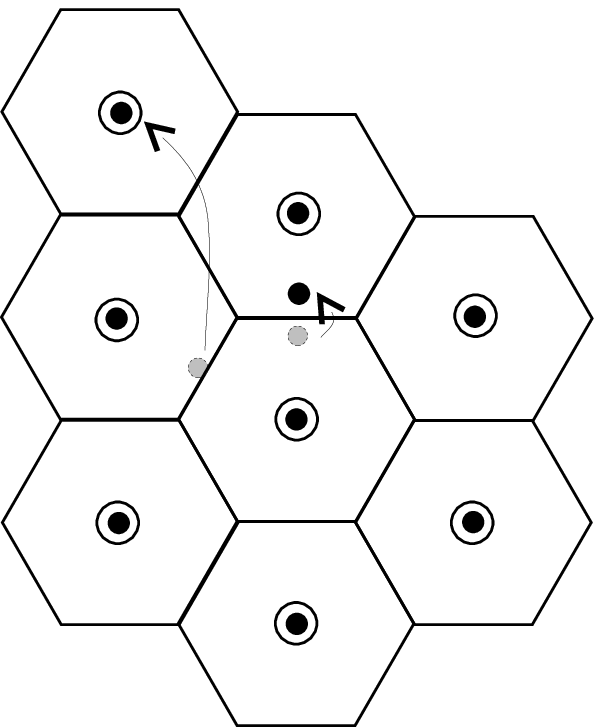}}}
&
{\scalebox{0.30}{
\includegraphics[]{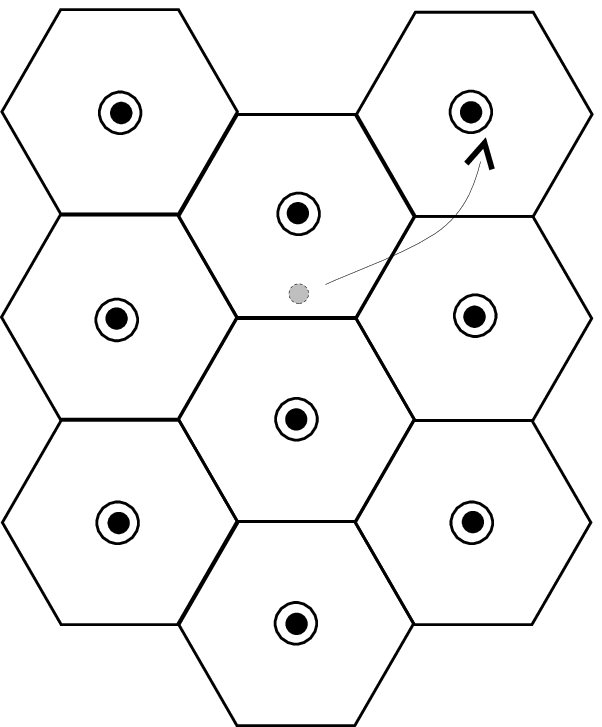}}}\\
~&
\\
(a)&(b)&(c)&(d)&(e)&(f)
\end{tabular}
\end{center}
\caption{An example of snap and push activities:
(a) starting configuration; (b) $s_\texttt{init}$ snaps itself at the center of the first tile; (c) $s_\texttt{init}$ selects six sensors to make them snap in the adjacent hexagons; (d) configuration after the snap activity of $s_\texttt{init}$; (e) $s_\texttt{init}$ pushes a sensor to a nearby hexagon, while a just deployed sensor gives rise to a new snap activity; (f) a snapped sensor causes the snap of the sensor that it has just received from the starter.
} \label{fig:example}
\end{figure*}

\paragraph{Pull activity.}
Snapped sensors may detect a coverage hole adjacent to their hexagon
and may not find available sensors to make them snap. In this
case, they send hole trigger messages,
and reactively attract non-snapped sensors and make them fill the hole.
Such sensors keep on advertising the presence of a hole until either the holesis filled or a
timeout occurs.

\paragraph{Tiling merge activity.}
The possibility that many sensors act as starters can give rise to
several tiling portions with different orientations.
In order to characterize and distinguish each
tiling portion, the time-stamp of each starter is
included in the header of all messages.
As a result, messages coming from sensors located in
different tiling portions will be characterized by
different starter time-stamps.
Our algorithm provides a mechanism to merge all
these tiling portions into a unique regular
and uniformly oriented tiling.
When the boundaries of two
tiling portions come in radio proximity with each
other, the one with older starter time-stamp
absorbs the other one by making its snapped
sensors move into more appropriate snapping
positions.

\medskip
Figure \ref{fig:example} shows an example of the
execution of the first two activities.
Namely, Figure \ref{fig:example}(a) depicts the starting
configuration, with nine randomly placed sensors
and Figure \ref{fig:example}(b) highlights
$s_\texttt{init}$ starting
the hexagonal tiling.
In Figure \ref{fig:example}(c)  the
starter sensor $s_\texttt{init}$ selects six sensors to make them snap in adjacent hexagons,
according to the minimum distance criterion.
Figure \ref{fig:example}(d) shows the
configuration after the snap activity of $s_\texttt{init}$.
In Figure
\ref{fig:example}(e), a just deployed sensor starts a
new snap activity while $s_\texttt{init}$ starts
the push activity sending a non-snapped sensor to a
lower density hexagon.
In Figure \ref{fig:example}(f) one of the deployed
sensors causes the snap of the sensor just received from the
starter,
thus leading to the final configuration.
Figure \ref{fig:grid_merge} shows an example of
the execution of the tiling merge activity.
In particular, Figure \ref{fig:grid_merge}(a) shows two tiling portions meeting each other.
The portion on the left has the oldest time-stamp, hence it absorbs the other one.
Two nodes of the right portion
detect the presence of an older tiling and abandon
 their original position (Figure \ref{fig:grid_merge}(b)) to honor snap commands
coming from a sensor of the left portion (Figure \ref{fig:grid_merge}(c)).
These just snapped sensors, now belonging to the older portion,  detect the presence of three
nodes belonging to the right portion (Figure \ref{fig:grid_merge}(d)) and make them snap as soon as they
leave their original tiling portion
(Figures \ref{fig:grid_merge}(e)-(f)).
\begin{figure}[h]
\centering
 \vspace{-0.6cm} \hspace{-.6cm}
\begin{tabular}{c c c}
\subfigure[]{\scalebox{0.17}{
\includegraphics[]{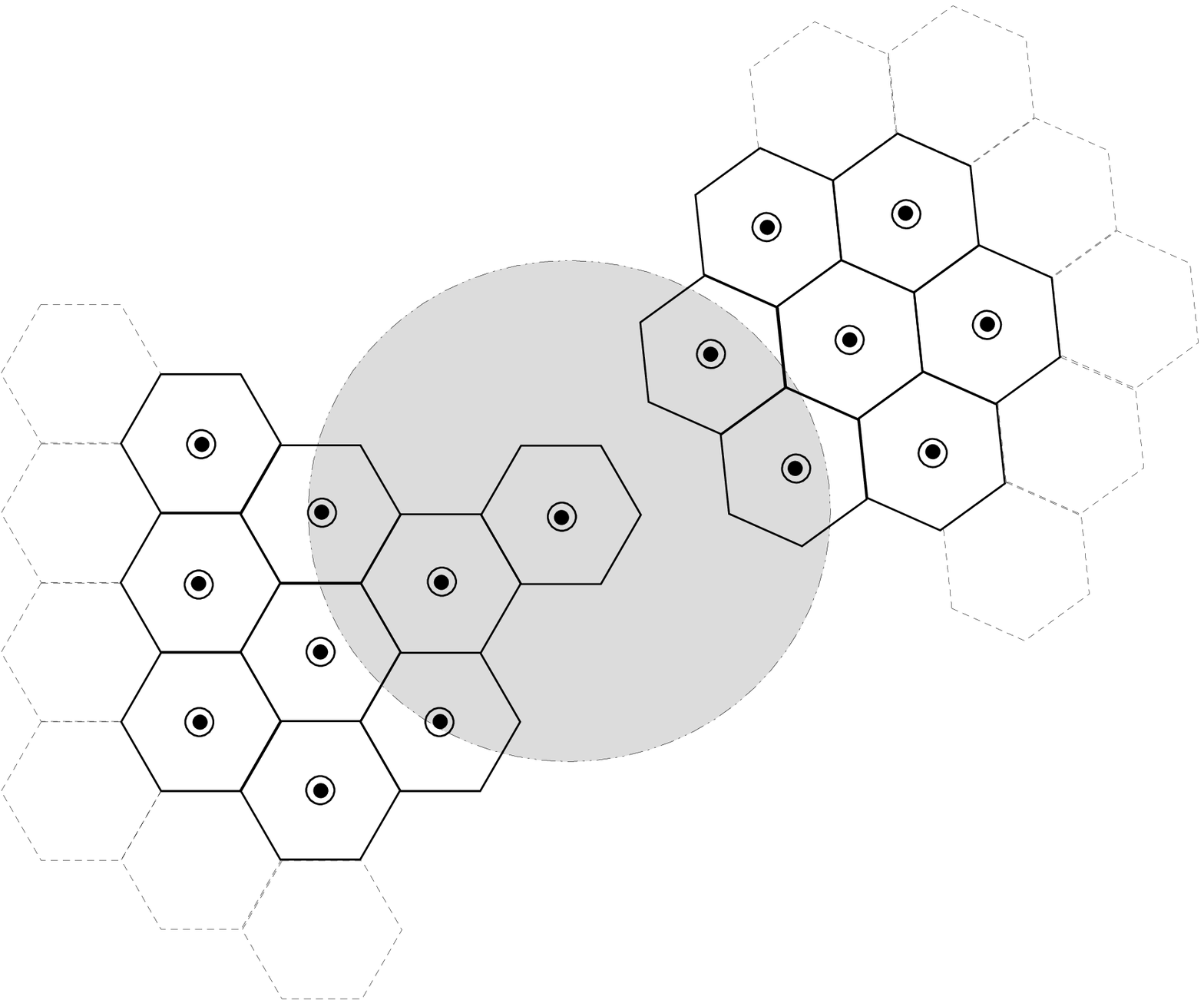}}}
& \subfigure[]{\scalebox{0.17}{
\includegraphics[]{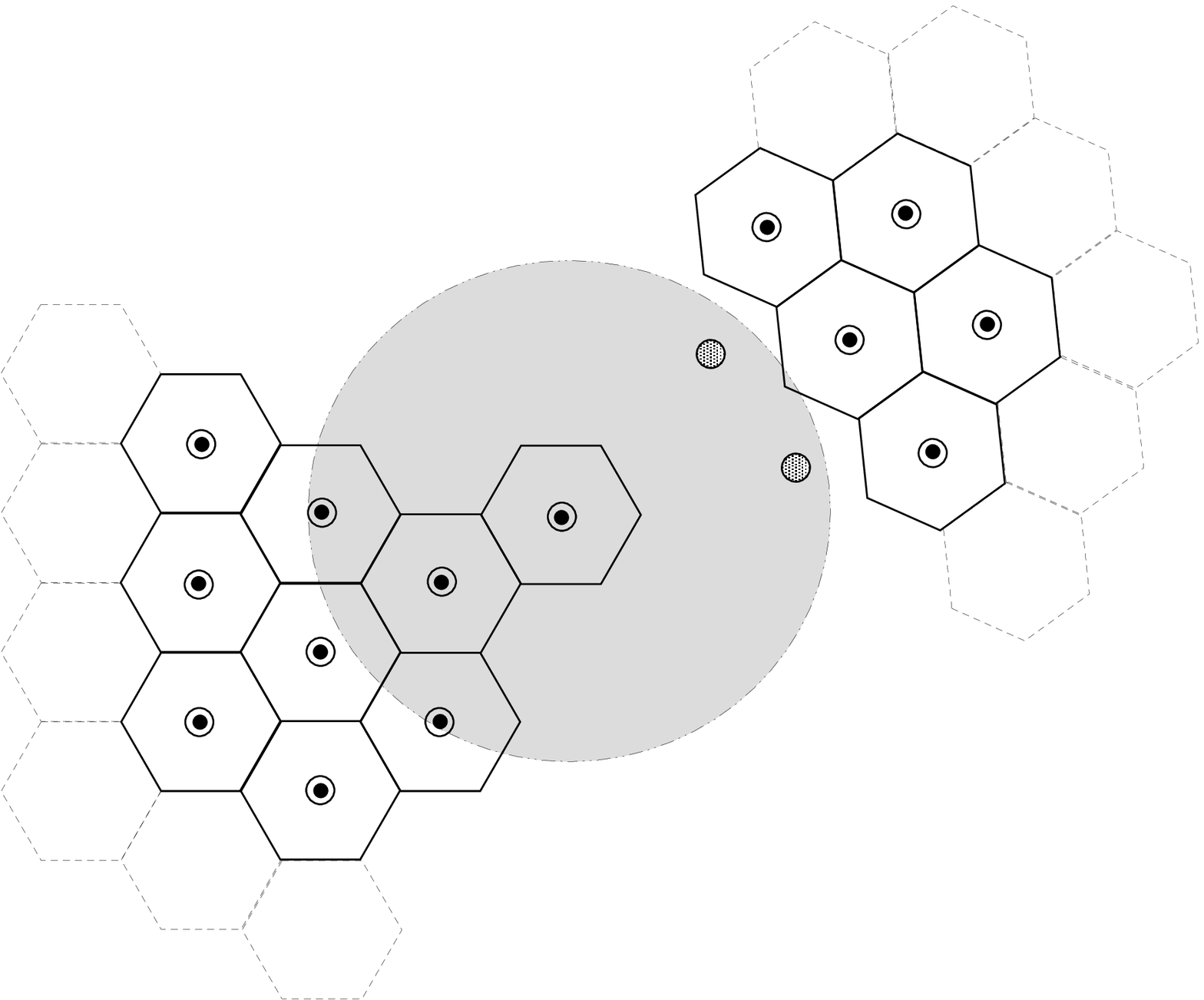}}}
& \subfigure[]{\scalebox{0.17} {\includegraphics[]{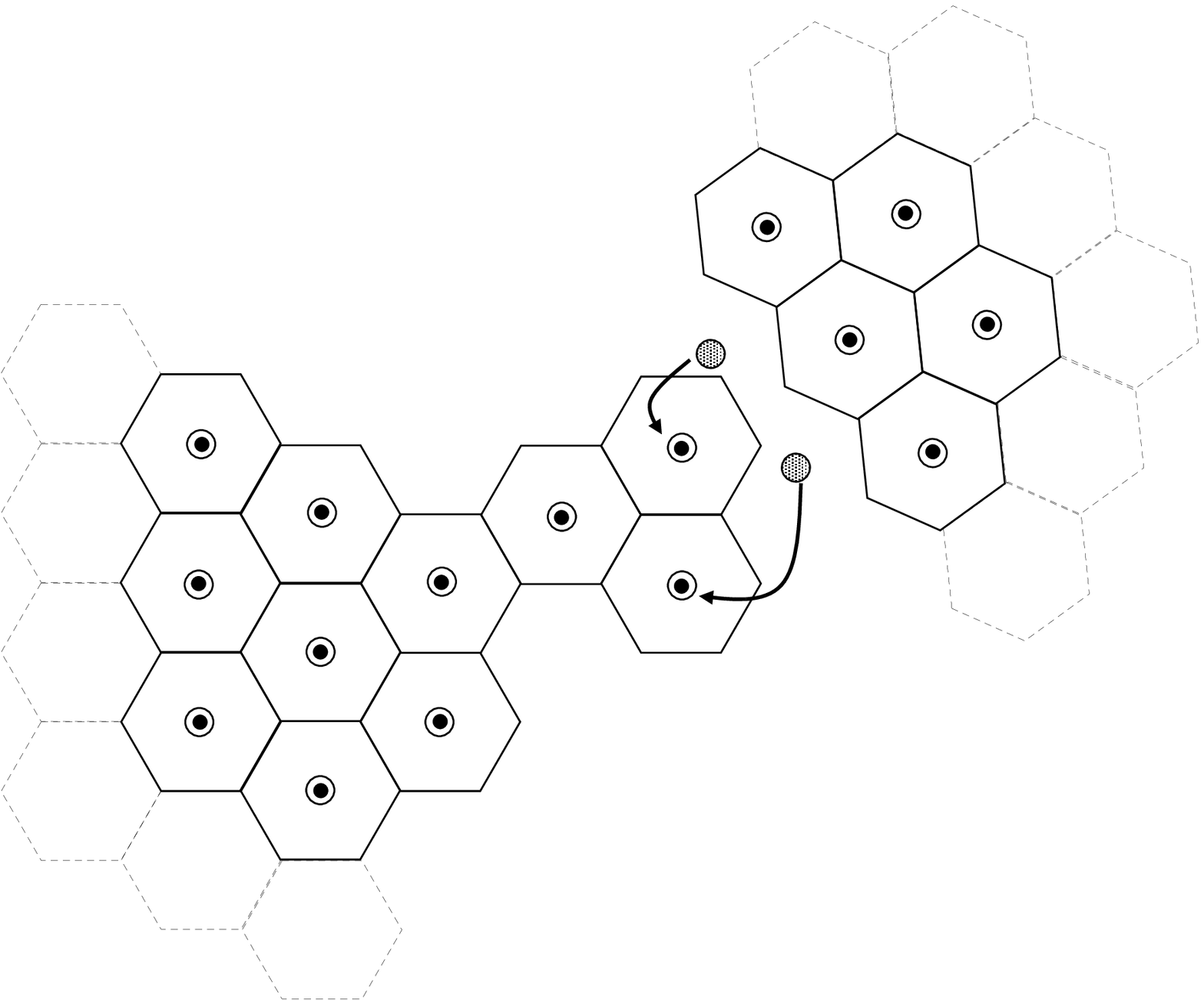}}}
\vspace{-0.8cm}
\\
\subfigure[]{\scalebox{0.17}{
\includegraphics[]{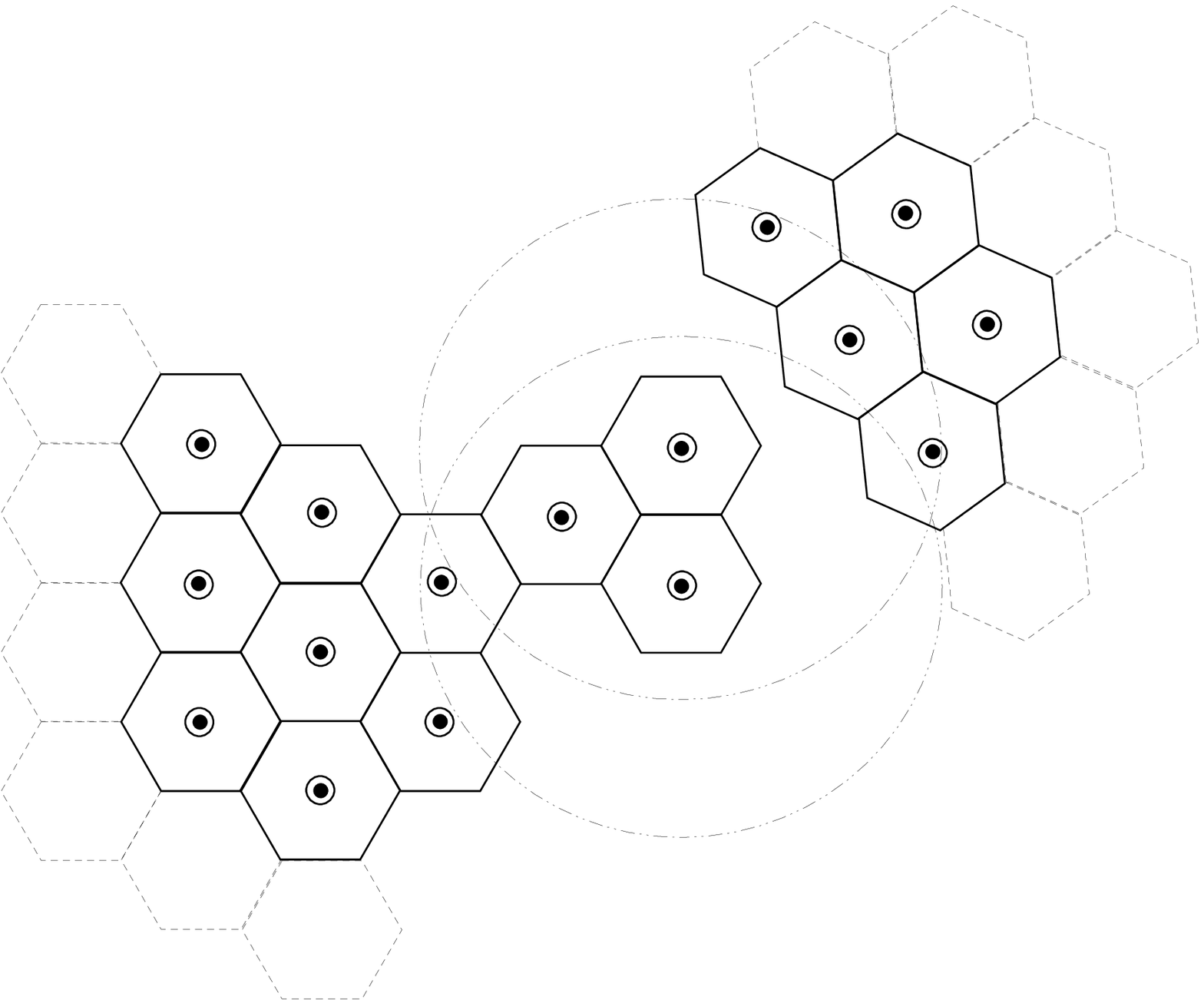}}}
& \subfigure[]{\scalebox{0.17}{
\includegraphics[]{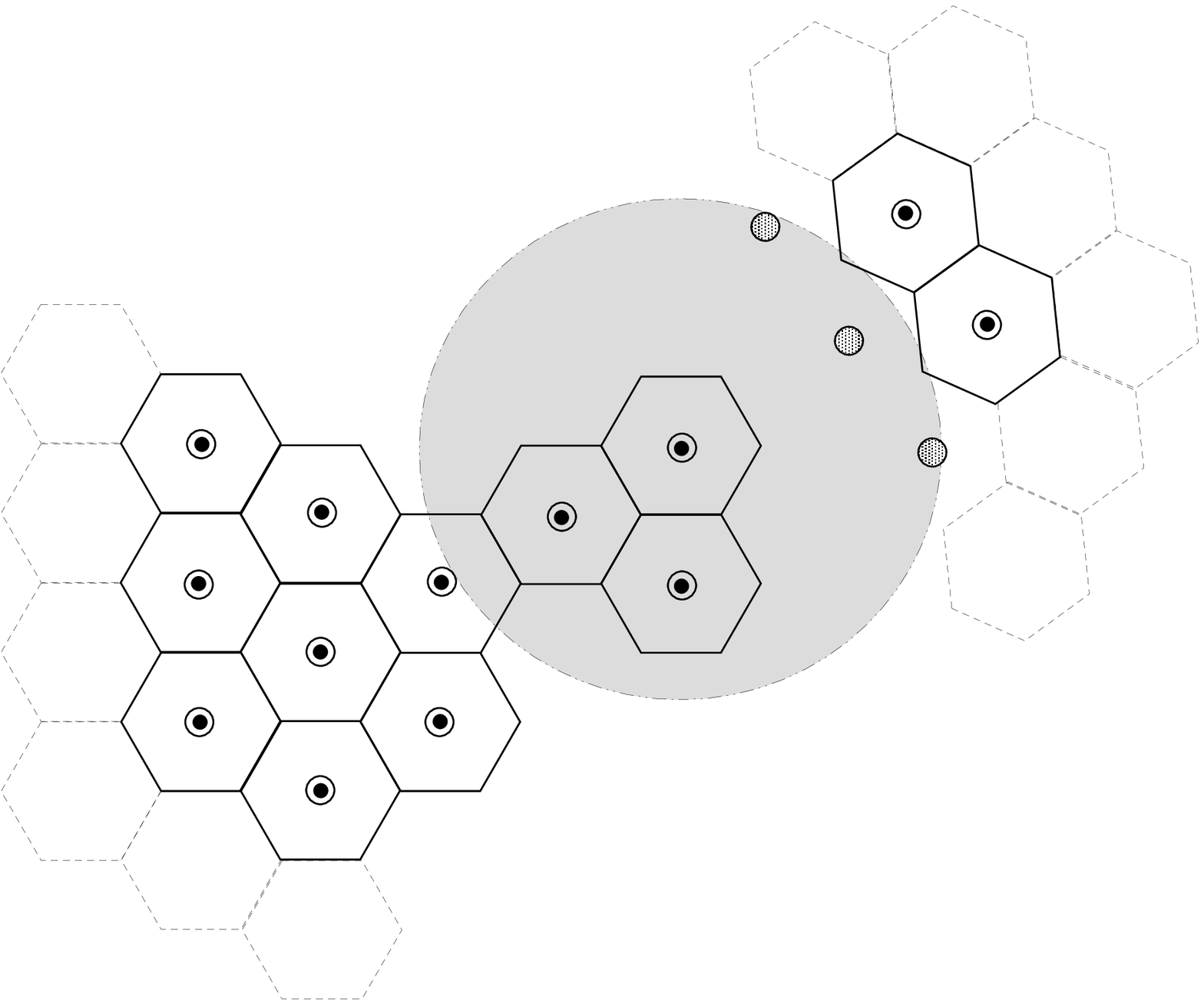}}}
&\subfigure[]{\scalebox{0.17}{
\includegraphics[]{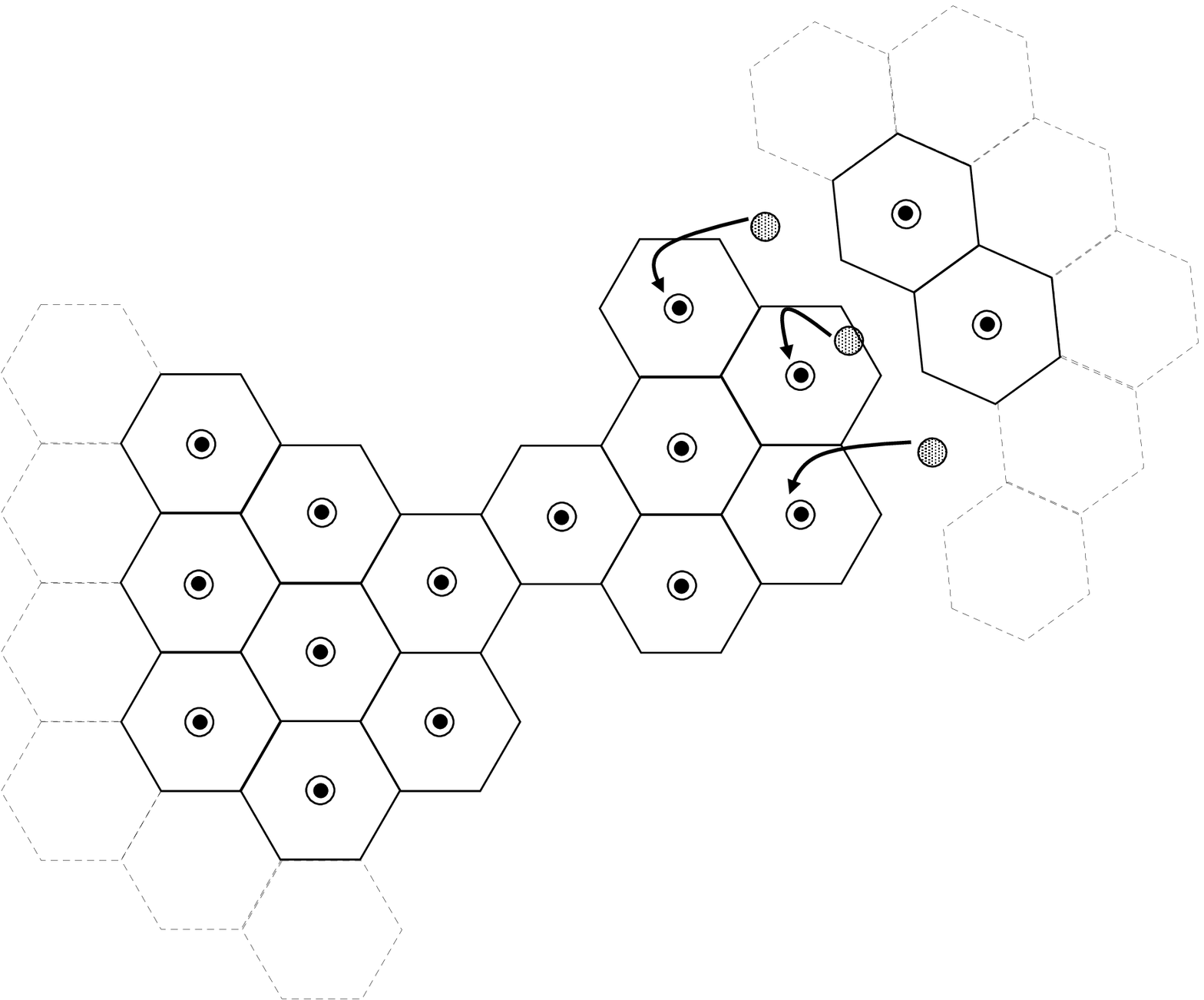}}}
\end{tabular}
\caption{An example of tiling merge activity: (a) two tiling portions meet each other (the one on the left has the oldest time-stamp); (b) two nodes of the right portion
detect the presence of the older portion;  (c) the two nodes abandon
 their original portion and are snapped to new positions in the older portion; (d) these just snapped sensors detect the presence of three
nodes belonging to the right portion and (e-f) make them snap.} \label{fig:grid_merge}
\end{figure}

We defer the introduction of the example
regarding the pull activity to the next section
 when more details will be
available to clarify the explanation.

\subsection{Details of {\HC}} \label{sec:details}
In order to describe the algorithm in more detail, we give some definitions and specify the
operative setting.

Let $V$ be a set of equally equipped sensors able to determine their own location,
endowed with boolean sensing capabilities.
We adopt an isotropic communication model and  assume that sensors are
in active mode for all the deployment phase.
We set the hexagon side length $l_\texttt{h}$ to
the {\em sensing radius} $R_\texttt{s}$.
This setting guarantees both coverage and connectivity
when $R_\texttt{tx} \geq \sqrt{3}R_\texttt{s}$. This requirement is  not restrictive as
most wireless devices can adjust their transmission range by properly setting their transmission power.

All sensors that are neither snapped nor slaves are called {\em free}.
Given a sensor $x$, snapped to the center of a
hexagon, we denote by $S(x)$
the set of slaves of $x$
and by $Hex(x)$  the hexagonal region whose center is covered by $x$.
We define $L(x)$ the set composed by the
sensors located in radio proximity from  $x$ (i.e. the free sensors in radio proximity from $x$
the slaves $S(x)$). We also refer to $VP(x)$ (vacant positions)  as to
the set of positions detected by the sensor $x$ at the center of the hexagons adjacent to $Hex(x)$ that are not yet occupied by any snapped sensor.

We now give additional details on the activities sketched in Section \ref{sec:idea}.
\paragraph{Snap activity.}

At the beginning of the deployment process, each sensor may act as starter of a snap activity from
its initial location at an instant randomly chosen over a given time interval.
In order to propagate a tiling portion, a snapped
sensor $x$ performs a {\em neighbor
discovery},
that allows $x$ to gather
information regarding $S(x)$ and all the free and
snapped sensors located in
radio proximity from $x$ and the positions belonging to $VP(x)$.
To give start to new snap activities, $x$ selects
the sensor in  $L(x)$ which is the closest to
each uncovered position and snap it there.
A snapped sensor leads the snapping of as many
adjacent hexagons as possible and gives start to
the push activity,
 as described in Figure
\ref{fig:snap_flow}.

If some of the positions in $VP(x)$ cannot be covered
because  $L(x)$ does not contain enough sensors,
$x$ starts the pull activity.
If otherwise all the hexagons adjacent to $Hex(x)$ have
been covered and $VP(x)=\emptyset$, $x$ stops any further snapping,
 and uses the available slaves (if any) to give start to the
the push activity.

Figure \ref{fig:flow_chart_extended} shows a detailed flow chart of the Snap, Push and Pull activities, in agreement with the underlying coordination protocol which is described
in \cite{ICNP2008}. In this figure, for clarity, we denote with $G(x)$ the set of snapped sensors
located in hexagons adjacent to $Hex(x)$.

\begin{figure}[t]
\begin{center}
\scalebox{0.35}{\includegraphics[]{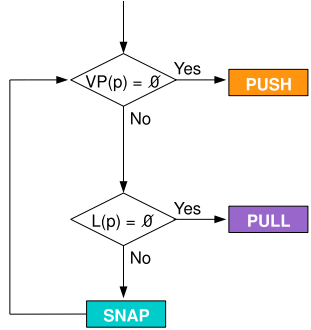}}
\end{center}
\caption{Behavior of the snapped sensor $p$.} \label{fig:snap_flow}
\end{figure}


\begin{figure}[t]
\centering
 \hspace{-0.3cm}
\scalebox{0.45}{\includegraphics[]{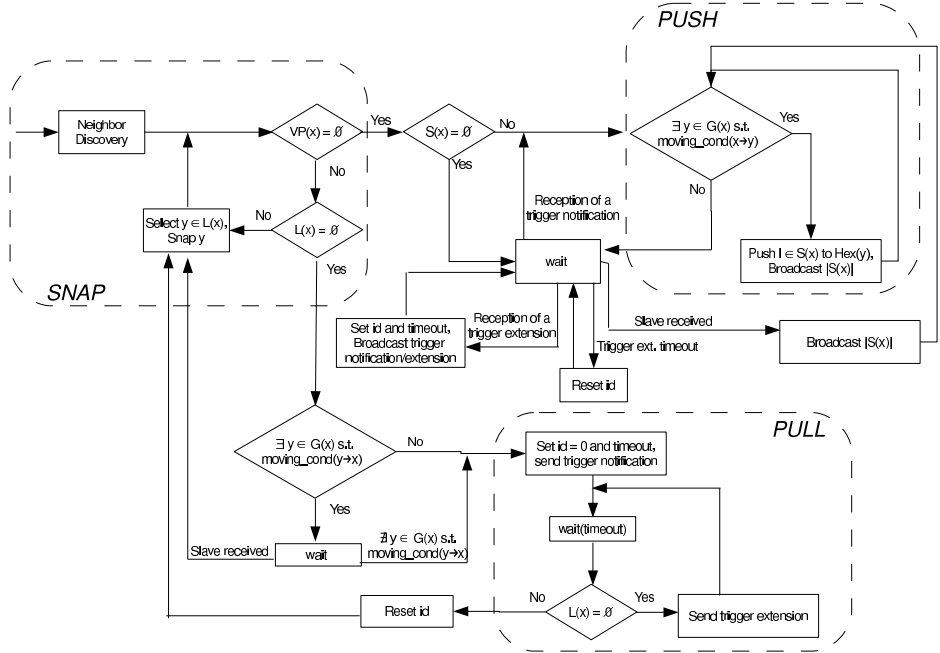}}
\caption{A detailed flow chart of the Snap, Push and Pull activities.} \label{fig:flow_chart_extended}
\end{figure}

\begin{figure*}[t]
\begin{center}
\begin{tabular}{c c c}
\subfigure[]{\scalebox{0.3}{
\includegraphics[]{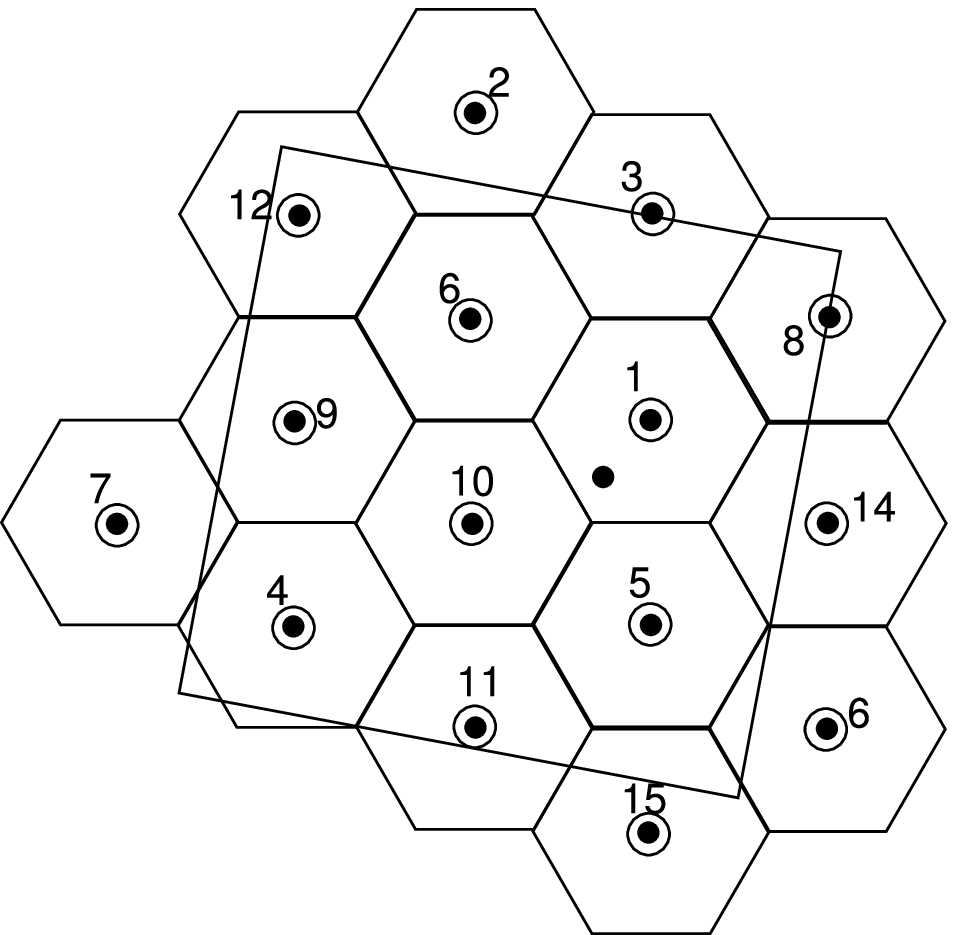}}}
&
\subfigure[]{\scalebox{0.3}{
\includegraphics[]{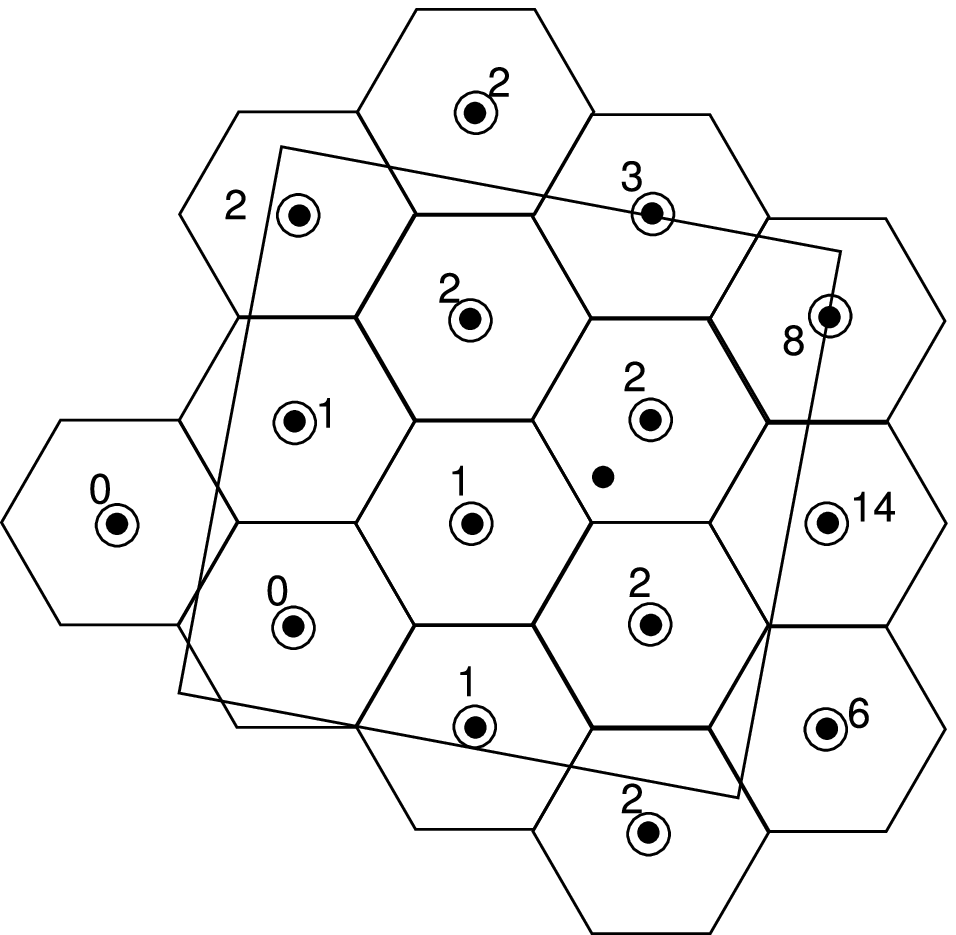}}}
&
\subfigure[]{\scalebox{0.3}{
\includegraphics[]{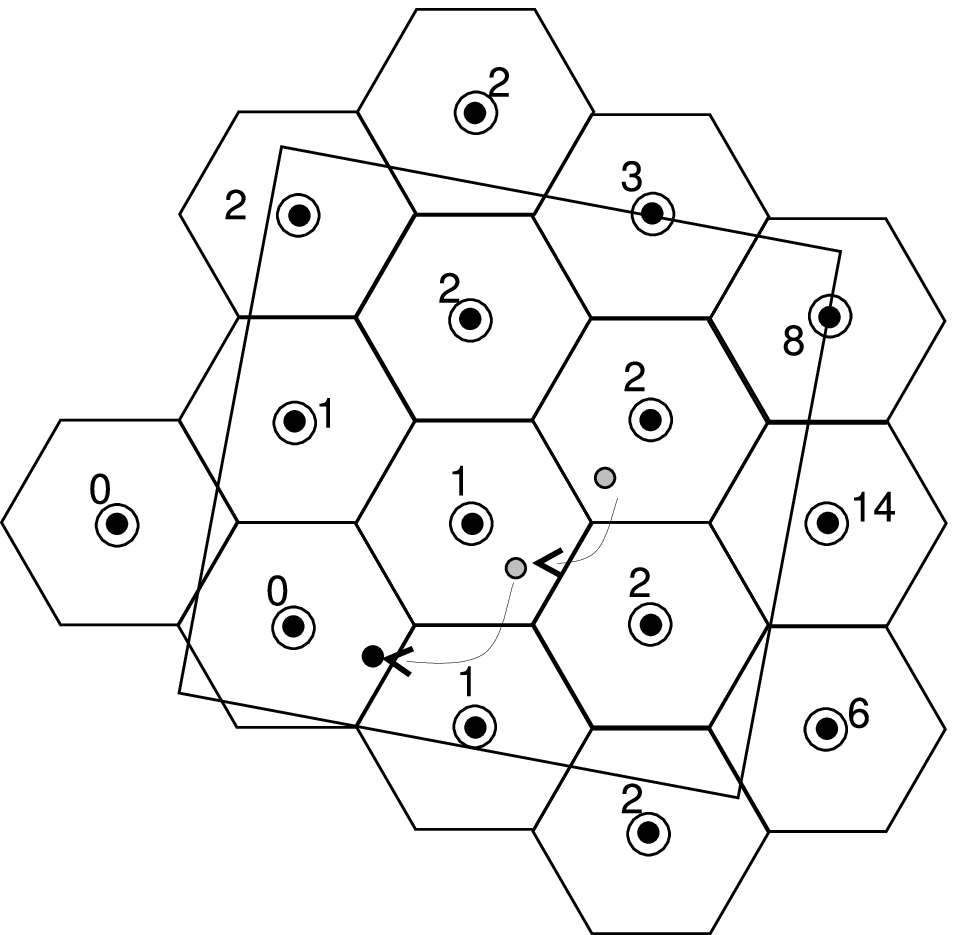}}}
\end{tabular}
\begin{tabular}{c c}
\subfigure[]{
\scalebox{0.3}{\includegraphics[]{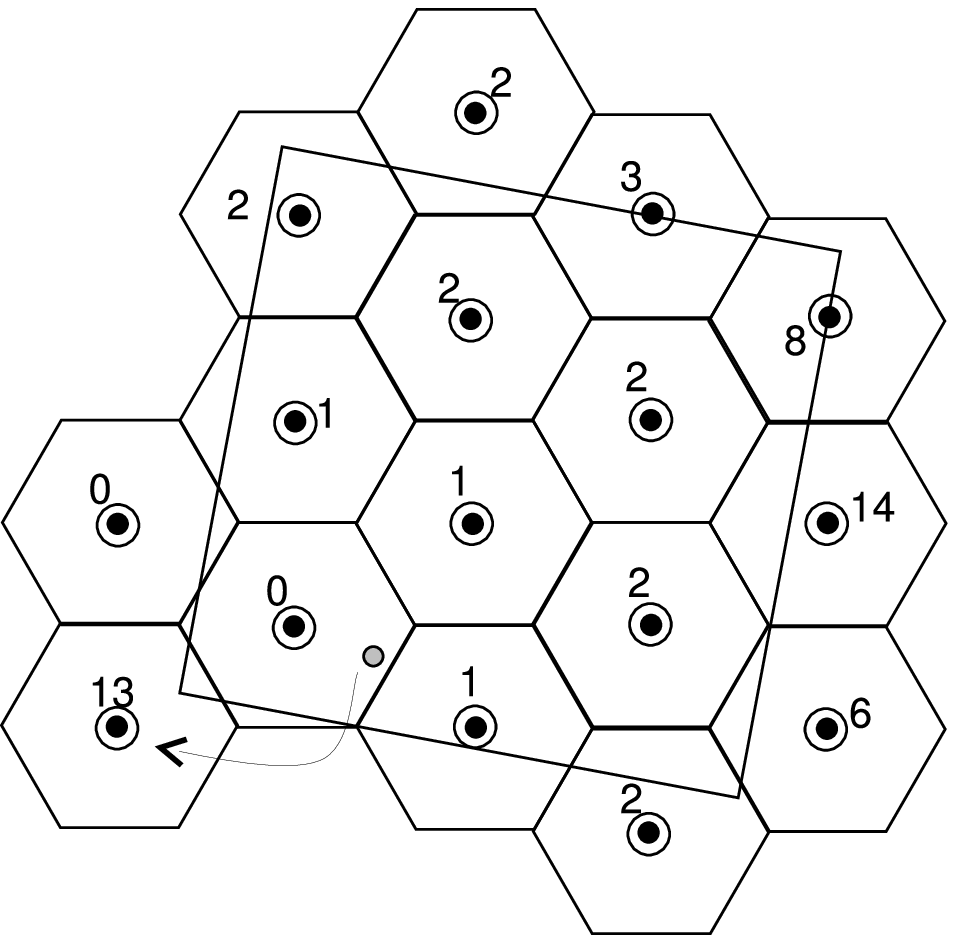}}}
&
\subfigure[]{\scalebox{0.3}{
\includegraphics[]{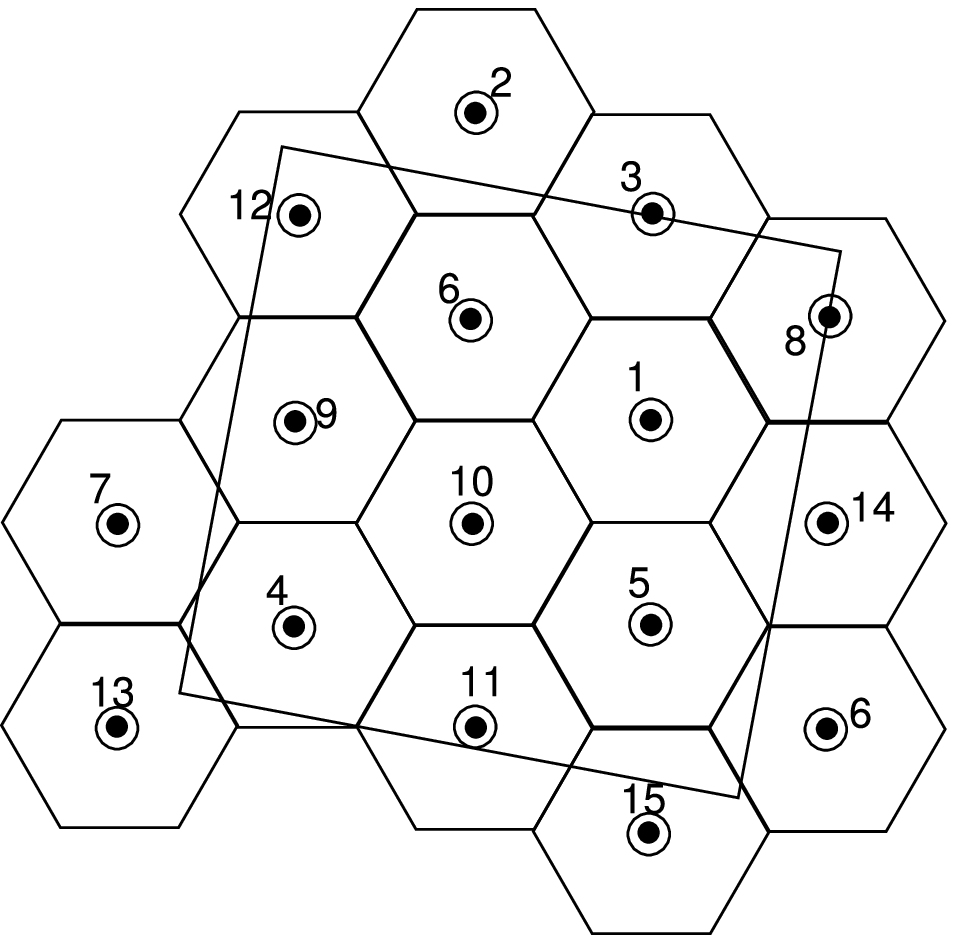}}}
\end{tabular}
\end{center}
\caption{An example of the pull activity: (a) a tiled AoI with a coverage hole in the bottom left corner;
(b) the two snapped nodes detecting the hole set their $id$ to 0 and send a trigger notification message that is propagated by their neighbors which  modify their $id$;
(c) the closest available slave moves towards the hole;
(d) the hole is covered and (e)
all $id$s are set back to the previous values.} \label{fig:pull_action}
\end{figure*}

\paragraph{Push activity.}

Given two snapped sensors $x$ and $y$ located in
radio proximity to each other,
$x$ may offer one of its slaves to $y$ and push
it inside its hexagon if $|S(x)| \geq
|S(y)|+1$.
Note that, if $|S(x)| = |S(y)|+1$,
the flow of a sensor from $Hex(x)$ to $Hex(y)$
leads to a symmetric situation in which  $|S(y)|
= |S(x)|+1$
possibly causing  endless cycles.
In such cases we restrict the push activity to
only one direction: $x$ pushes its slave to $y$
only if
$id(y)<id(x)$, where $id(\cdot)$ is a function
initially set to the unique identity code of the
sensor radio device (notice that this is not the only possibility, $id(\cdot)$
could be set for example to a random non negative value).
We formalize these observations by defining the
following {\em Moving Condition}, that enables
the movement of a sensor from $Hex(x)$ to
$Hex(y)$:
$$\{|S(x)| > |S(y)|+1 \} \textrm{ }\vee \textrm{ }
 \{ |S(x)| = |S(y)|+1  \textrm{ } \wedge \textrm{
 } id(x)>id(y) \}.
 $$

The snapped sensor $x$ executes a push action by sending one of its slaves $s$
towards the hexagon of a snapped sensor
$y$.

The destination hexagon $Hex(y)$ is selected such that $x$ verifies the moving condition with respect to $y$.
In particular, as destination of the push action, $x$ selects the closest hexagon among those with the lowest number of slaves.
Among the sensors which can be pushed to the destination, $x$ selects the closest to $Hex(y)$.

If a snapped sensor receives a neighbor discovery
request while involved in a push activity, it
replies as
if the ongoing movements were already concluded.
Indeed, if a snapped sensor communicated its own
number of
slaves without keeping into account the ongoing
movements, it could cause inconsistencies
(for example either too many sensors could move
to the same hexagon or the same sensor could be
offered to several snapped sensors).
The snapped sensors involved in a push activity
always alert their neighborhood of the
changed number of slaves.

\paragraph{Pull activity.}

The sole snap and push activities are not sufficient to
ensure the maximum expansion of the tiling.
This may happen when there exists a direction in
which the density decreases of at most one
sensor at a time, and the Moving Condition is false due to the order relationship induced by function
$id(\cdot)$.
The same problem may cause also non-uniform coverage.
For this reason, we introduce the pull activity
that makes use of a {\em trigger mechanism} when
some holes occur.
Namely, let $x$ be a snapped node detecting a hole in an adjacent hexagon, with $S(x)=\emptyset$.
If $x$ has not the possibility to receive any
slave from its neighbor hexagons, i.e. the
Moving Condition is not verified for any of them,
then it activates the following trigger mechanism.
Sensor $x$ temporarily alters the value of its $id$ function
to 0 and notifies its neighbors of this change by
means of a {\em trigger notification message}.
This could be sufficient to make the Moving
Condition true with at least a snapped neighbor,
so $x$ waits until either a new slave comes into
its hexagon or a timeout occurs.
If a new slave enters  $Hex(x)$, $x$ sets back
its  $id$ value and snaps the new sensor,
thus filling the hole.
If the timeout expires and the hole has not been
covered yet, $x$
extends the trigger to its adjacent hexagons by sending them a
a {\em trigger extension message}. As a consequence, the snapped neighbors of $x$ set
their  $id$ value to 1 and send the
related trigger notification message.
This mechanism is iterated by $x$ over snapped
sensors at larger and larger distance in the tiling
until the hole is covered.
Each snapped sensor
involved in the trigger extension mechanism sets
its  $id$ to a value that is
proportional to the distance from $x$. All the
timeouts related to each new extension are set
proportionally to the maximum distance reached by
the trigger mechanism.
At this point, as a consequence of timeouts, each
involved node sets back its  $id$ to the
original value.
In order to better detail the trigger mechanism, we show the following example.
Figure \ref{fig:pull_action}(a) shows a tiled AoI with a coverage hole in the bottom left corner.
Snapped nodes detecting the hole set their $id$ to 0 and send a trigger notification message.
As their neighbors do not have slaves, they need to send a trigger extension message,
provoking a propagation of the $id$ modification  (see Figure \ref{fig:pull_action}(b)).
As soon as the unique snapped sensor with a slave alters its $id$ to honor the trigger mechanism,
the Moving Condition is satisfied and therefore the slave is pushed towards a snapped sensor
that is closer  to the hole, as shown in Figure \ref{fig:pull_action}(c).
In Figure \ref{fig:pull_action}(d) the hole coverage is highlighted and, after the timeouts expire,
all $id$s are set back to the previous values (Figure \ref{fig:pull_action}(e)).

It should be noted that more snapped nodes adjacent to the
same hole may independently activate the trigger
mechanism, possibly at different times.
In this case, if a node receives a trigger
extension message from two or more nodes, it
honors only the one with the lowest  $id$.
The detection of several holes may cause the same
node $y$ to receive several trigger
extension messages. These are stored in a
pre-emptive priority queue, privileging
the messages related to the closest hole.

\paragraph{Tiling merge activity.}

If several sensors act as starters, several tiling portions  can be generated with different orientations.
By contrast, our algorithm aims to cover the AoI with
a perfectly regular tiling thus minimizing overlaps of the sensing disks and
enabling a complete and uniform coverage.
Hence, we design a merge mechanism according to
which as soon as a sensor $x$ receives a neighbor
discovery message from another  tiling portion it joins the
oldest one (it discriminates this situation by evaluating the
time-stamp of the starter action).
It should be noted that the detection of the sole neighbor
discovery messages is sufficient to ignite the
tiling merge activity because such messages are
sent after any tiling expansion and, if two tiling
portions come in radio proximity, at least one of
them is increasing its extension.
In the following, we call $G_\texttt{old}$ and
$G_\texttt{new}$ the tiling portions with lower and
higher time-stamp, respectively.
We distinguish three possible cases:\\
1) $x$ belongs to $G_\texttt{new}$:
if $x$ is a slave, sensor $x$ switches its state
to free and communicates its new state to the
neighborhood. From now on $x$ will honor only the
messages coming from $G_\texttt{old}$ and will ignore
those from $G_\texttt{new}$.
This proactive communication is needed to
advertise the presence of $G_\texttt{new}$ when
there is no message exchange within
$G_\texttt{new}$ perceivable by the sensors in
$G_\texttt{old}$.
This way, the snapped sensor to which $x$ belonged,
can properly update its slave set.
If $x$ is instead a snapped sensor, it cannot
immediately switch its state to free because
of its leading role inside $G_\texttt{new}$ (e.g.
it leads the slave sensors in $S(x)$ and performs
push and pull activities).
Hence, $x$ temporarily assumes a hybrid role: it declares itself
as free to the nodes of $G_\texttt{old}$ and, at
the same time, acts as a
snapped sensor in $G_\texttt{new}$ until it
receives a snap command coming from
$G_\texttt{old}$.
After the reception of such a snap command, $x$
moves to the new snap position and elects one of
its slave as a substitute. If no slave is
available, $x$ advertises  its departure to its neighbors in
$G_\texttt{new}$.\\
2) $x$ belongs to $G_\texttt{old}$:
if $x$ is a slave, it ignores all messages from $G_\texttt{new}$.
If $x$ is snapped, it performs a neighbor discovery,
ignores all messages coming from $G_\texttt{new}$ (apart from the neighbor discovery replies) and honors only messages
from $G_\texttt{old}$.
Observe that the neighbor discovery is necessary to ignite the merge mechanism.
The neighbor discovery allows each snapped sensor
in  $G_\texttt{old}$ to collect complete
information about nearby sensors that previously
belonged to  $G_\texttt{new}$.\\
3) $x$ is free:
the sensor $x$ honors only messages from
$G_\texttt{old}$ and ignores those from
$G_\texttt{new}$.

\subsection{Balancing energy consumption}

According to the previous description of {\HC},
slaves consume more energy than snapped sensors,
because they are
involved in a larger number of message exchanges and movements.
We introduce a {\em mechanism to balance the energy
consumption} over the set of available sensors
making them exchange
their roles. This mechanism is similar to the technique of
{\em cascaded movements} introduced in \cite{LaPorta_Relocation}.
Namely, any time a slave has to make a movement across a
hexagon as a consequence of either push or pull
activities, it evaluates the opportunity to
substitute itself
with the snapped sensor of the hexagon it is traversing.
The criterion at the basis of this mechanism is that
two sensors exchange their role whenever the energy imbalance is reduced.
As a result, the energy balance is significantly enhanced, though
the role exchange has a small cost for both the slave
and the snapped sensor involved in the
substitution.
Indeed, the slave sensor has to reach the center of the current
hexagon and perform a {\em profile packet} exchange
with the snapped sensor that has to move towards
the destination of the slave. A profile packet
contains the key information needed by a sensor
to perform
its new role after a substitution.

\subsection{The sensor coordination protocol}
The implementation of our algorithm requires the definition of a protocol for
the coordination of activities among locally communicating sensors.
The coordination protocol provides the rules to
solve
contentions that may happen in several cases.
 For
example, two or more snapped sensors can decide
to issue
a snap command to more than one sensor towards
the same hexagon tile or a low density hexagon
can be selected
by several snapped sensors as candidate for receiving redundant slaves.
These contentions are solved by properly
scheduling actions according to message
time-stamps and by advertising
related decisions as soon as they are made.
This  protocol is designed to minimize energy
consumption in terms of  message
exchanges,
which is possible because the algorithm decisions
are only based on a small amount of local
information.
Furthermore, we assume that the  protocol of \HC\ is implemented over a
communication protocol stack which handles possible errors and losses that may occur
on the radio channels by means of timeout and retransmission mechanisms.
We do not give any
further detail on the protocol underlying {\HC} as it is beyond the scope of this paper. The interested
reader can refer to \cite{ICNP2008}.

\begin{figure}[h]
\centering
\scalebox{0.20}{
\includegraphics[]{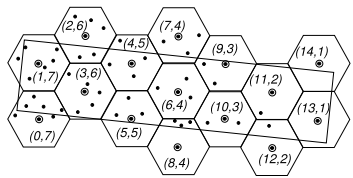}}
\caption{Local formation of a stairwise density distribution.} \label{fig:scaletta}
\end{figure}

\section{A discussion on uniformity implications} \label{sec:discussion}

The execution of \HC\ guarantees coverage uniformity only if the number of available sensors is
exactly the minimum for the given orientation of the final grid.
If redundant sensors are available, their movements are regulated by the moving condition, that
precludes the flow of redundant sensors from high density to low density hexagons if the difference
between the local densities is only of one sensor.
This may cause local formation of a stairwise density distribution when the order function is monotonically
increasing in the precluded flow direction. An example of such situations is depicted in Figure \ref{fig:scaletta}, where for each
hexagon the elements of the pairs represent the order value and the number of sensors, respectively.

The length of these formations is usually very short due to the random distribution of the order value over the set of sensors.
The worst case, albeit improbable, happens when a stairwise distribution is as long as the diameter of the AoI.

In order to guarantee the uniformity of the sensor deployment even in the presence of redundant sensors,
we introduce a {\em shrinked grid mode} as a variant of the \HC\ algorithm. From now on,
we will refer to the basic version with the name PP1 and to the shrinked grid mode with the name PP2.
In Section \ref{sec:properties}
we prove that PP2  enables a uniformly redundant coverage, and we provide metrics and related formulas
to calculate the guaranteed redundancy level.

In order to formally describe such mode we introduce
the following definition: the {\em tight number of sensors} is
the maximum number of
hexagons of side length $l_\texttt{h}$ necessary to cover the AoI for each
possible initial position of the sensor set
and each possible tiling orientation.
This number represents the maximum number of sensors that can be necessary
to cover the AoI with a hexagonal tiling of side $l_\texttt{h}$, regardless of the position and orientation of the grid.

We denote this number by $N_\texttt{tight}(l_\texttt{h},\textrm{AoI})$, for short
$N_\texttt{tight}(l_\texttt{h})$, as the AoI is clear from the context.
An upper bound on this number can be
calculated by increasing the AoI
with a border whose width is the maximal diameter
of the tiling hexagon that is $2l_\texttt{h}$ and dividing the
area of such a region (call it AoI'$(l_\texttt{h})$) by the hexagon area.
Formally:
\begin{equation}
N_\texttt{tight}(l_\texttt{h}) \leq \left\lceil
\frac{\textrm{Area(AoI'}(l_\texttt{h}) \textrm{)}}{(3 \sqrt{3}/2)l_\texttt{h}^2}
\right\rceil
\label{eq:n_tight}
\end{equation}

It should be noted that this upper bound is valid in the general case but its calculation  can be improved
if the AoI has a particularly regular shape.

PP2 is executed with a shorter hexagon
side length. Namely, $l_\texttt{h}$ is set to reduce as much as possible the number of slave sensors in the
whole deployment, and is calculated as the value that makes the number of sensors equal to the tight
number for that side length, and therefore is the inverse function of $N_\texttt{tight}(\cdot)$, calculated in $N$, where $N$ is the number of sensors.
More formally, $l_\texttt{h}=N_\texttt{tight}^{-1}(N)$.
Since function $N_\texttt{tight}(\cdot)$ is not known, we calculate an upper bound on $l_\texttt{h}$ as the inverse of the upper bound on $N_\texttt{tight}(l_\texttt{h})$, because $N_\texttt{tight}(\cdot)$ is
a decreasing function of $l_\texttt{h}$.

PP1 and PP2 produce sensor deployments
with different performance in terms of energy consumption and fault tolerance.
The choice between them depends on the particular application requirements, as discussed below here.

In terms of energy consumption, PP2 performs
worse than PP1, as we will show and motivate in
Section \ref{sec:exp_results}. Nevertheless, as it guarantees a
uniformly redundant coverage, it makes possible the use of topology
control  algorithms \cite{Poduri2007} that permit selective sensor
activation saving energy during the operative phase, which, in turn, follows
the deployment phase of the network. Moreover, this mode is beneficial
when the application requires enhanced environmental
monitoring  and strong fault-tolerance capabilities. From the fault
tolerance point of view, PP1 may be endowed with a
periodic polling scheme to detect new possible coverage holes
determined by sensor failures. This way, the detection of new holes
causes the restart of the algorithm and the execution of the pull
activity that attracts redundant sensors possibly located far from
the coverage hole.
Hence, PP1 presents
self-healing properties which are not found in previous solutions.
 An example of such a mechanism is shown in Figure
\ref{fig:bomba_pp1}. Figure \ref{fig:bomba_pp1} (a) shows the deployment achieved
after the application of PP1. Figure \ref{fig:bomba_pp1} (b) shows a subsequent situation in which
a certain number of sensors failed, creating a coverage hole.
The presence of such a coverage hole is detected by the nearby sensors, which give start to the pull activity,
attracting some redundant sensors located in higher density areas. Figure \ref{fig:bomba_pp1} (c) shows an
intermediate situation, before the redundant sensors succeed in covering the hole, as shown in Figure \ref{fig:bomba_pp1} (d).

Instead, PP2 can tolerate several failures, even closely located, in a number which is proportional to the redundancy level. By contrast, PP2 is not able to fill newly detected holes, because (almost) all sensors
are snapped and do not take part in the movements determined by the pull activity.
For this reason we do not introduce any polling mechanism in PP2, as there are too few slaves available and it would produce
an inefficient pull activity in the case of a hole detection.
On the other hand such version guarantees a uniform redundant coverage, if a sufficient number of sensors are
available, tolerating even numerous sensor failures, as shown in Figure \ref{fig:fault_pp2}, which depicts the
occurrence of several co-located sensor failures without any loss of coverage.

\begin{figure}
\centering
\vspace{1cm} \subfigure[]{\scalebox{0.027}{ \centering
\includegraphics[]{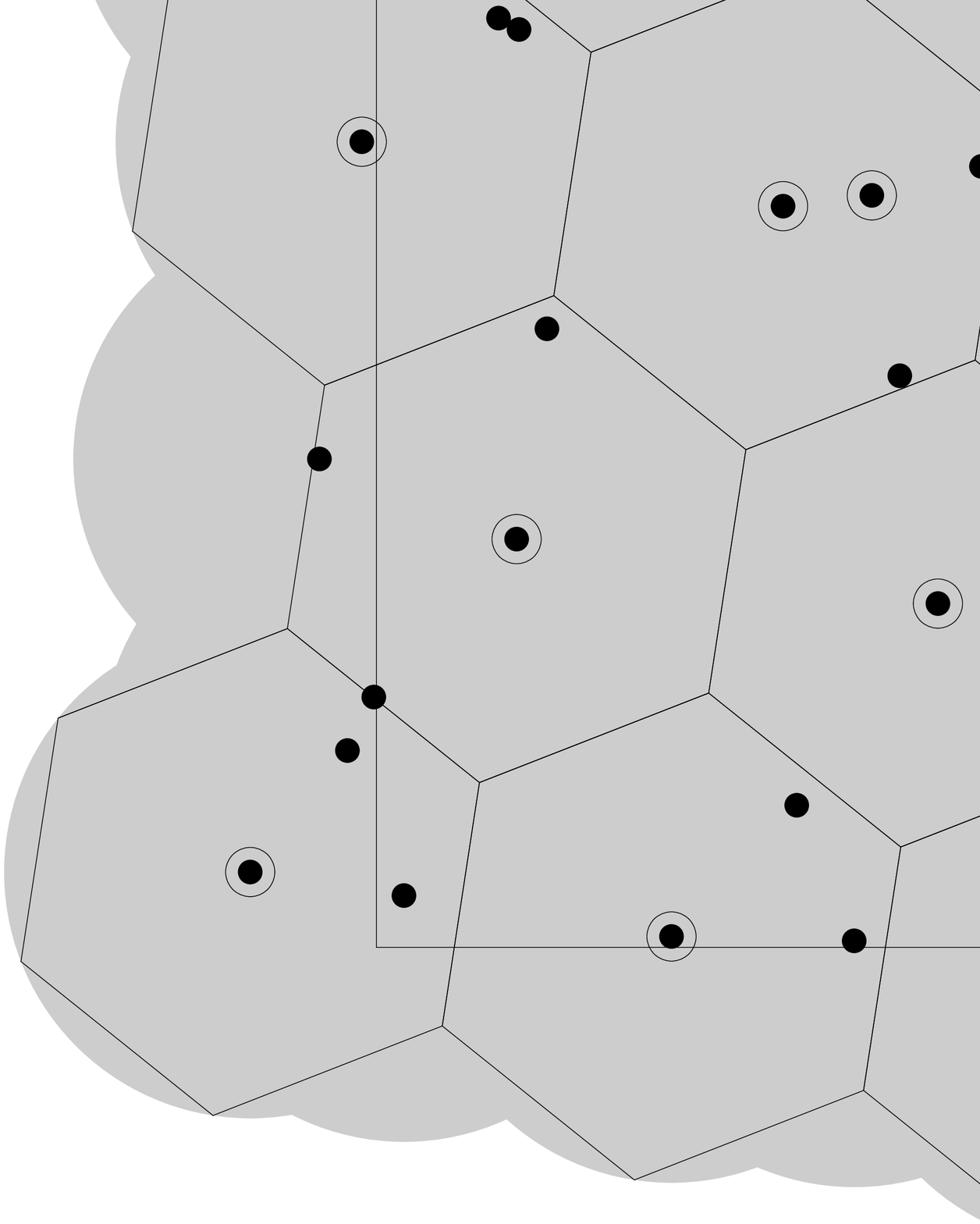}}}
\hspace{2cm} \subfigure[]{\scalebox{0.027}{ \centering
\includegraphics[]{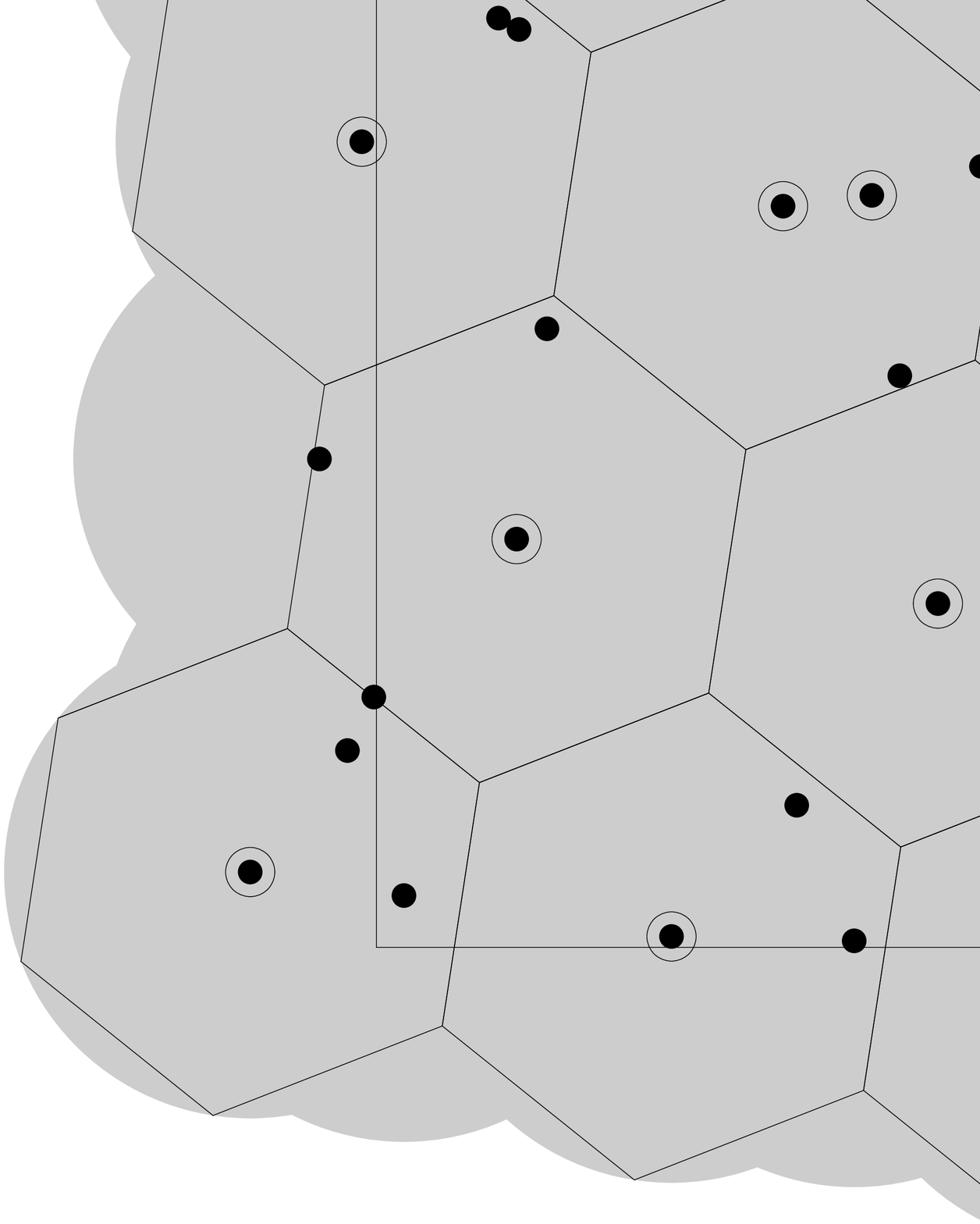}}}
\hspace{2cm} \subfigure[]{\scalebox{0.027}{\centering
\includegraphics[]{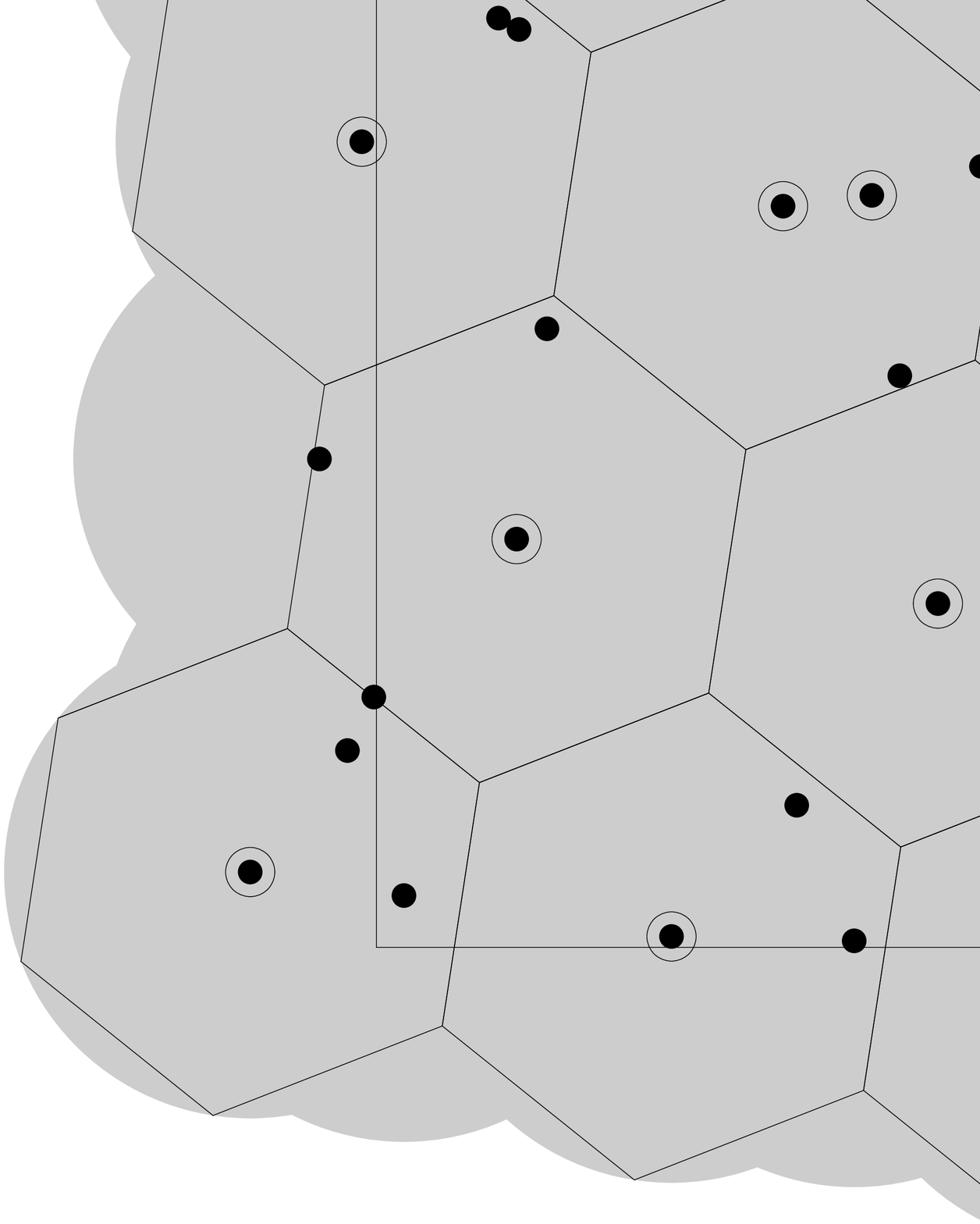}}}
\hspace{2cm} \subfigure[]{\scalebox{0.027}{\centering
\includegraphics[]{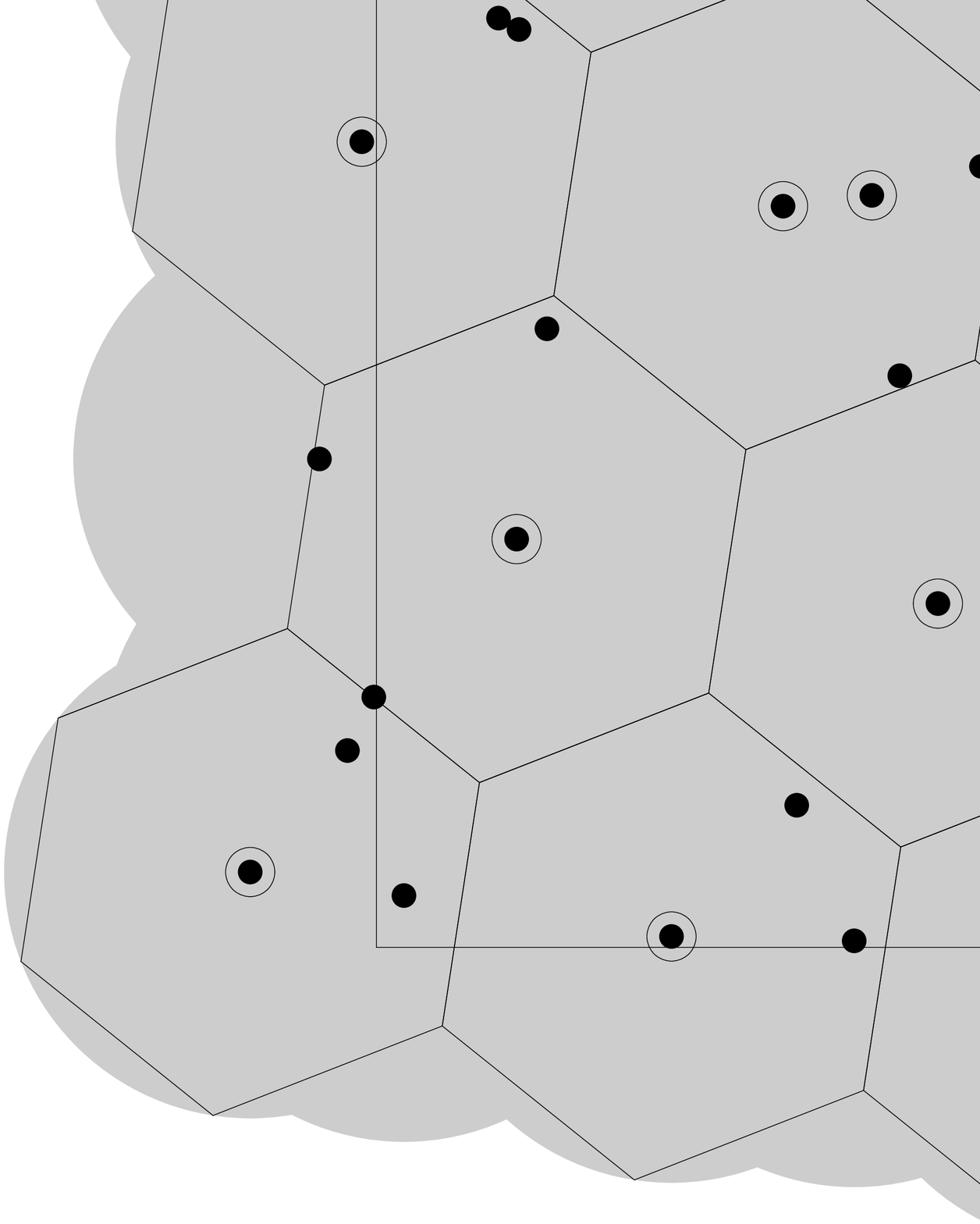}}}

\caption{Self-healing capability of  PP1: (a) sensor deployment after the execution of PP1; (b) failure of  several closely located sensors; (c) an intermediate step in the execution of the pull activity; (d) the coverage hole is filled.} \label{fig:bomba_pp1}
\end{figure}

\begin{figure}
\vspace{1cm}
\begin{center}
\scalebox{0.035}{
\includegraphics[]{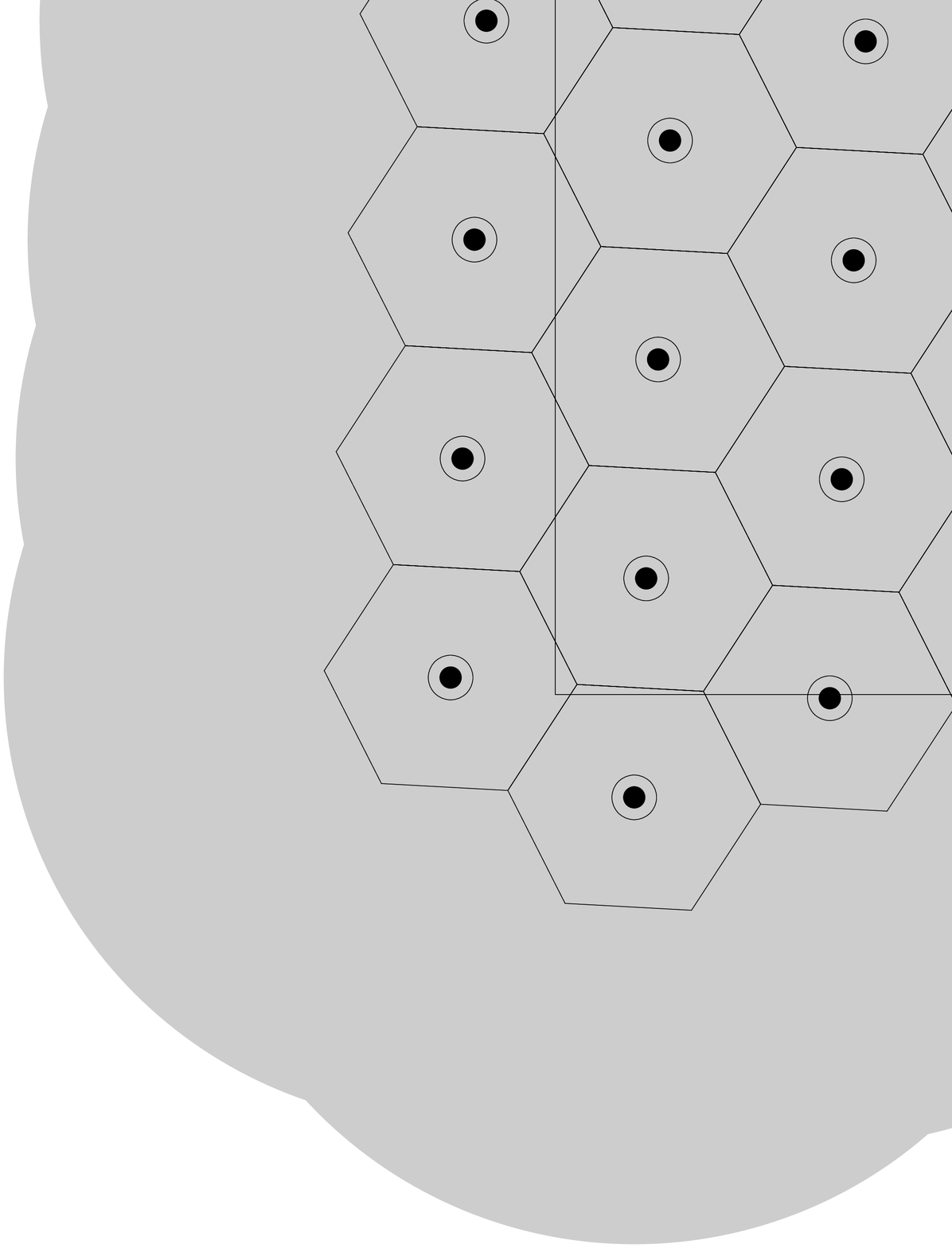}
}
\end{center}
\caption{Fault tolerance of PP2 to the failure of several closely located sensors} \label{fig:fault_pp2}
\end{figure}


\section{Algorithm properties} \label{sec:properties}
In this section we prove some key properties of the
{\HC}\ algorithm: coverage, connectivity
and termination.

\subsection{Coverage Completeness}\label{sec:completeness}

In this subsection we prove that \HC\ (both modes) guarantees the complete coverage.

\begin{theorem}\label{th:coverage}
Algorithm {\HC}\ guarantees the complete coverage,
provided that at least the tight number of sensors $N_\texttt{tight}(l_\texttt{h})$
are available.
\end{theorem}

\begin{proof}
Let us assume that a coverage hole exists.
As our algorithm is designed, this hole will
eventually be detected by a sensor $x$.
Furthermore, by the hypothesis on the number of sensors, it certainly exists a
hexagon with at least one redundant slave.
Let us call $C_x$ the connected component containing sensor $x$.
Two different cases may occur depending on the position of
  the redundant slaves with respect to $C_x$.

\noindent 1) A redundant slave exists in $C_x$: the
snapped sensor $x$ starts the trigger mechanism
that eventually reaches a redundant slave so that
it is pushed towards $x$ and consequently it fills the hole.

\noindent 2) All redundant slaves are located in
connected components different from $C_x$:
the area surrounding each connected component
is in fact a coverage hole that will eventually
be detected by a snapped node located at the
boundary. According to what stated for the case 1), all the separated connected components containing
redundant slaves will expand themselves to fill
as many coverage holes as possible.
Since, by hypothesis, the number of sensors is at
least $N_\texttt{tight}(l_\texttt{h})$, it certainly exists a
component containing redundant slaves that
will eventually merge in $C_x$, leading to the
situation described in the case 1), thus
proving the theorem.
\end{proof}

Notice that, having $N_\texttt{tight}(l_\texttt{h})$ sensors is a sufficient condition to
guarantee the coverage completeness, but this number is not also necessary.
Indeed, $N_\texttt{tight}(l_\texttt{h})$ is calculated as the maximum among all the minimum numbers of
sensors necessary to cover the AoI, for each orientation of the final grid with side length $l_\texttt{h}$. So it is
possible that $N_\texttt{tight}(l_\texttt{h})$ is larger than the number of sensor strictly necessary
for a fixed orientation and position of the oldest starter.

\subsection{Coverage uniformity}

We consider two different coverage redundancy metrics.
The first metric evaluates the coverage only in correspondence to the hexagonal grid points.
This metric, named {\em grid coverage level}, is of interest for the applications that do not require a continuous sensing of the
area of interest but rely on interpolation of local measurements.
On the contrary, the second metric, named {\em continuous coverage level}, is more restrictive and is introduced to evaluate the coverage redundancy at each point
of the area of interest.
\begin{definition}
The {\em grid coverage level} is the minimum number of sensors covering each point of a regular
 grid.
\end{definition}

\begin{definition}
The {\em continuous coverage level} is the minimum number of sensors covering any point of the area of interest.
\end{definition}

In order to compute such metrics for PP1 and PP2,
we introduce the following lemma.

\begin{lemma}
Given a triangular lattice of side $l_{\texttt{h}}$, any circle of radius $R$ and centered in a point of the lattice contains
{\small$$
 n(R)=\sum_{i=-\left\lfloor \frac{R}{\sqrt{3}l_\texttt{h}} \right\rfloor}^{ \left\lfloor \frac{R}{\sqrt{3}l_\texttt{h}} \right\rfloor } \left( 1+2 \left\lfloor
\frac{\sqrt{
R^2-3l_\texttt{h}^2 i^2}}{3l_\texttt{h}}    \right\rfloor     \right)   +  4\sum_{i=0}^{ \left\lfloor \frac{R}{\sqrt{3}l_\texttt{h}}-\frac{1}{2} \right\rfloor }
\left(1+\left\lfloor
\frac{\sqrt{R^2 - 3 l_\texttt{h}^2 \left( i+ \frac{1}{2}  \right)^2 } }    {3l_\texttt{h}} - \frac{1}{2}
\right\rfloor \right)
$$}
points of the lattice.
\end{lemma}

\begin{proof}
We observe that the points inside the circle of radius $R$ and centered in any point of the lattice always lie
in the same position with respect to the center of the circle (see Figure \ref{fig:grid_calculus}), then we can slightly modify the reasoning for the well known Gauss' circle problem, dealing with squared grids.

\begin{figure}
\begin{center}
\scalebox{0.18}{
\includegraphics[]{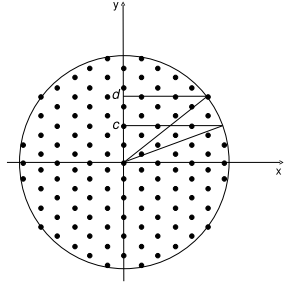}
}
\end{center}
\caption{Calculus of the grid coverage level.} \label{fig:grid_calculus}
\end{figure}

Let the center of the circle be the origin of a Cartesian plane with axis aligned with the grid.

We count the points inside the circle considering them as arranged by horizontal rows.

The number of points in interval $( 0, R ]$ of the $x$ axis is
$\left\lfloor \frac{R}{3l_{\texttt{h}}}  \right\rfloor$
and similarly in interval $[-R, 0)$.
So, counting the origin, there are $1+2\left\lfloor \frac{R}{3l_{\texttt{h}}}  \right\rfloor$
points in interval $[-R, R]$.

Now we count the number of sensors lying on the rows having a sensor on the $y$ axis.

Let us consider one of these rows lying on the line $y=c$, it contains $1+2 \left\lfloor \frac{\sqrt{R^2-c^2}}{3 l_{\texttt{h}}} \right\rfloor$ points.
As two such consecutive rows in the same semiplane are $\sqrt{3} l_{\texttt{h}}$ far from each other,
it follows that the whole number of sensors on all the rows having a sensor on the $y$ axis and lying on the positive semiplane is
$$
\sum_{i=1}^{ \left\lfloor \frac{R}{\sqrt{3}l_\texttt{h}} \right\rfloor } \left( 1+2 \left\lfloor
\frac{\sqrt{
R^2-3l_\texttt{h}^2 i^2}}{3l_\texttt{h}}    \right\rfloor     \right).
$$

Finally, we count the number of sensor lying on the rows not having a sensor on the $y$ axis.

Let us consider one such row lying on the line $y=d$; the sensor closest to the $y$ axis has $x$-coordinate $\frac{3}{2} l_{\texttt{h}}$, so we consider the interval $\sqrt{R^2-d^2}-\frac{3}{2}l_{\texttt{h}}$ long.
Hence, the number of sensors on this row is $2 \left( 1+ \left\lfloor \frac{\sqrt{R^2-d^2}-\frac{3}{2}l_{\texttt{h}}}{3 l_{\texttt{h}}} \right\rfloor \right)$.
With arguments similar to the previous case, we have that the number of sensors lying on these rows in the positive semiplane is:
$$\sum_{i=0}^{ \left\lfloor \frac{R}{\sqrt{3}l_\texttt{h}}-\frac{1}{2} \right\rfloor }
2\left(1+\left\lfloor
\frac{\sqrt{R^2 - 3 l_\texttt{h}^2 \left( i+ \frac{1}{2}  \right)^2 } - \frac{3 l_\texttt{h}}{2}  }    {3l_\texttt{h}}
\right\rfloor \right).$$

The result follows by summing all the described contributions.
\end{proof}

\begin{theorem}\label{th:uniformity}
Under the assumption that at least a tight number of sensors
are available, and the shrinked grid mode is enabled, algorithm {\HC}\
guarantees a $k$ grid coverage level,
where $k=n(R_\texttt{s})$.

\end{theorem}

\begin{proof}
The definition of $l_\texttt{h}$ and lemma \ref{th:coverage}
imply that, under the given assumptions, algorithm \HC\ provides a complete coverage.
Given the geometric regularity of the obtained deployment, every sensing circle
surrounding a snapped sensor contains at least a fixed number $k$ of snapped sensors
belonging to the triangular lattice determined by the hexagonal grid deployment.
As all sensors have the same sensing radius, the sensing redundancy level at the center of the circle is at least $k$.

\end{proof}

\begin{figure}[t]
\centering
\scalebox{0.18}{
\includegraphics[]{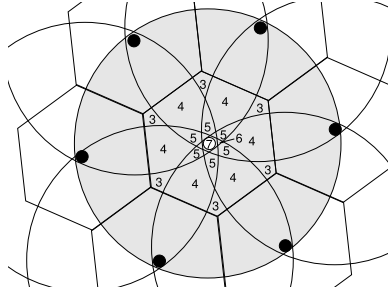}}
\caption{An example of continuous coverage levels} \label{fig:coverage}
\end{figure}

In order to estimate the continuous coverage level of any sensor deployment, in \cite{Huang2005} the authors
introduce the notion of perimeter coverage. They define a sensor $s$ to be \emph{$k$-perimeter covered} if all
points in the perimeter of the sensing circle of $s$ are covered by at least $k$ sensors (not counting $s$).

The same authors also prove (see Theorem 1 in \cite{Huang2005}) that the sensor deployment provides a continuous
coverage level $k$ if and only if each sensor is $k$-perimeter covered.

\begin{theorem}\label{th:uniformity_continuous}
Under the assumption that at least a tight number of sensors
is available, and the shrinked grid mode is enabled, algorithm {\HC}\
guarantees a $k$ continuous coverage level,
where $k\geq \frac{n(R_\texttt{s})-1}{3} +
\frac{n(\sqrt{3}R_\texttt{s})-n(R_\texttt{s})}{6}$.

\end{theorem}

\begin{proof}
According to the above cited theorem \cite{Huang2005}, the level of continuous coverage enabled by
the algorithm \HC\ can be calculated as the minimum perimeter coverage over all the snapped sensors.
In order to calculate such coverage level, we distinguish two main contributions, the first one
coming from sensors located inside the sensing circle of $s$, and the second one, coming from
the sensors located outside.

All sensors located inside the sensing circle of $s$
contribute to the
perimeter coverage with a circular sector of amplitude $\alpha$, with $\frac{2}{3}\pi\leq \alpha < \pi $.
Since any of these sensors is symmetric to other five sensors inside the circle, with a rotation of $\pi/3$ centered in the position of sensor $s$,
all the six of them contribute to at least a double coverage of the sensing circle perimeter of $s$.
The sensors forming this first contribution amount to $n(R_\texttt{s})-1$ (not counting the sensor $s$ itself), and
all of them globally guarantee $2\left\lfloor \frac{n(R_\texttt{s})-1}{6}\right\rfloor$-perimeter coverage.

The second contribution is related to the sensors located outside the sensing
circle of $s$. We note that the sensing circle of the sensors located farther than $2R_\texttt{s}$ from $s$ do not intersect the sensing circle of $s$, while the sensing circle of sensors located at a distance $d$ such that $\sqrt{3}R_\texttt{s}< d \leq 2R_\texttt{s}$ intersect the sensing circle of $s$, determining a circular arc of amplitude less than $\pi/3$.
Since we are calculating a lower bound on the minimum perimeter coverage, we do not consider the contribution of this latter sensors as it does not guarantee
a complete perimeter coverage and therefore may not affect its minimum value.

For this reason, as a second contribution to the perimeter coverage, we only consider the
sensors located inside the circular crown determined by the radii $R_\texttt{s}$ and $\sqrt{3} R_\texttt{s}$.
This sensors
contribute to the
perimeter coverage with a circular sector of amplitude $\beta$, with $\pi/3 \leq \beta < \frac{2}{3}\pi $.
Since any of these sensors is symmetric to other five sensors inside the crown, with a rotation of $\pi/3$ centered in the position of sensor $s$,
all the six of them contribute to at least one single coverage of the sensing circle perimeter of $s$.
The sensors forming this second contribution amount to $n(\sqrt{3}R_\texttt{s})-n(R_\texttt{s})$
and all of them globally guarantee a $\left\lfloor \frac{n(\sqrt{3}R_\texttt{s})-n(R_\texttt{s})}{6}\right\rfloor$-perimeter coverage. Notice that the particular 3-axis symmetry, induced by the hexagonal tiling, makes it possible to remove the floor operator from the two terms, as $n(R) - 1$
is always divisible by 6.

By summing the two contribution to the perimeter coverage, we derive the claimed lower bound.

\end{proof}

\subsection{Coverage and connectivity}
\label{sec:lattice}

In this subsection we motivate the choice of the hexagonal tiling and
the assumption that the sensors operate with $R_\texttt{tx} \geq \sqrt{3} R_\texttt{s}$, that is
a less restrictive condition than usually required in the literature.

In \cite{Zhang2005},
the authors demonstrate that coverage implies connectivity if and
only if $R_\texttt{tx}$ is twice the value of $R_\texttt{s}$.
This statement is generally valid regardless of the particular
distribution of the sensors over the AoI, be it a regular geometrical mesh or
a random deployment.
A hexagonal tiling with side length $R_\texttt{s}$
is the one that minimizes node density while ensuring coverage completeness at the same time, as argued in \cite{Brass2007}.
Since our algorithm works exactly with this kind of tiling, which corresponds to a triangular lattice, we can relax the relationship
between $R_\texttt{tx}$ and $R_\texttt{s}$.
If the sensors are regularly deployed on a hexagonal tiling, the distance between any two tiling neighbors
is exactly $\sqrt{3}R_\texttt{s}$, implying the following result.

\begin{theorem}
Under a complete triangular lattice coverage with side length $R_\texttt{s}$,
a necessary and sufficient condition to guarantee connectivity is that
 $R_\texttt{tx}\geq \sqrt{3}R_\texttt{s}$.
\end{theorem}

\subsection{Termination of {\HC}}

We conclude this section by proving the termination of our algorithm.

Let $L= \{ \ell_1, \ell_2,
\ldots, \ell_{|L|} \}$ be the set of snapped sensors.

\begin{definition}\label{def:netState}
A {\em network state} is a vector  $\mathbf{s}$  whose $i$-th component
represents the number of sensors deployed inside
the hexagon $Hex(i)$ governed by the snapped
sensor $i$.
Therefore $\mathbf{s}=<s_1, s_2, \ldots,
s_{|L|}>$ where $s_i=|S(i)|+1$, $\forall i =
1,\dots,|L|$.
\end{definition}

\begin{definition}\label{def:stableState}
A state $\mathbf{s}=<s_1,\dots,s_{|L|}>$ is \emph{stable},
if the Moving Condition
is false for each
couple of snapped sensors in $L$ located in radio
proximity to
each other.
\end{definition}

\begin{theorem}
  Algorithm {\HC}\ terminates in a finite time.
\end{theorem}

\begin{proof}
As long as new sensors are being snapped, the covered area keeps on growing.
This process eventually ends either because the
AoI has been completely covered or because the sensors have reached a
configuration that does not allow any further expansion of the tiling.
Due to  Theorem \ref{th:coverage} the latter can only happen when all sensors are snapped
and thus the state of the network is stable.
In order to prove the theorem, it suffices to prove
that, once the AoI is fully covered, the algorithm reaches a stable
configuration in a finite time.
Therefore we can consider the set of snapped sensors
$L$  as fixed. The value of the order function
related to each snapped sensor, $ id(\ell_i)$, is set during the
unfolding of the algorithm, it can be modified
only temporarily by the pull activity a finite
number of times and remains steady onward.
Let us define $f:\mathbb{N}^{|L|} \rightarrow \mathbb{N}^2$ as follows: \\
$ f(\mathbf{s})= \left(
\sum_{i=1}^{|L|} s_i^2, \sum_{i=1}^{|L|}
s_i \cdot id(\ell_i) \right)$.
We say that $f(\mathbf{s}) \succ
f(\mathbf{s^\prime})$ if $f(\mathbf{s})$ and
$f(\mathbf{s^\prime})$ are in lexicographic order.
Observe that
function $f$ is lower bounded by the
pair \\$(|L|, \sum_{i=1}^{|L|}  id(\ell_i))$, in
fact
$1\leq s_i \leq |V|$.
Therefore,  if we prove that  the value of $f$  decreases at every state change,
we also prove that no infinite sequence of state changes is possible.
To this purpose, let us show that every state
change from $\mathbf{s}$ to $\mathbf{s}^\prime$
causes $f(\mathbf{s})\succ f(\mathbf{s}^\prime)$.
Let us consider a generic state change which
involves the snapped sensors $x$ and $y$, with
$x$ sending a slave sensor to $Hex(y)$.
We have that $s_i=
s'_i \quad \forall i \ne x,y$, and $s'_x=s_x-1$ and $s'_y = s_y +1$.
As the transfer of the slave has been done according to the
  Moving Condition,
two cases are possible:
either $s_x > s_y + 1$, or $(s_x = s_y + 1) \wedge ( id(x) >  id(y))$.
In the first case, the inequality $s_x > s_y + 1$ implies that
$\sum_{i=1}^{|L|} {s}_i^2  >  \sum_{i=1}^{|L|} {s'}_i^2$.
In the second case,  since
$s_x = s_y + 1$ and $ id(x) > id(y)$, lead to $ \sum_{i=1}^{|L|} s_i^2 =
\sum_{i=1}^{|L|} {s'}_i^2 $
and
$\sum_{i=1}^{|L|}  s_i \cdot id(\ell_i)  >
\sum_{i=1}^{|L|}  id(\ell_i) {s'}_i
$. Therefore in both cases   $f(\mathbf{s})\succ f(\mathbf{s}^\prime)$.
The function $f$ is lower bounded and always decreasing by
discrete quantities (integer values) at any state change. Thus, after
a finite number of steps, it is impossible to perform a further
state change, i.e. the network will be in a \emph{stable state} in a
finite time.
\end{proof}

\section{Simulation results}
\label{sec:exp_results}

In order to evaluate the performance of {\HC}\ and to compare it with previous solutions,
we developed a simulator using the
wireless module of the  OPNET modeler software
\cite{opnet}.

\begin{figure*}[h]

\subfigure[]{\scalebox{0.04}{\includegraphics[]{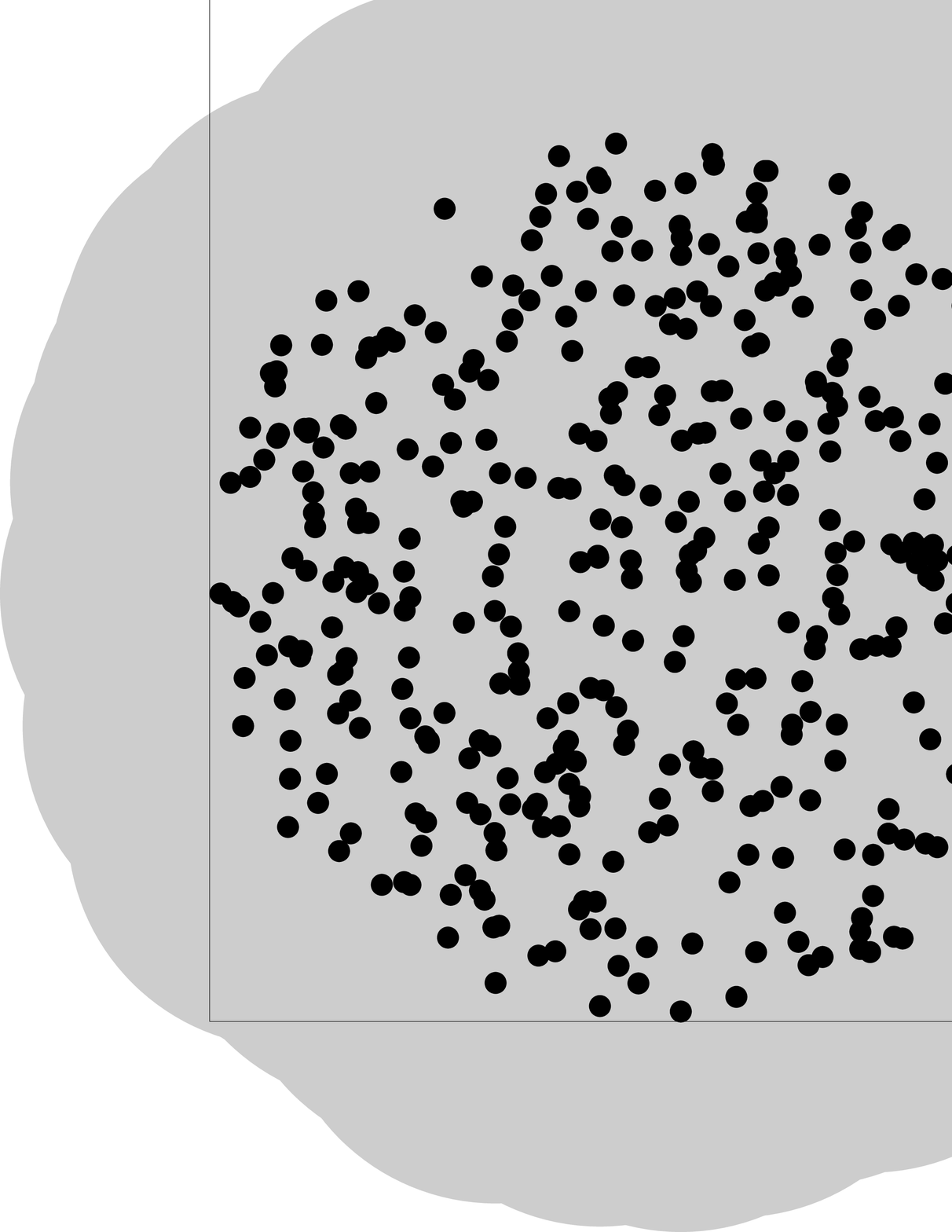}}}
\hspace{5cm}
\subfigure[]{\scalebox{0.04}{\includegraphics[]{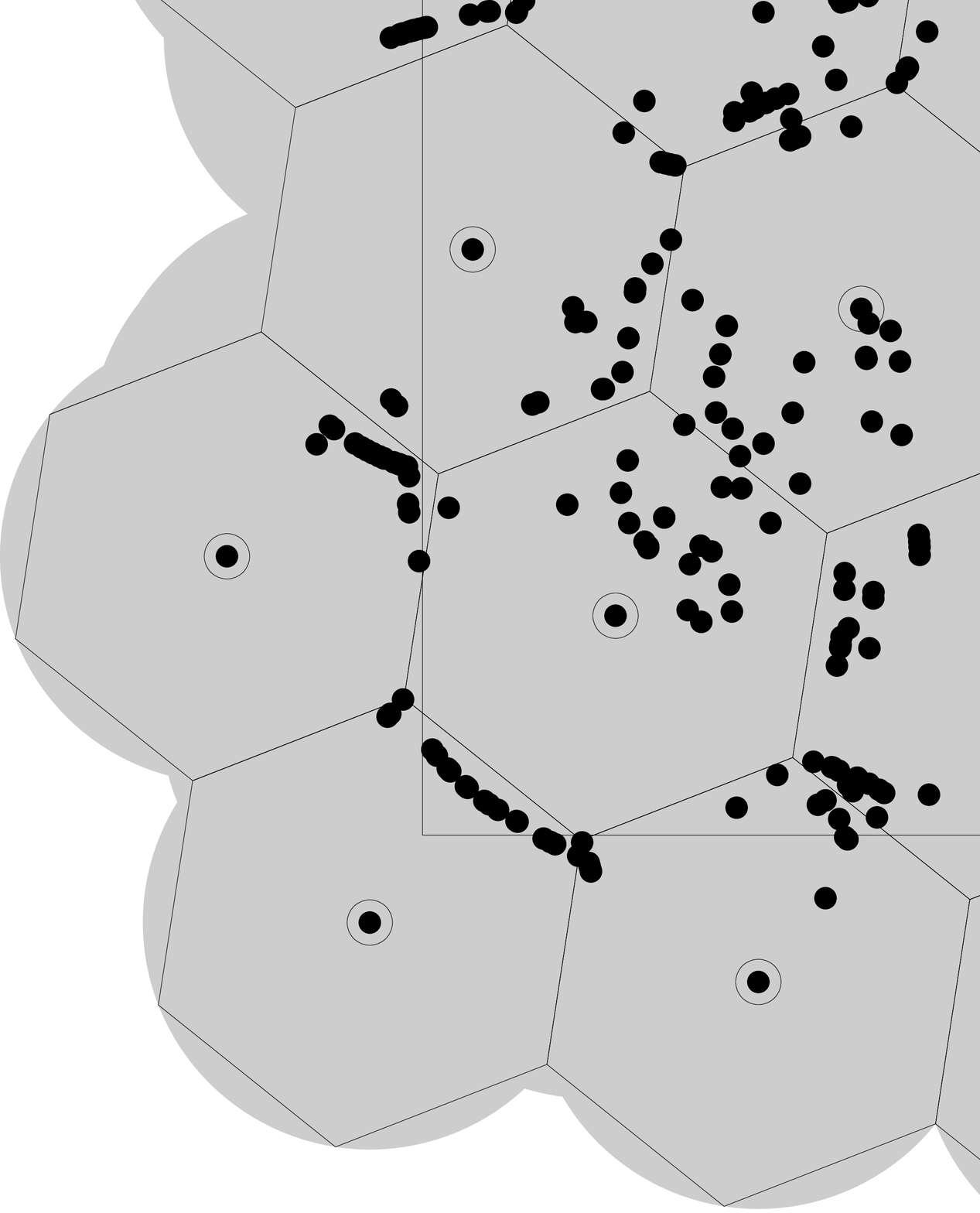}}}
\vspace{1cm}
\\
\subfigure[]{\scalebox{0.04}{\includegraphics[]{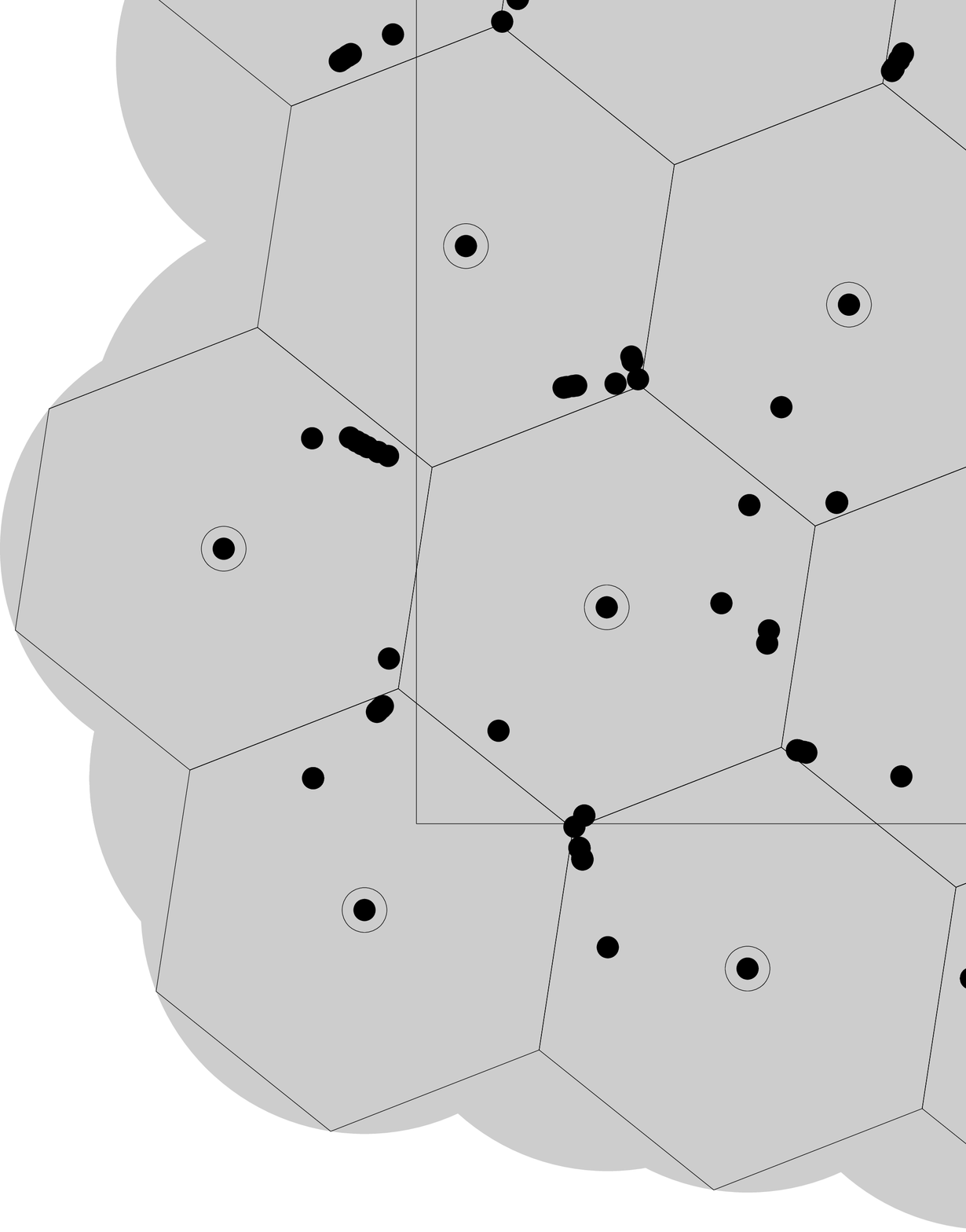}}}
\hspace{5cm}
\subfigure[]{\scalebox{0.04}{\includegraphics[]{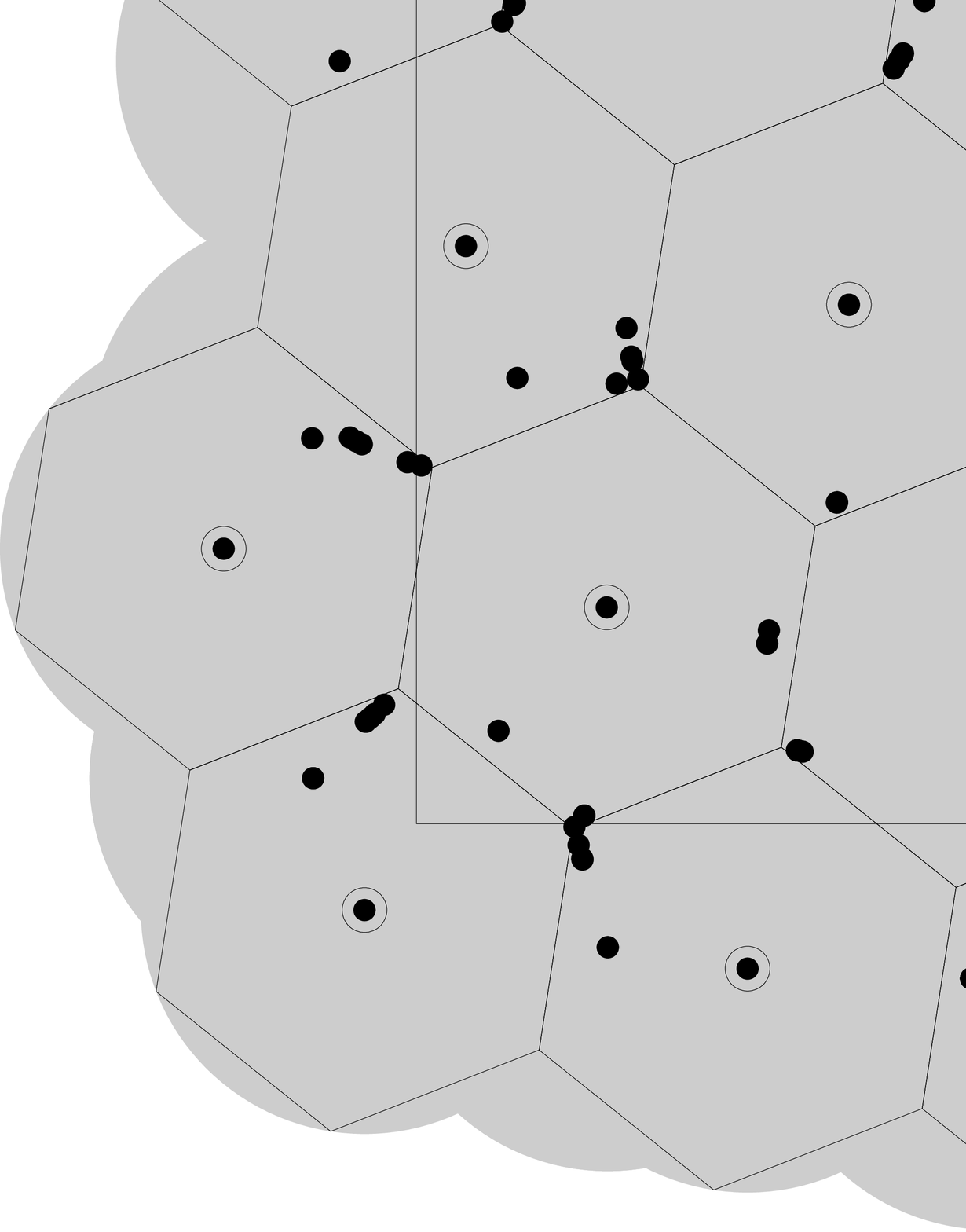}}}
\caption{Coverage of an irregular AoI under
PP1.}\label{fig:grumo_corridoio_pp1}
\end{figure*}

\begin{figure}[h]
\subfigure[]{\scalebox{0.04}{\includegraphics[]{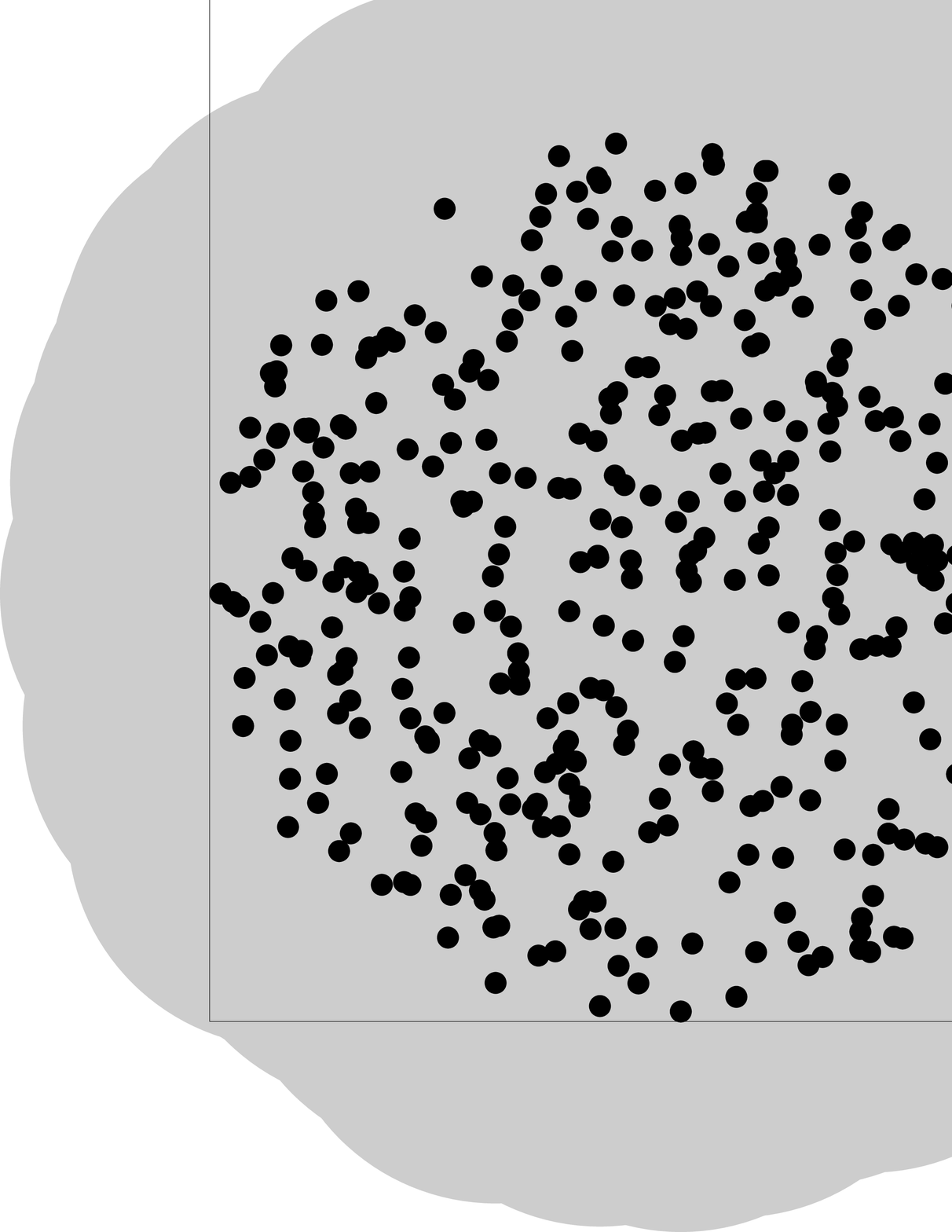}}}
\hspace{5cm}
\subfigure[]{\scalebox{0.04}{\includegraphics[]{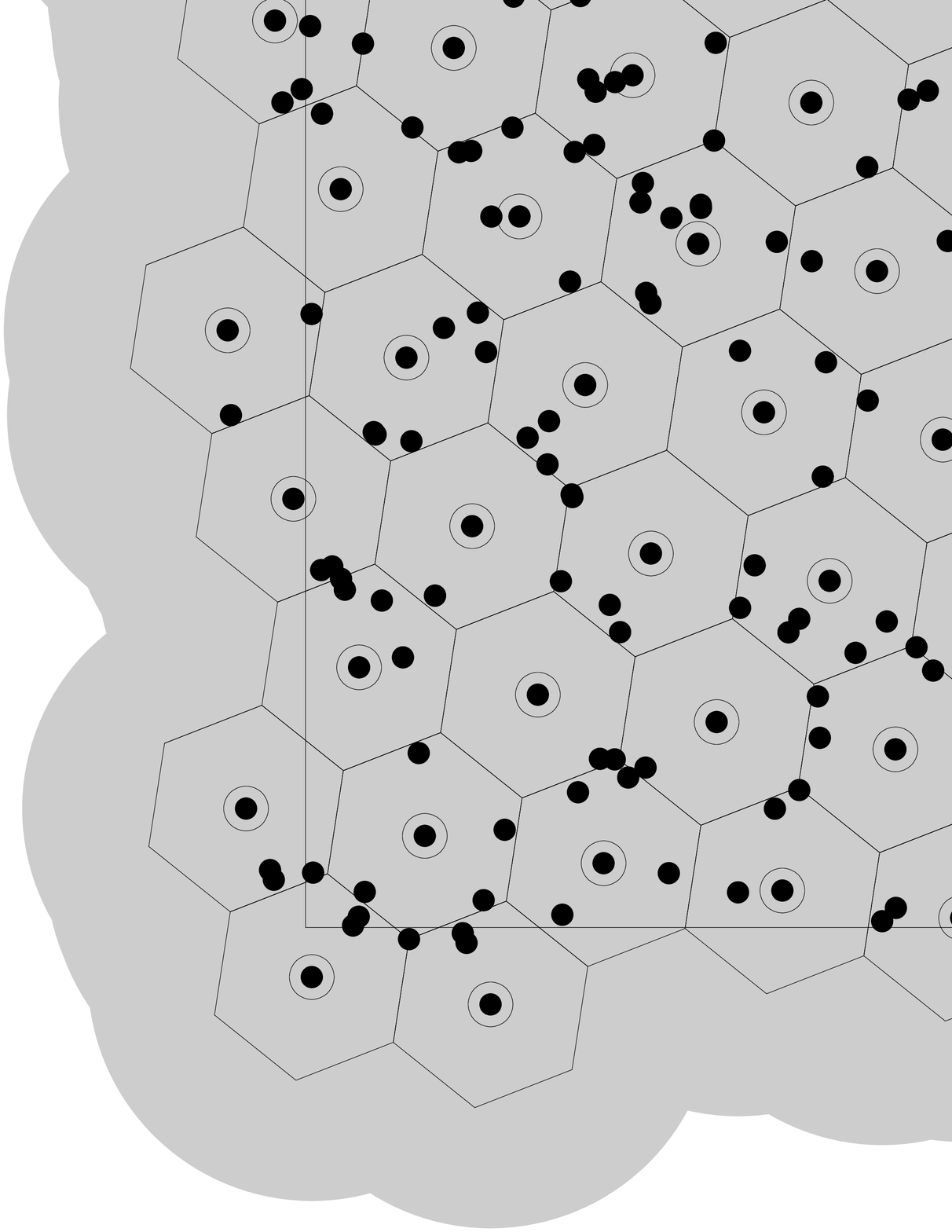}}}
\vspace{1cm}
\\
\subfigure[]{\scalebox{0.04}{\includegraphics[]{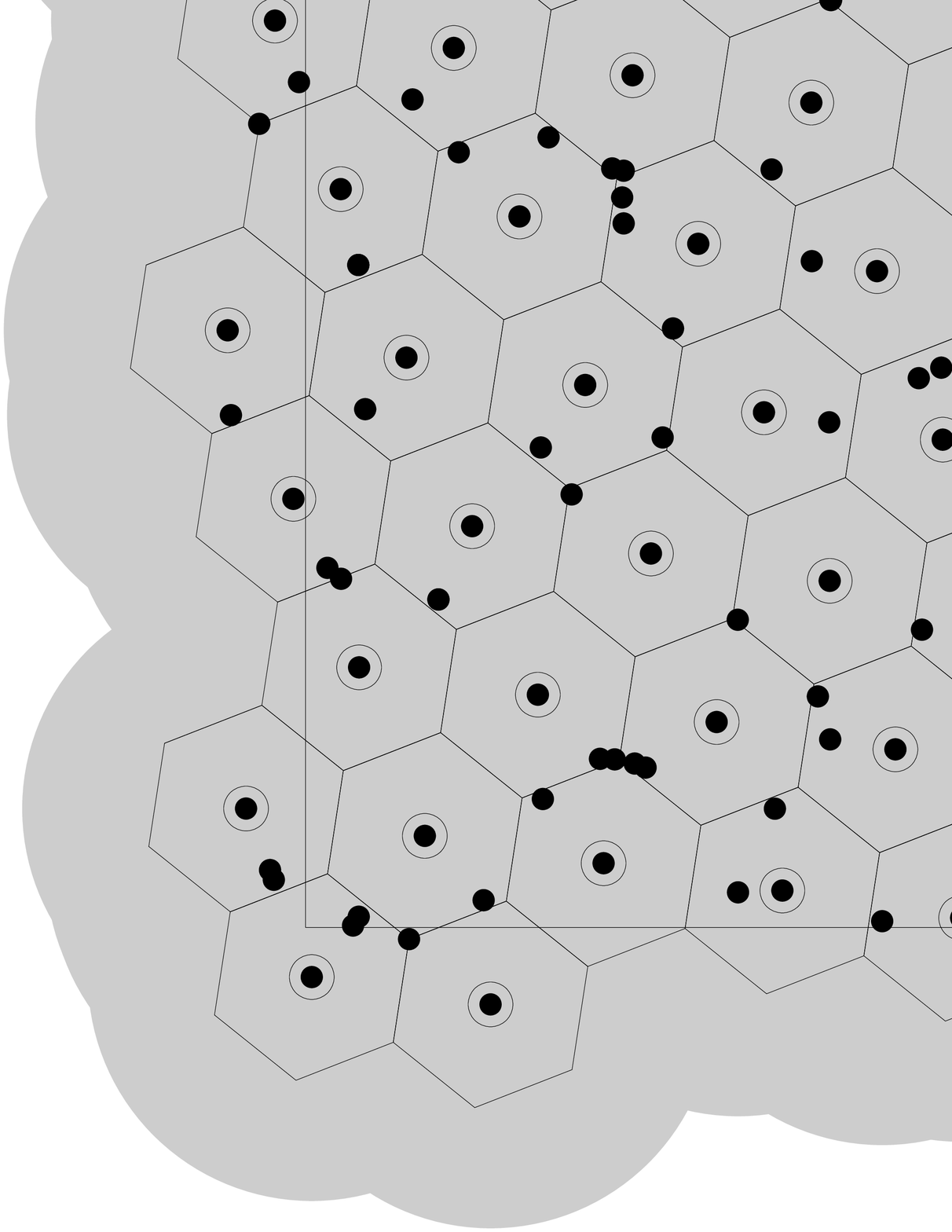}}}
\hspace{5cm}
\subfigure[]{\scalebox{0.04}{\includegraphics[]{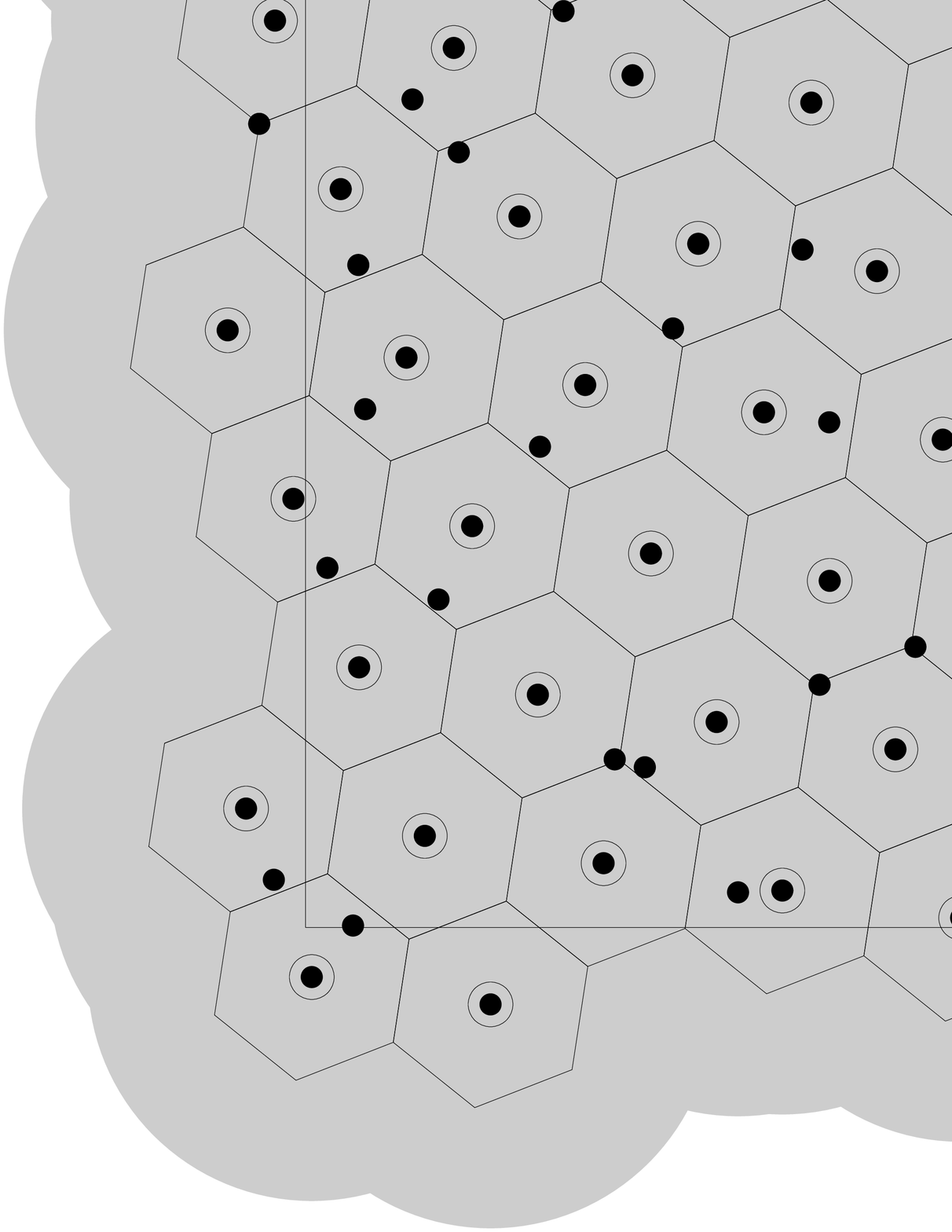}}}
\caption{Coverage of an irregular AoI under
PP2.}\label{fig:grumo_corridoio_pp2}
\end{figure}

We compare our proposal to one of the most acknowledged and
cited algorithms \cite{LaPorta06},
which is based on the construction of the Voronoi diagram determined by the current
sensor deployment. According to this approach, each sensor adjusts
its position on the basis of a local calculation of its Voronoi cell.
This information is used to detect coverage holes and, consequently, calculate
new target locations according to three possible variants.
Among these variants we chose Minimax, that gives better guarantees in terms of coverage extension.
Of this algorithm we adopted all the mechanisms provided to preserve connectivity, to guarantee the algorithm termination, to avoid oscillations and to deal with position clustering
\cite{LaPorta06}. In the rest of this section this algorithm will be named VOR$_\texttt{MM}$.

We set the parameters

$R_\texttt{tx}=11$ m and $R_\texttt{s}=5$ m. Such values satisfy the
VOR$_\texttt{MM}$ requirement $R_\texttt{tx} > 2 R_\texttt{s}$
detailed in \cite{LaPorta06} and do not significantly affect the
qualitative evaluation of \HC. The sensor speed is set to 1 m/sec.

\subsection{Examples of mobile sensor deployment}

We show some examples of deployment evolution under the two \HC\
modes: PP1 and PP2.

Figures \ref{fig:grumo_corridoio_pp1} and \ref{fig:grumo_corridoio_pp2} give a synthetic representation of
how the sensor deployment evolves under PP1 and PP2, respectively,
when
400 sensors are initially located in a high density region.
The AoI has a complex shape in which a narrows connects two
square regions
40 m $\times$ 40 m.
Notice that previous approaches fail
when applied to such irregular AoIs.
For example, VOR$_\texttt{MM}$ does not contemplate the
presence of concavity in the AoI, while the
virtual force based approaches are not able to push sensors through narrows \cite{Howard2002}.

As a second example of sensor deployment, we show experiments conducted with three different starting configurations
over an AoI which is a square 80 m $\times$ 80 m.
\begin{figure}[t]
\begin{center}
\vspace{1cm} \subfigure[]{\scalebox{0.030}{
\includegraphics[]{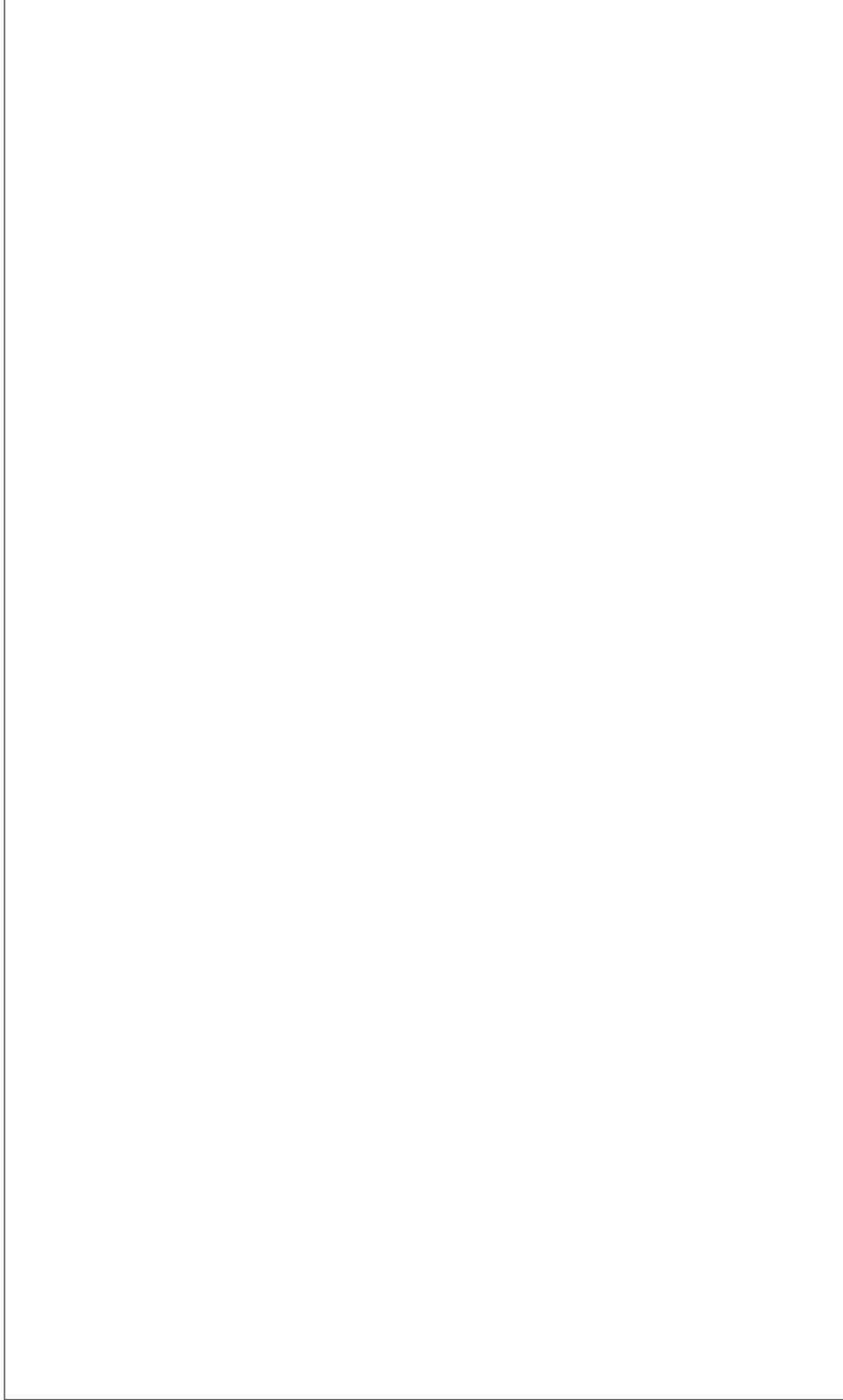}}}
\hspace{3cm} \subfigure[]{\scalebox{0.030}{
\includegraphics[]{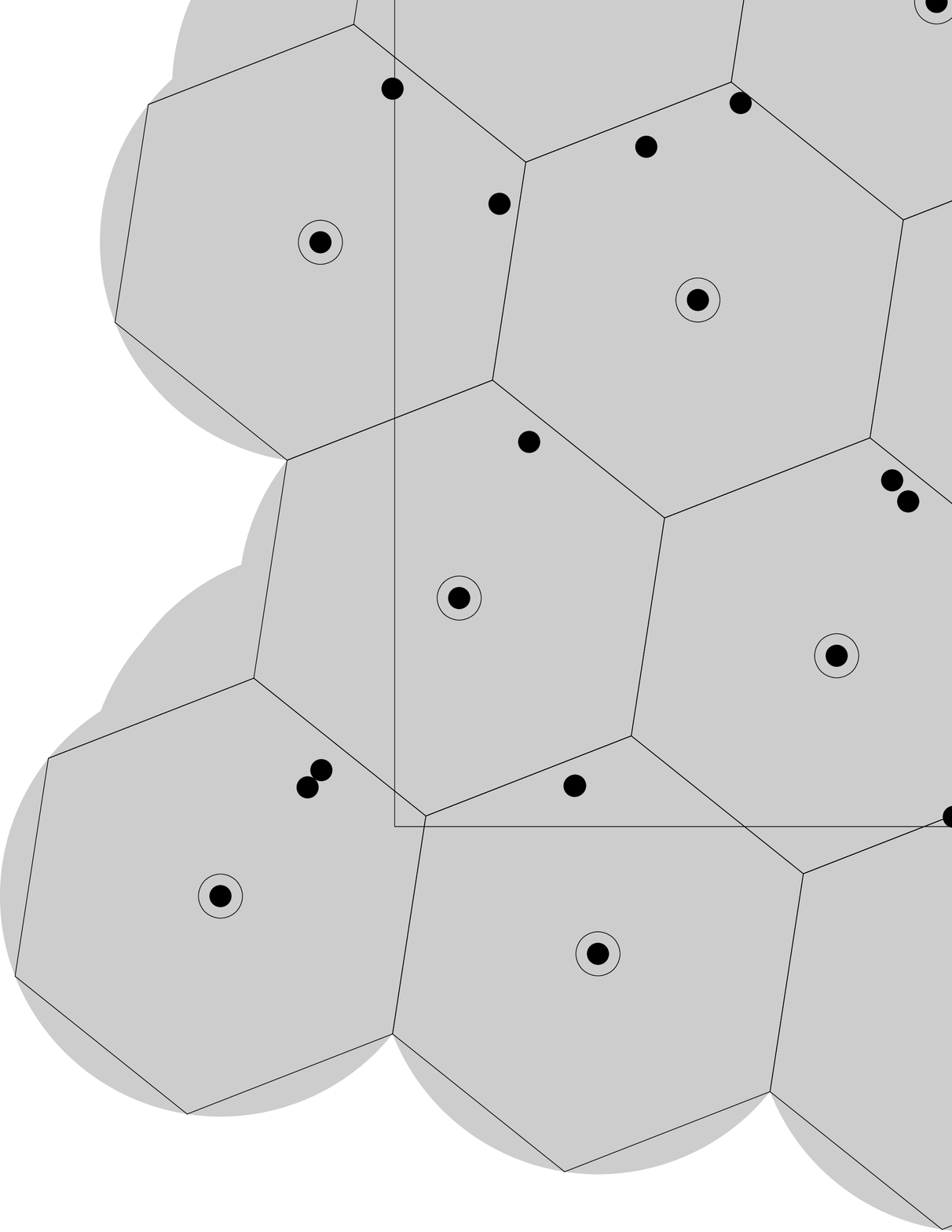}}}
\hspace{3cm} \subfigure[]{\scalebox{0.030}{
\includegraphics[]{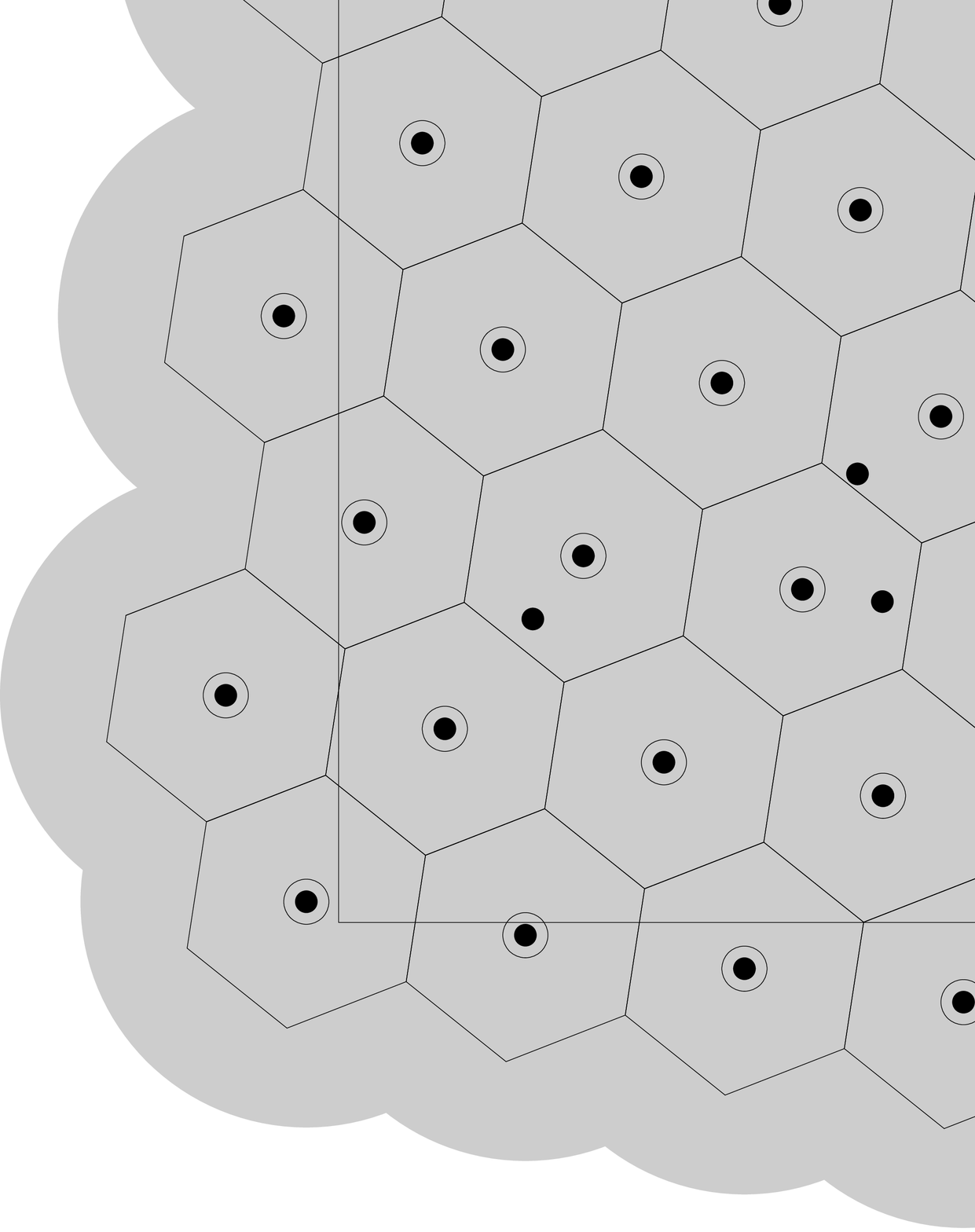}}}
\hspace{3cm} \subfigure[]{\scalebox{0.030}{
\includegraphics[]{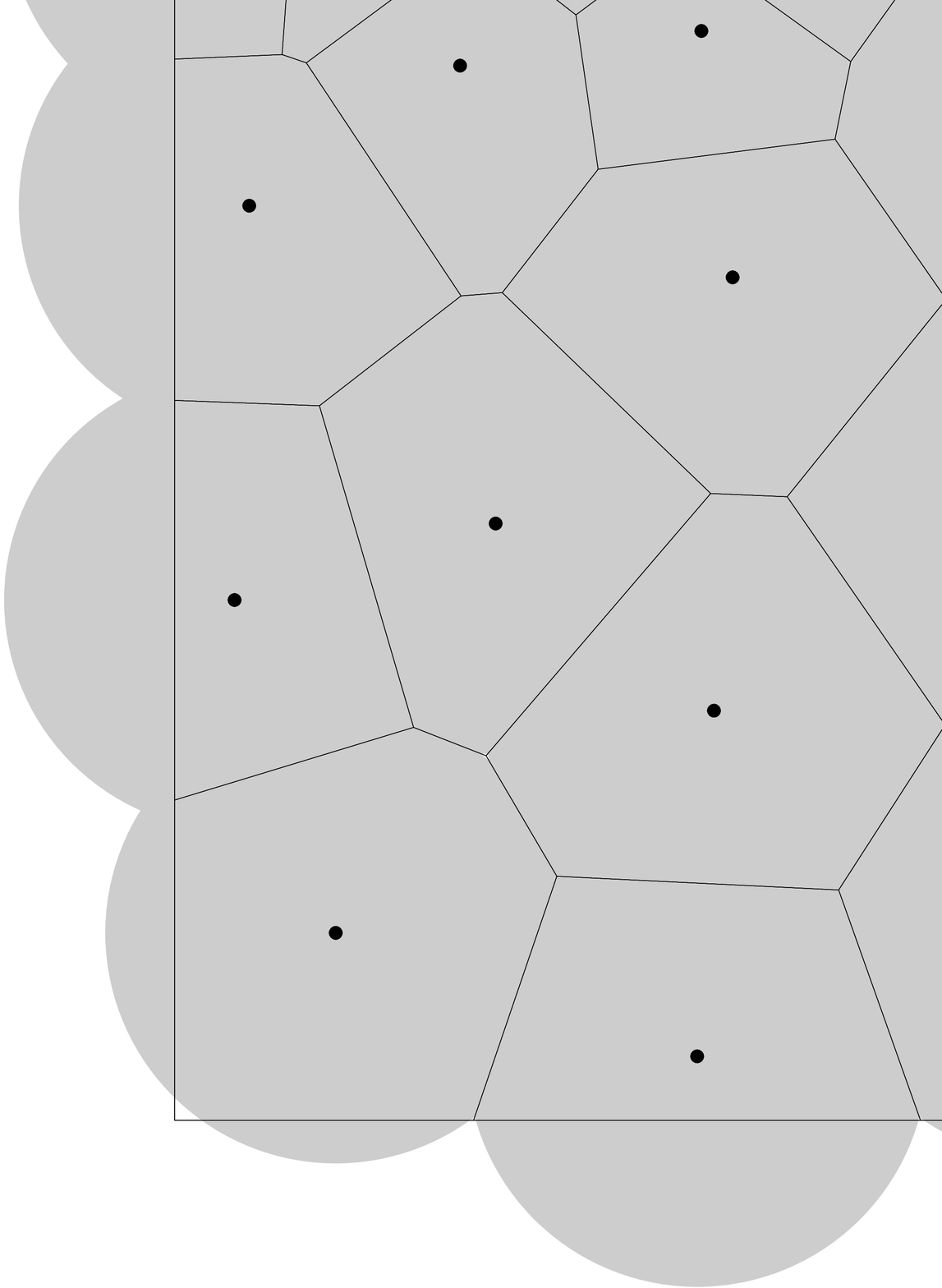}}}
\end{center}
\caption{Trail initial deployment (a) and comparison among PP1 (b), PP2 (c) and VOR$_\texttt{MM}$ (d).} \label{fig:deployment_comparison_SS}
\end{figure}
More precisely, in the first configuration, the initial deployment evidences a trail of sensors which crosses
the AoI, as shown in
Figure \ref{fig:deployment_comparison_SS}(a).
In the second configuration
the sensors are densely  placed in a corner of the AoI,
 as shown in Figure \ref{fig:deployment_comparison_GL}(a).
In the third configuration the initial deployment consists in a high
density region at the center of the AoI. ,

Notice that the first two initial deployments reflect the realistic scenarios
in which sensors are dropped from
an aircraft and sent from a safe location at the boundaries of the AoI.
The third deployment is introduced as is
widely studied in the literature, see for example \cite{LaPorta06} and \cite{Yang2007}.

\begin{figure}[t]
\begin{center}
\vspace{2cm}
\subfigure[]{\scalebox{0.03}{
\includegraphics[]{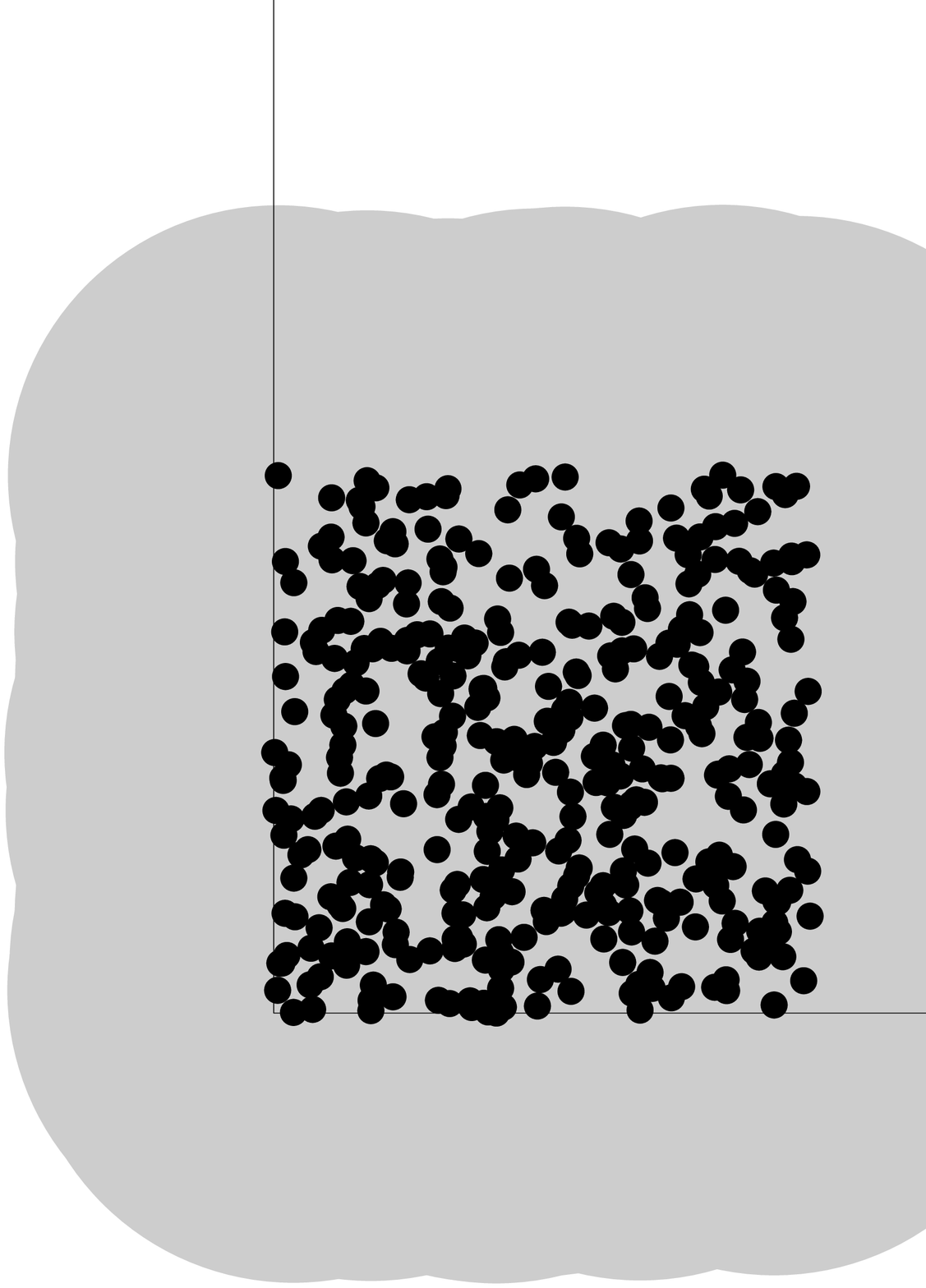}}}
\hspace{3cm} \subfigure[]{\scalebox{0.03}{
\includegraphics[]{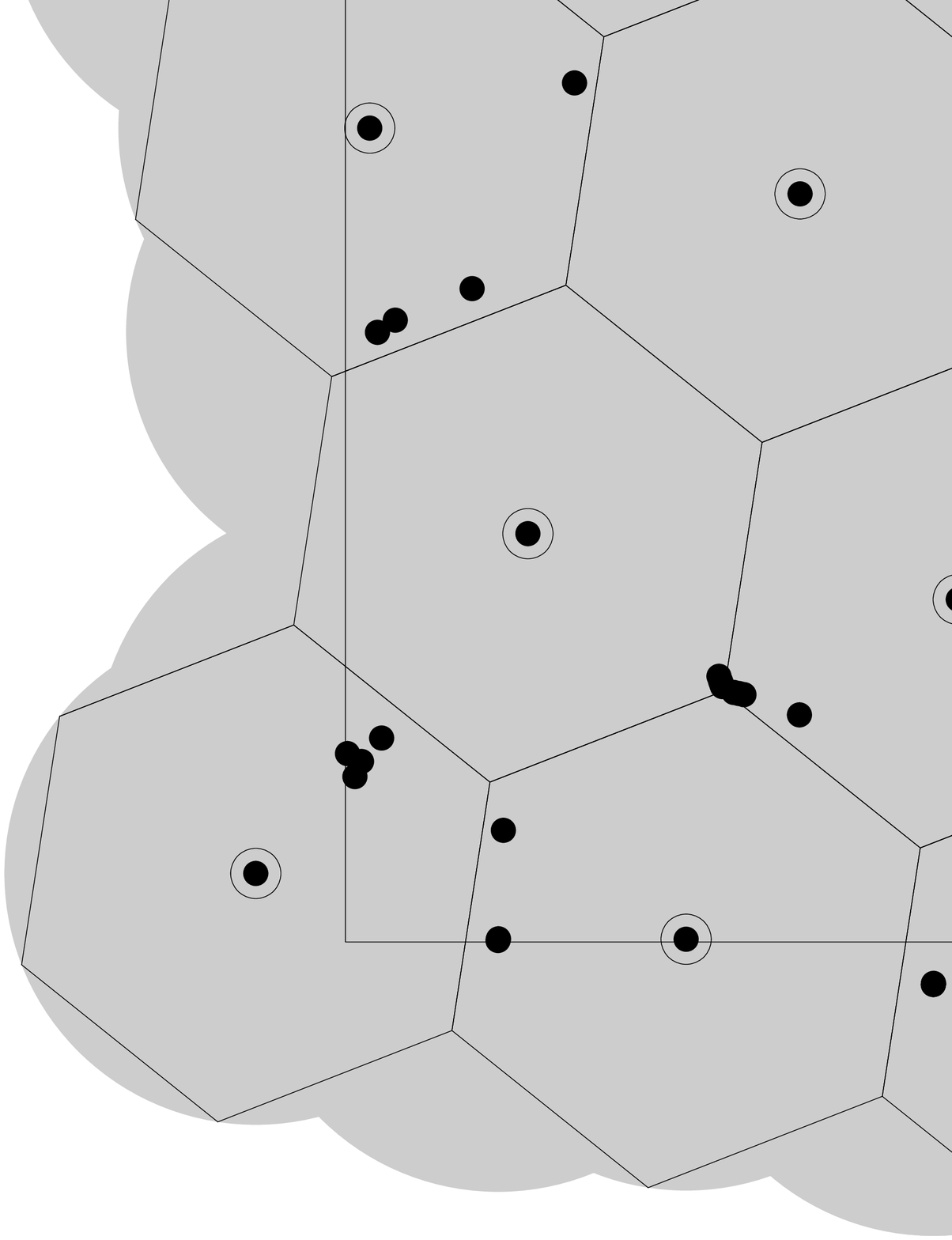}}}
\hspace{3cm} \subfigure[]{\scalebox{0.03}{
\includegraphics[]{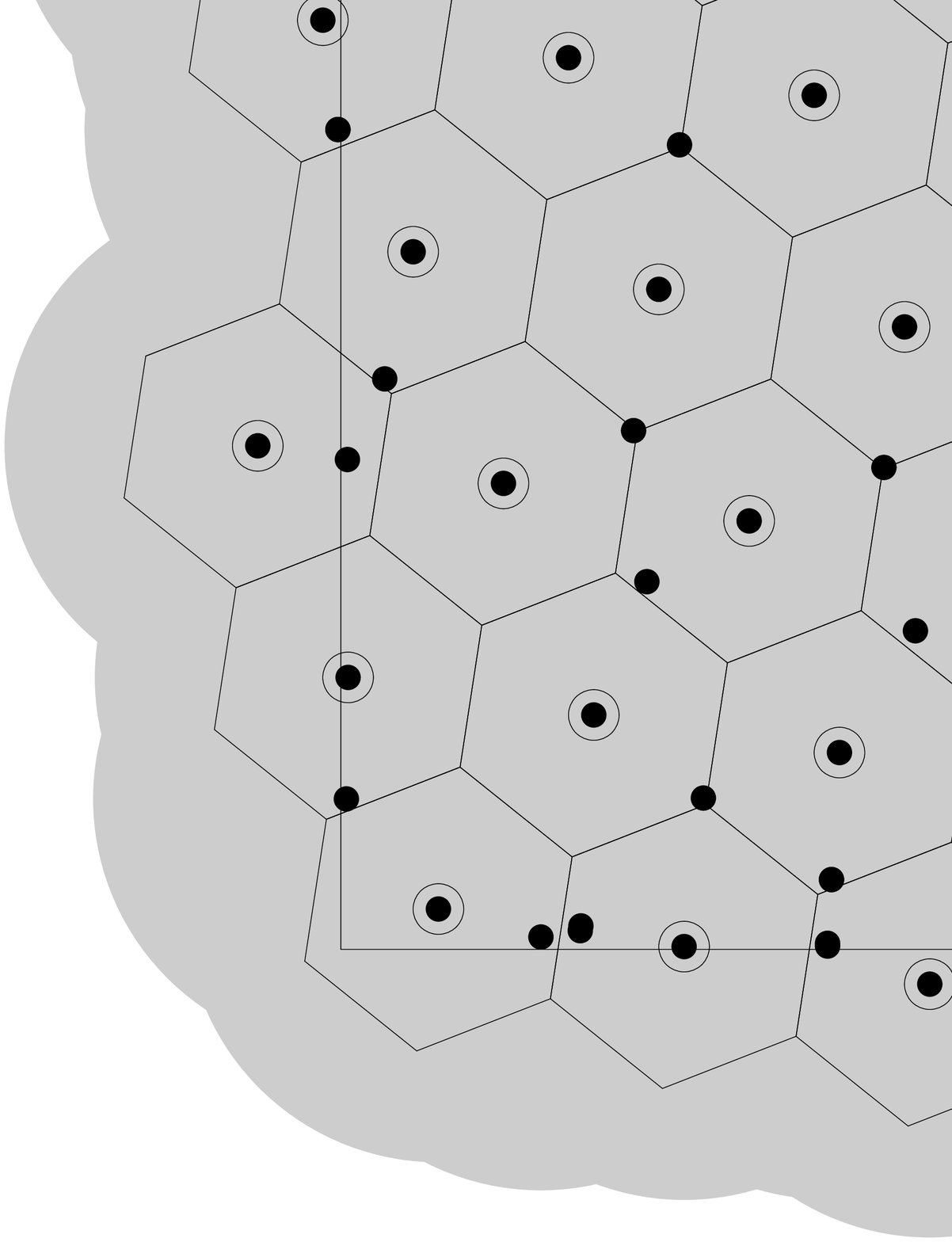}}}
\hspace{3cm}  \subfigure[]{\scalebox{0.03}{
\includegraphics[]{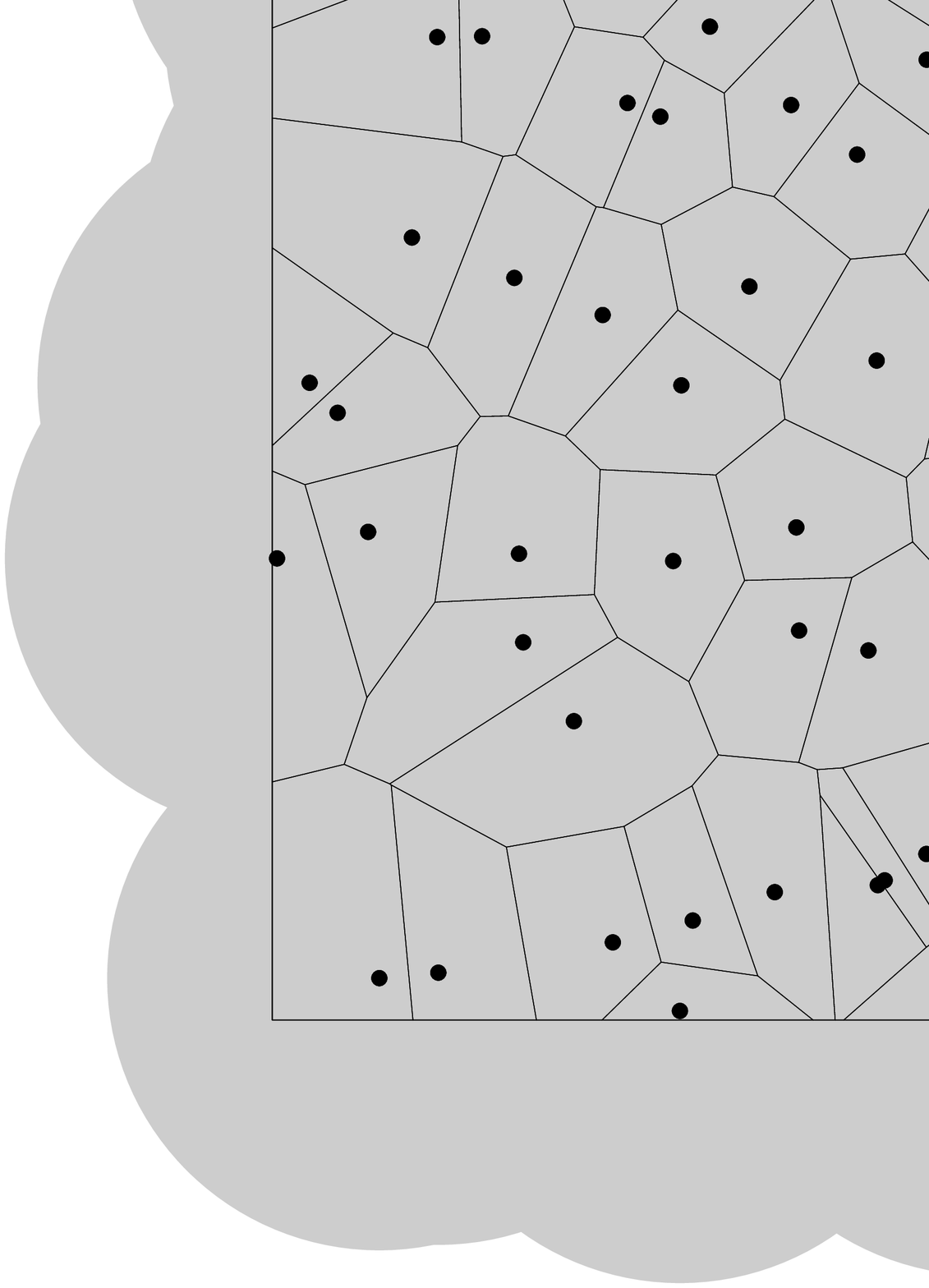}}}
\end{center}
\caption{Safe location initial deployment (a) and comparison among PP1 (b), PP2 (c) and VOR$_\texttt{MM}$ (d).} \label{fig:deployment_comparison_GL}

\end{figure}

In Figures \ref{fig:deployment_comparison_SS} and
\ref{fig:deployment_comparison_GL}, the subfigures indicated with
(b), (c) and (d) show the final deployments achieved by PP1, PP2 and
VOR$_\texttt{MM}$ respectively.

\subsection{Performance comparisons}

In the following we compare the performance of PP1, PP2 and VOR$_\texttt{MM}$ when executed over
a squared AoI, 80 m $\times$ 80 m.

In order to make reliable performance comparisons, we show the average results of 30
simulation runs (conducted by varying the seed for the generation of the initial deployment).

We compare the behavior of the three algorithms
with respect to several performance objectives: energy consumption,
coverage uniformity, termination and coverage completion time.

All the figures from \ref{fig:density} to \ref{fig:time} contain three plots each.
Plot (a) describes the performance obtained when starting from the trail initial deployment,
plot (b) refers to the case
 in which sensors are initially deployed in a safe corner
while plot (c) is related to the case of a dense initial deployment in the center of the AoI.
For a
better readability, we adopt
different scales of the vertical axis for the three
scenarios.

\subsubsection{Coverage uniformity}\label{sec:uniformity_exp}

The three algorithms give
different importance to the uniformity of the coverage.
Indeed,
\HC\ aims at making the coverage as uniform as
possible.

In particular, PP1
builds a coarse grained grid, then it tries to uniform the coverage
only on the basis of a local satisfaction of the Moving Condition.

Instead PP2 constructs a fine grained grid by setting
the hexagon side at the minimum length which guarantees the
full coverage of the AoI, thus making sensors traverse longer distances
than other solutions.

On the contrary, VOR$_\texttt{MM}$ aims
at covering the AoI regardless of the uniformity of the final coverage, and
sensors stop moving when the AoI is fully covered.
\begin{figure*}
\centering
\begin{center}
\subfigure[]{\scalebox{0.32}{
\centering
\includegraphics[height = \textwidth, angle=-90]{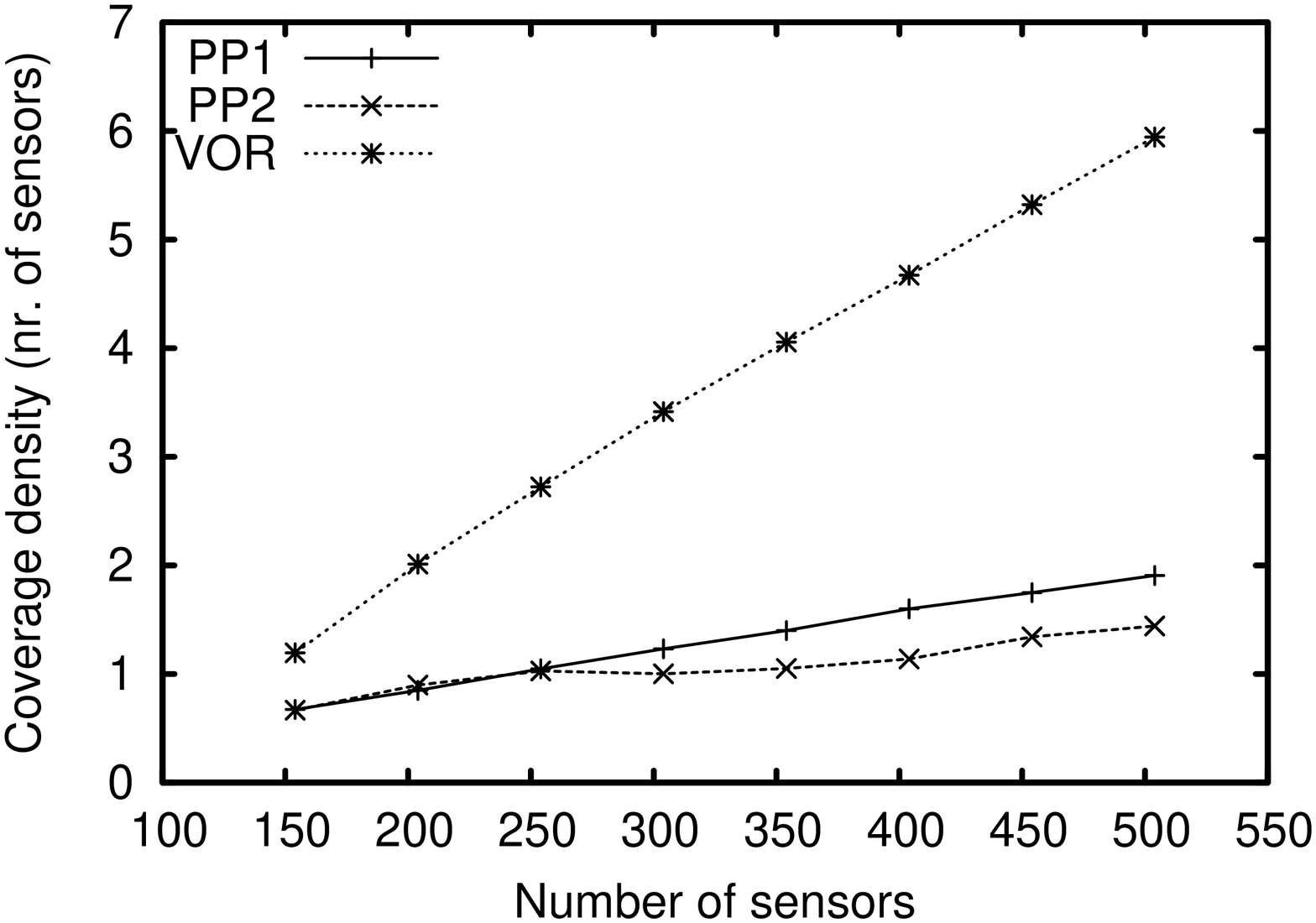}}}
\subfigure[]{\scalebox{0.32}{
\centering
\includegraphics[height = \textwidth, angle=-90]{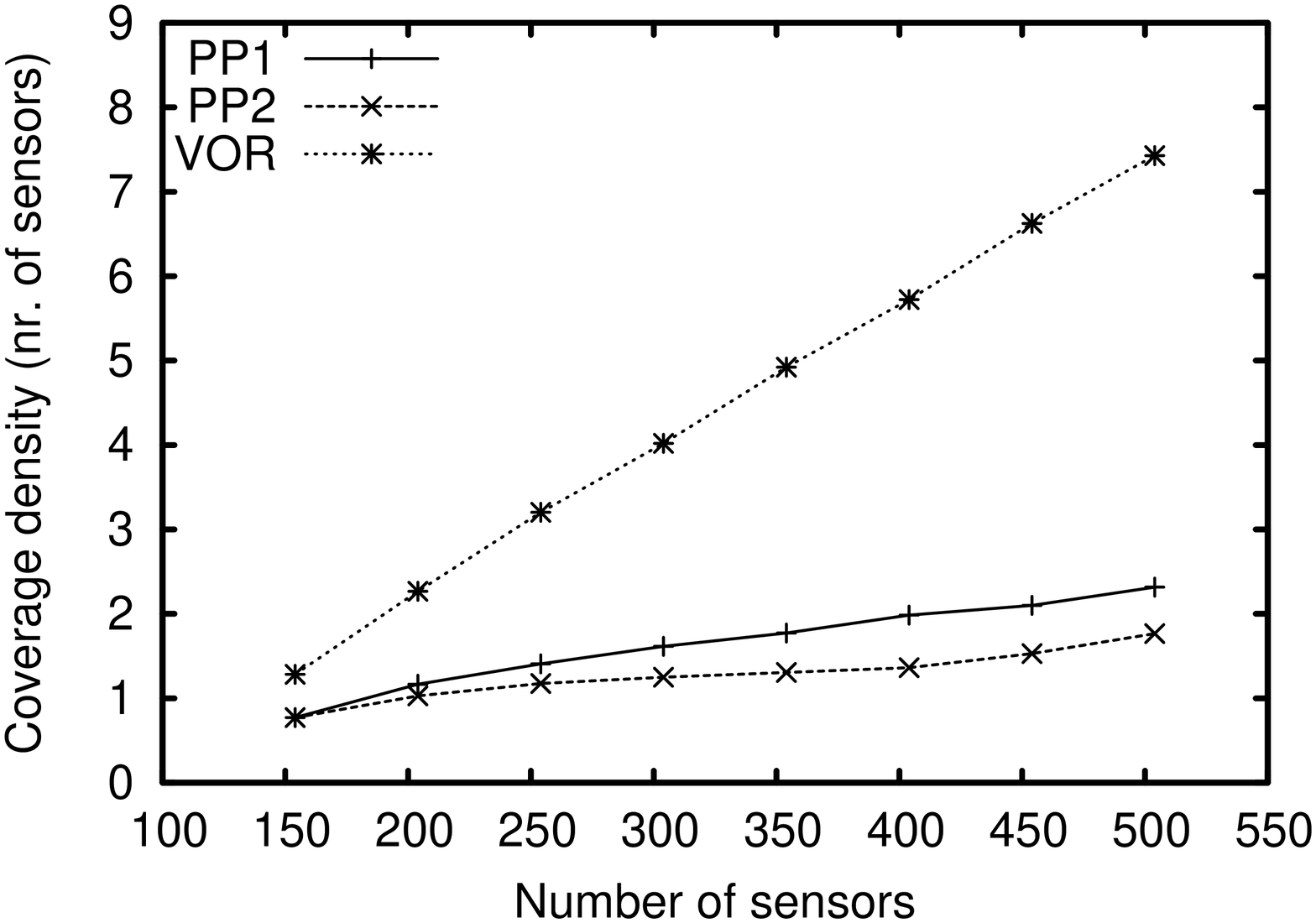}}}
\subfigure[]{\scalebox{0.32}{
\centering
\includegraphics[height = \textwidth, angle=-90]{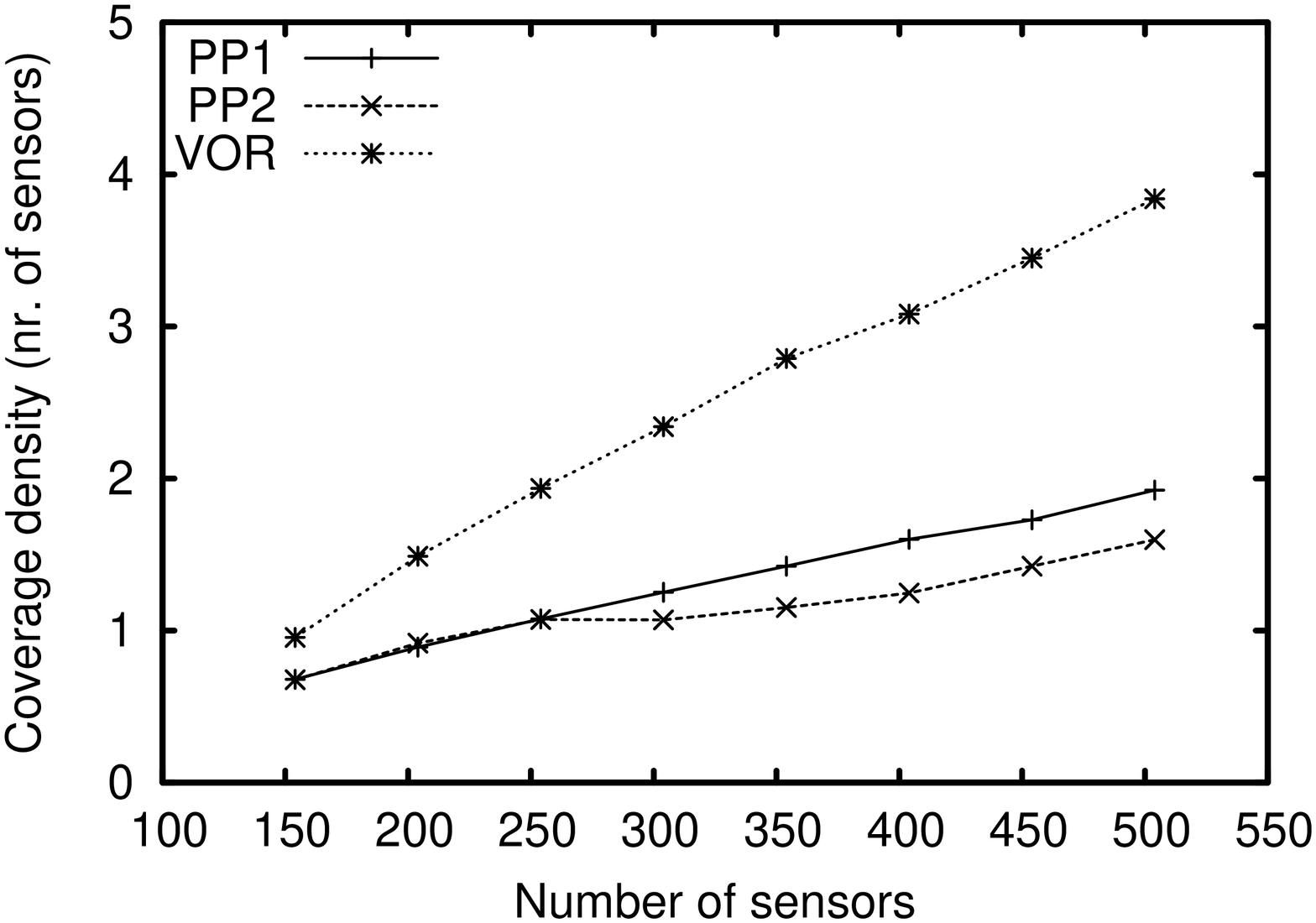}}}
\caption{Coverage density with trail (a), safe location (b) and
central (c) initial deployment.} \label{fig:density}
\end{center}
\end{figure*}

In order to evaluate the coverage uniformity, we compute the coverage density as the number of sensors covering the points of a squared mesh with side 1 m.

Figure \ref{fig:density} shows the standard deviation of the coverage density.
Notice that we do not show the average coverage density because it is not significant,
since it only depends on the number of available sensors.


The standard deviation of the coverage density achieved by PP1 and PP2 is smaller than the one obtained by VOR$_\texttt{MM}$.
In particular, VOR$_\texttt{MM}$ terminates as soon as the AoI is completely covered, without uniforming the density of the
sensor deployment, while PP1 and PP2 keep on moving until they uniform the  coverage.

This result is particularly important as  a uniformly redundant  sensor placement
provides  self-healing and fault tolerance  capabilities.
In the case of PP1, the presence of quite uniformly distributed slaves ensures the self-healing capability of the deployment, while
for what concerns PP2, the guaranteed continuous $k$-coverage gives  tolerance up to $(k-1)$ faults.

\subsubsection{Energy consumption}
We show an analysis of the energy consumption of the three algorithms
in terms of average traversed distance per sensor and average number of starting/braking actions.
Finally we give an overall evaluation which also comprises the
communication costs.

\begin{figure*}
\centering
\begin{center}
\subfigure[]{\scalebox{0.32}{
\centering
\includegraphics[height = \textwidth, angle=-90]{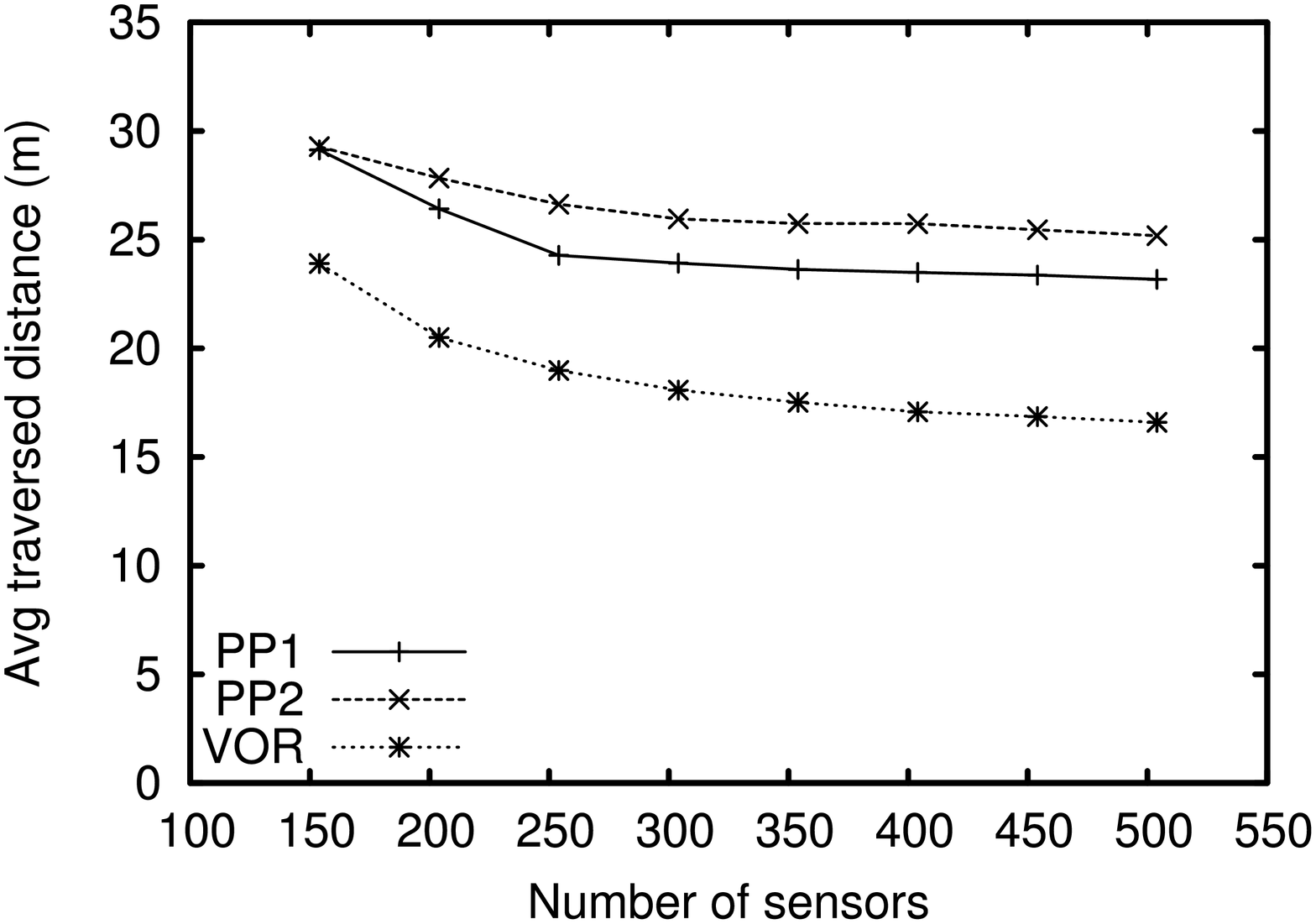}}}
\subfigure[]{\scalebox{0.32}{
\centering
\includegraphics[height = \textwidth, angle=-90]{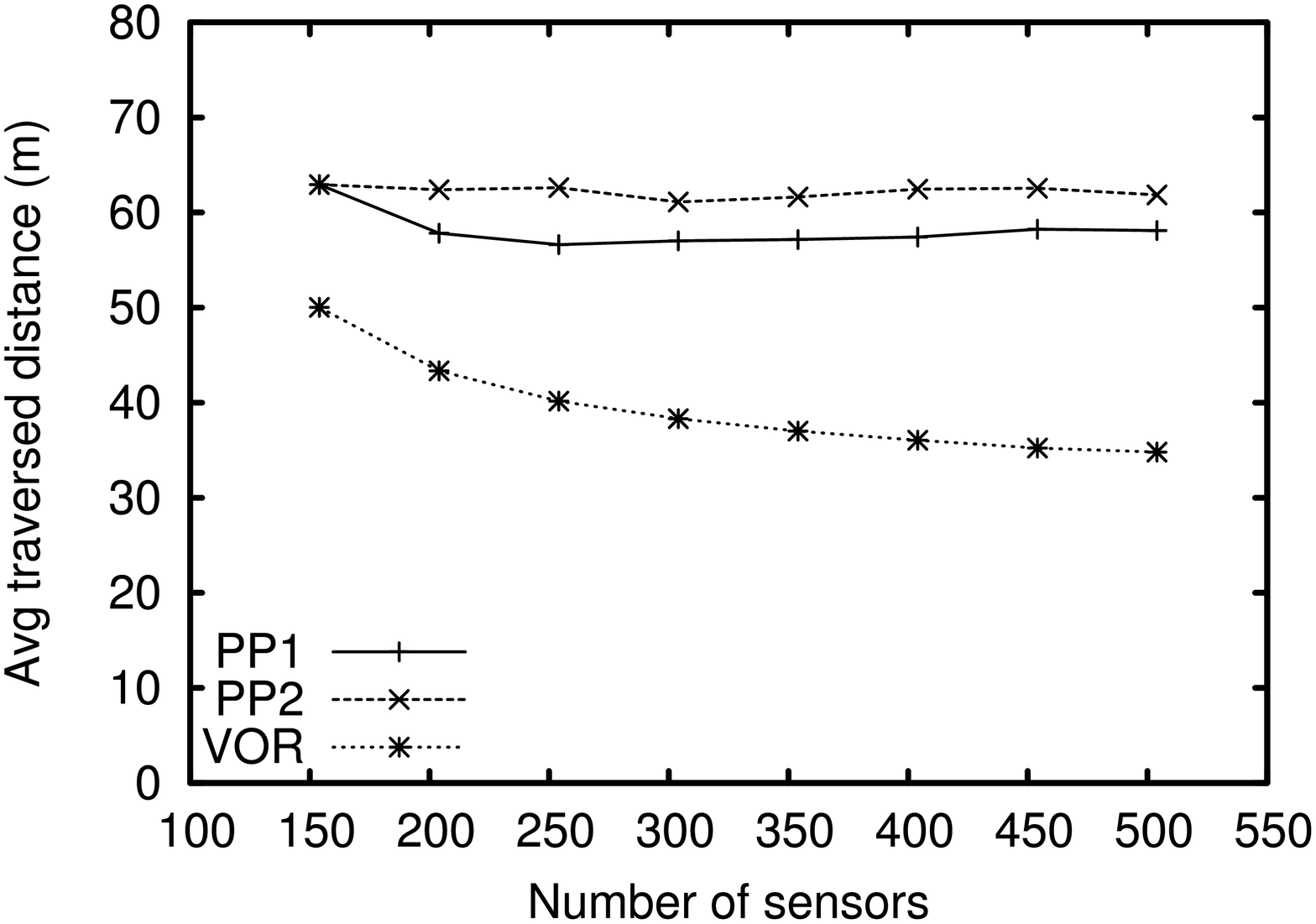}}}
\subfigure[]{\scalebox{0.32}{
\centering
\includegraphics[height = \textwidth, angle=-90]{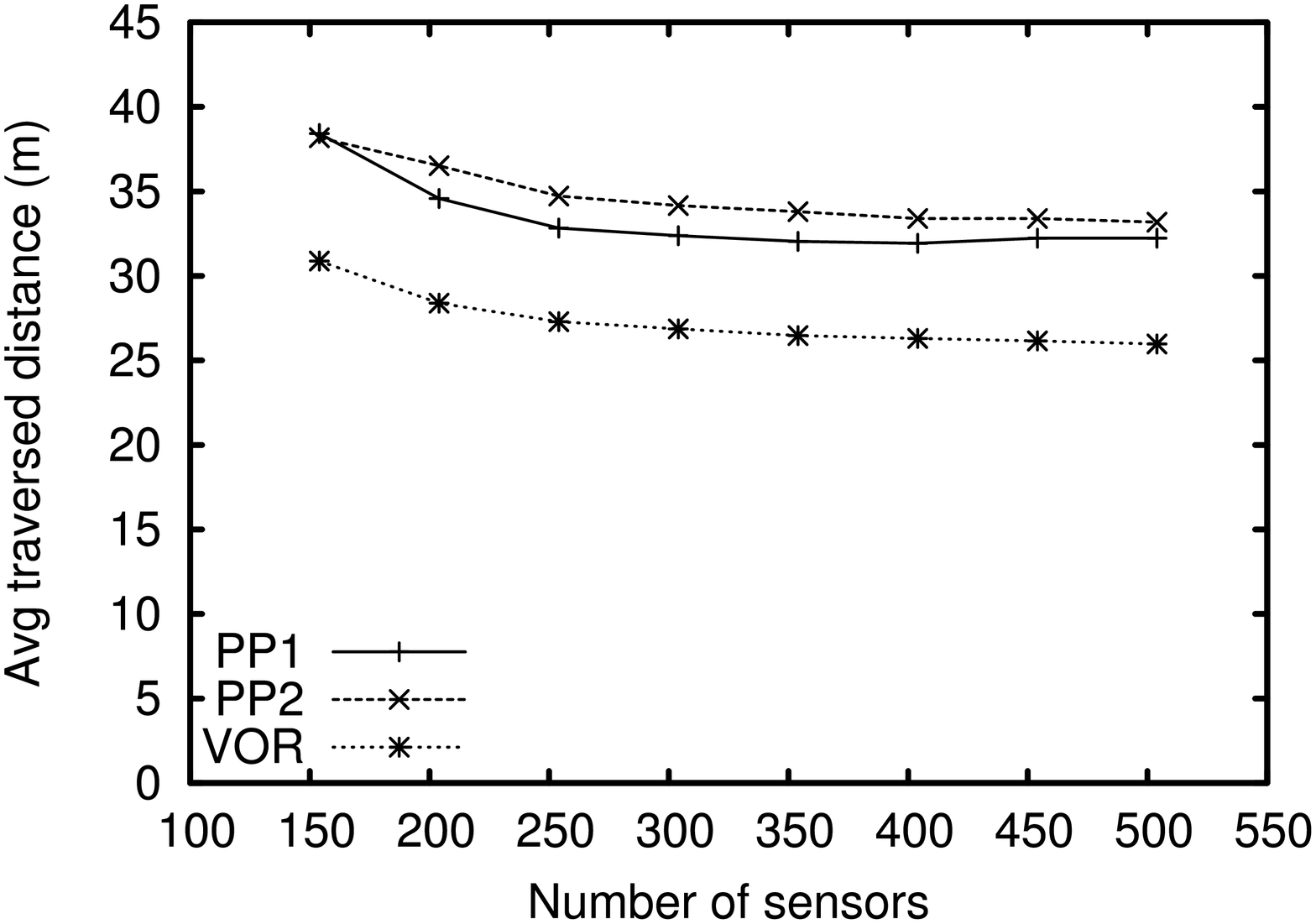}}}
\caption{Average traversed distance with  trail (a), safe location
(b) and central (c) initial deployment.} \label{fig:distance}
\end{center}
\end{figure*}

\paragraph{Average traversed distance per sensor.}
The different weight that the three algorithms give to the uniformity objective is reflected in the
different trends of the average traversed distance shown in Figure \ref{fig:distance}.

The average traversed distance of VOR$_\texttt{MM}$ decreases with the number of sensors.
This is due to the fact that
more and more sensors maintain their initial positions when no
coverage holes are detected. On the contrary, in both modes of \HC,
all sensors contribute to realize a quite uniform coverage, hence
the average traversed distance becomes approximately constant
for large numbers of sensors. This implies that VOR$_\texttt{MM}$ spends less
energy in movements than \HC\ at
the expense of the uniformity of the final coverage, in all the considered settings of the initial deployment.

\paragraph{Average number of starting/braking actions per sensor.}
We now consider the
number of starting/braking actions as
they require a high energy consumption \cite{LaPorta06}.
Figure \ref{fig:startstop} highlights that, when the number of sensors is relatively small, VOR$_\texttt{MM}$ performs a number of starting/braking actions higher than PP1 and PP2. On the contrary, when the number of sensors increases, VOR$_\texttt{MM}$ apparently performs better, showing a rapid decrease of the number of starting/braking actions. This is due to
the presence of a growing fraction of sensors which does not move at all, generating a final non uniform coverage as well as a high energy imbalance among sensors.
The most critical scenario for the VOR$_\texttt{MM}$ algorithm is the safe location initial deployment (notice the different vertical scales in Figure \ref{fig:startstop}).
\begin{figure*}
\centering
\begin{center}
\subfigure[]{\scalebox{0.32}{
\centering
\includegraphics[height = \textwidth, angle=-90]{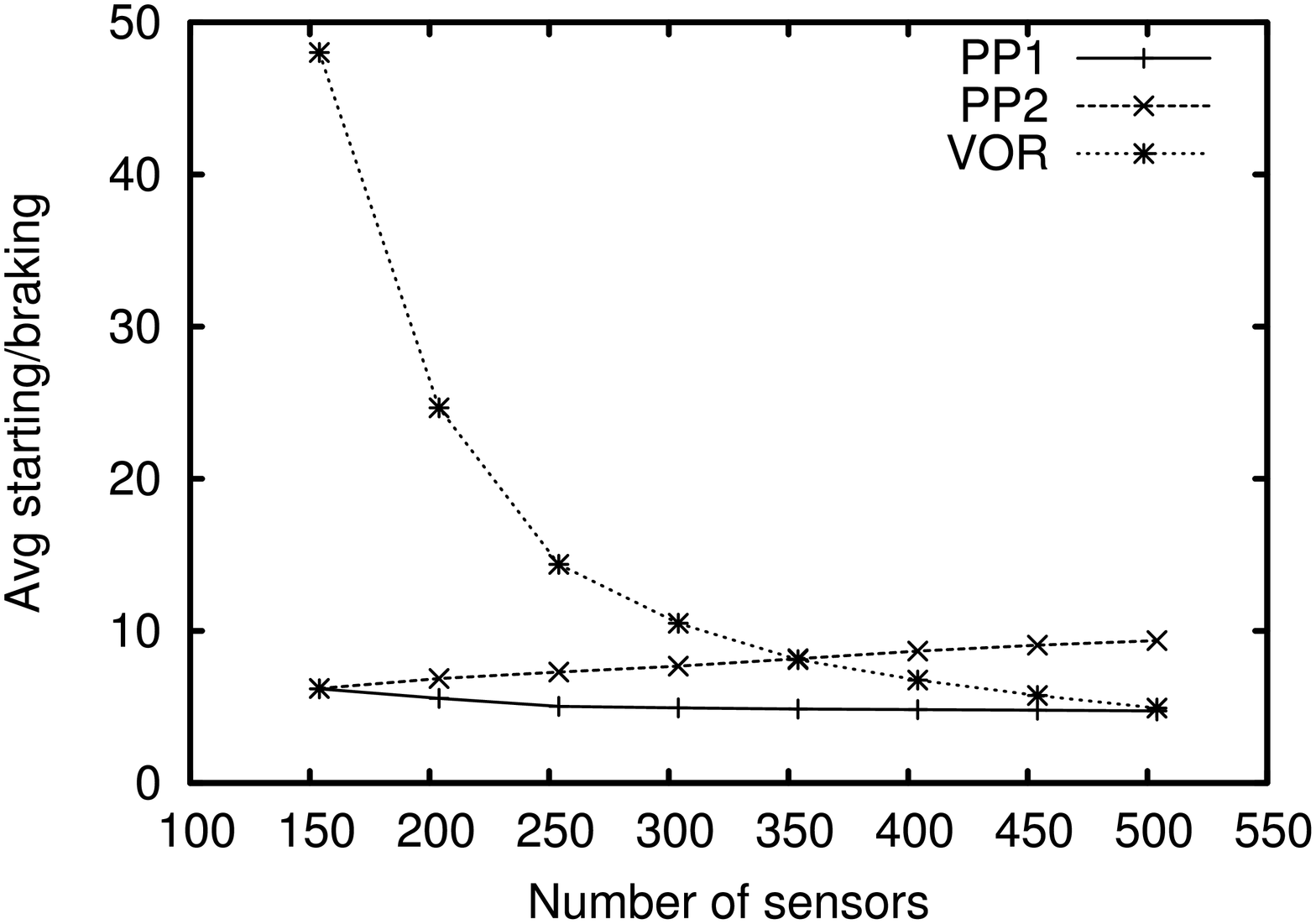}}}
\subfigure[]{\scalebox{0.32}{
\centering
\includegraphics[height = \textwidth, angle=-90]{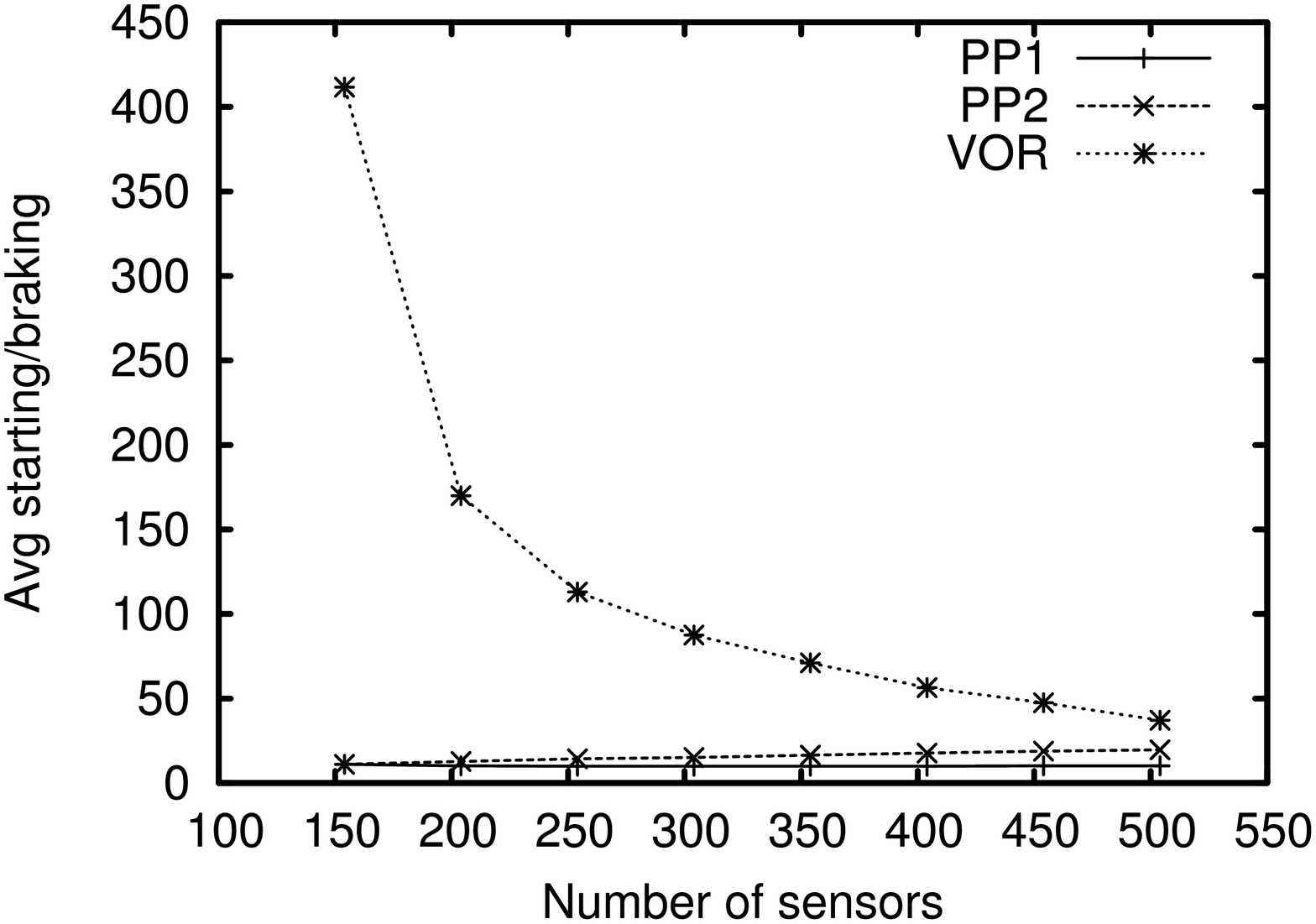}}}
\subfigure[]{\scalebox{0.32}{
\centering
\includegraphics[height = \textwidth, angle=-90]{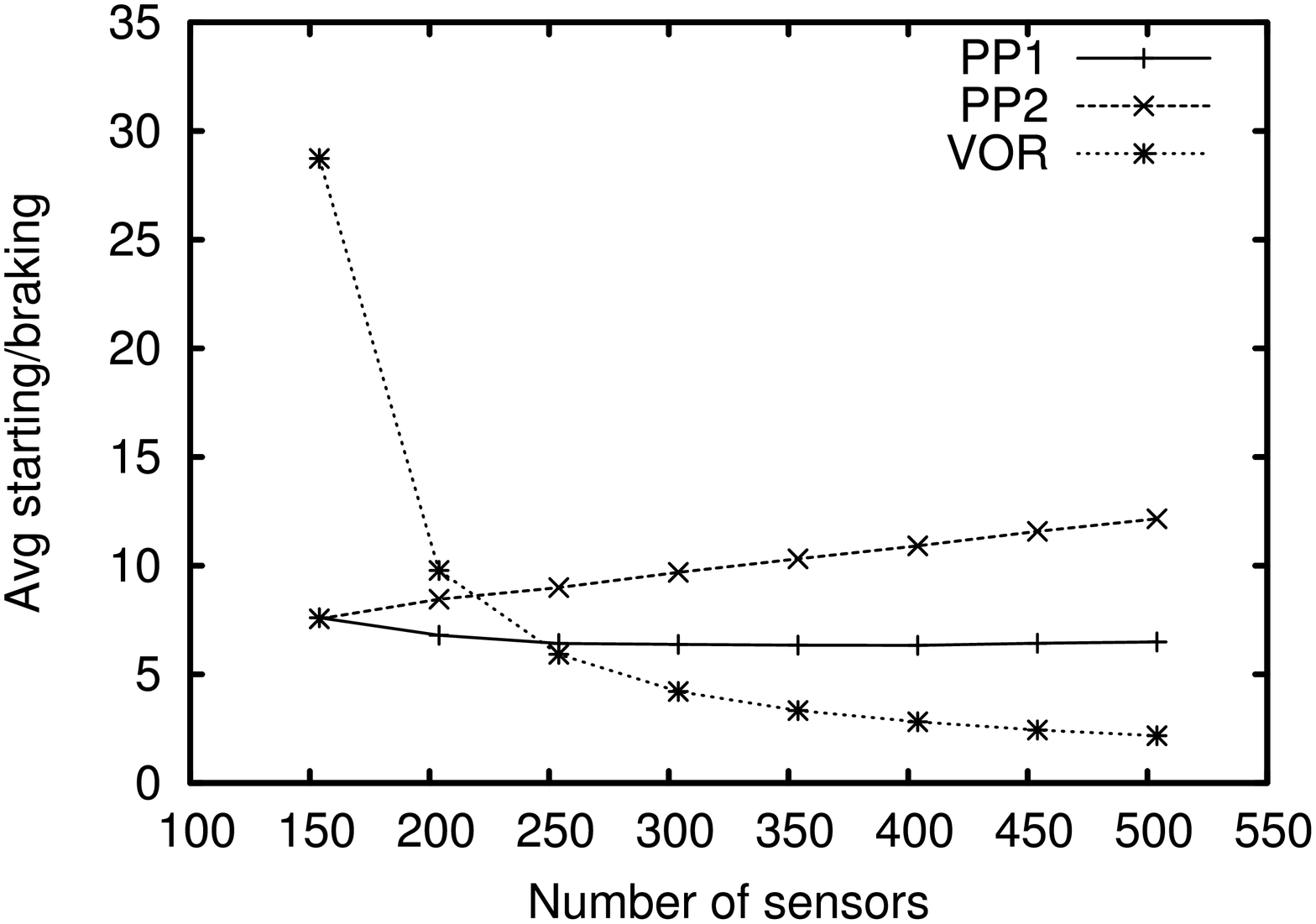}}}
\caption{Average number of starting/braking  with trail (a), safe location (b) and central (c) initial deployment.} \label{fig:startstop}
\end{center}
\end{figure*}

\paragraph{Average energy consumption.}
We now analyze the overall energy consumption of the
three algorithms.
We
utilize a unified energy consumption metric obtained as the sum of
the contributions given by movements, starting/braking actions and
communications.
The energy spent by sensors for communications and movements is expressed in energy units.
The reception of one message corresponds to one energy unit,
a single transmission costs the same as 1.125 receptions \cite{micamote2},
a 1 meter movement costs the same as 300 transmissions \cite{LaPorta06} and a starting/braking action  costs
the same as 1 meter movement \cite{LaPorta06}.

Figure
\ref{fig:energy} shows the energy consumption of PP1, PP2 and VOR$_\texttt{MM}$ in
the three considered scenarios.

\begin{figure*}[h]
\centering
\begin{center}
\subfigure[]{\scalebox{0.32}{
\centering
\includegraphics[height = \textwidth, angle=-90]{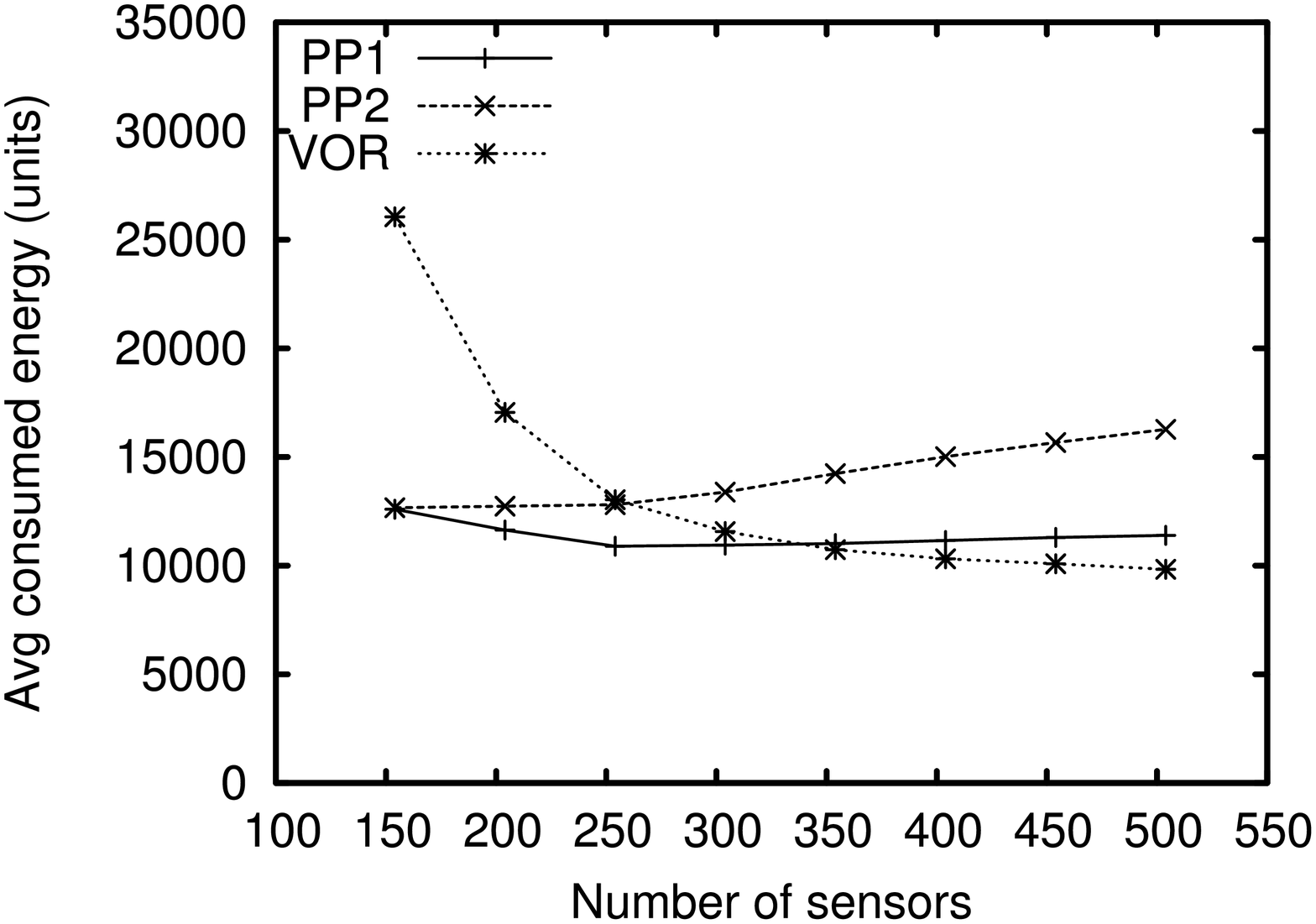}}}
\subfigure[]{\scalebox{0.32}{
\centering
\includegraphics[height = \textwidth, angle=-90]{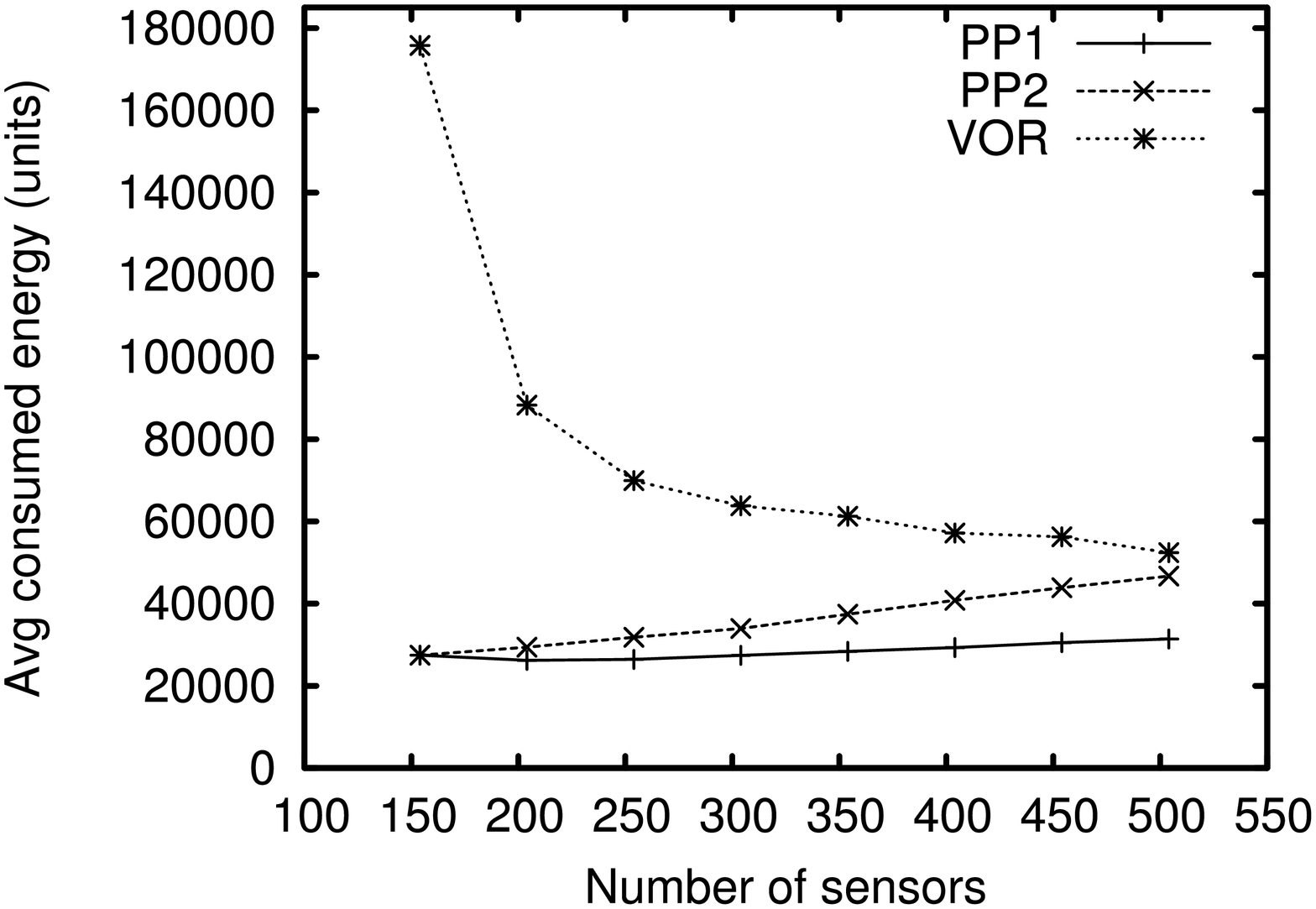}}}
\subfigure[]{\scalebox{0.32}{
\centering
\includegraphics[height = \textwidth, angle=-90]{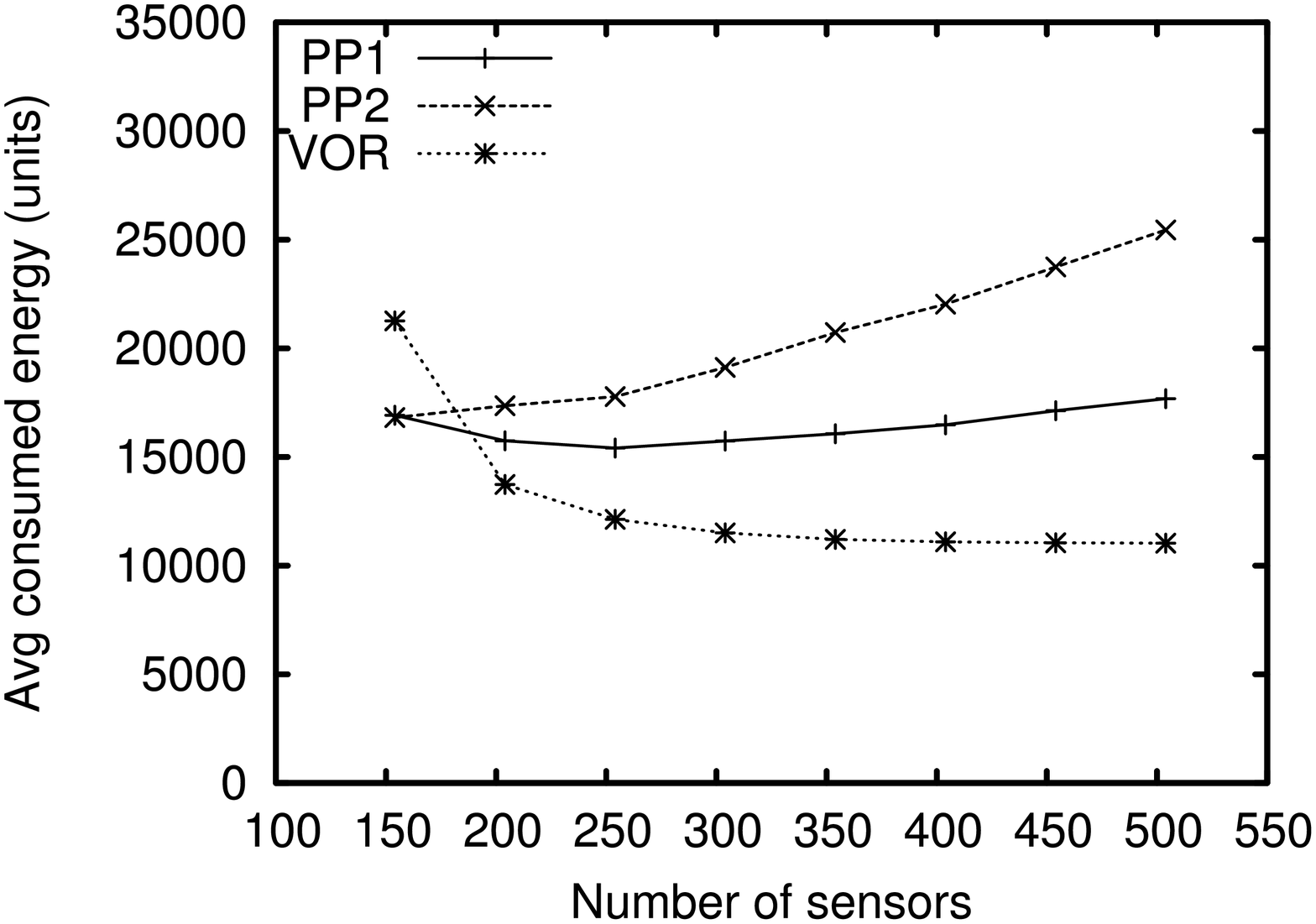}}}
\caption{Average energy consumption with trail (a), safe location (b) and central (c) initial deployment.} \label{fig:energy}
\end{center}
\end{figure*}

\begin{figure*}[h]
\centering
\begin{center}
\subfigure[]{\scalebox{0.32}{
\centering
\includegraphics[height = \textwidth, angle=-90]{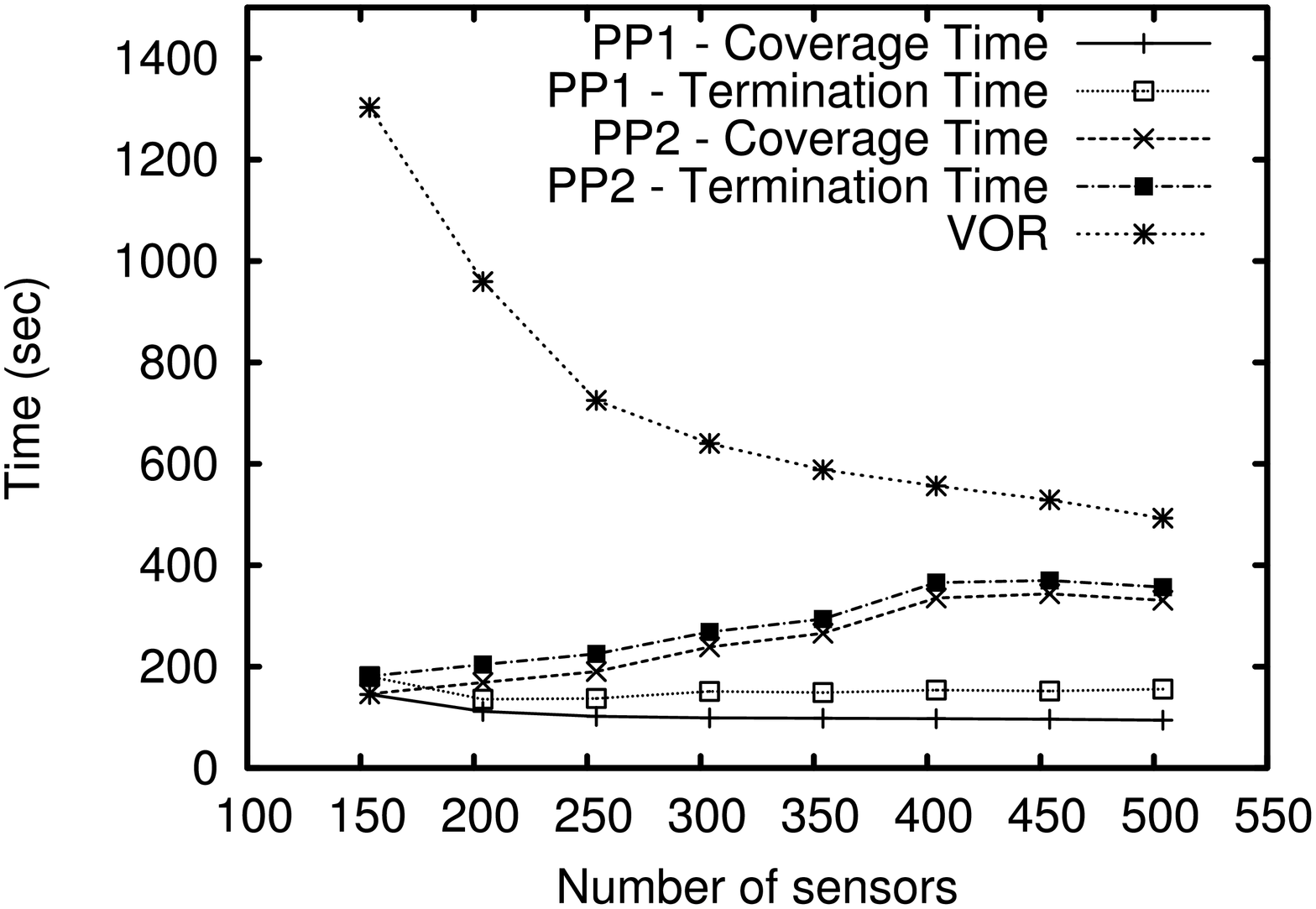}}}
\subfigure[]{\scalebox{0.32}{
\centering
\includegraphics[height = \textwidth, angle=-90]{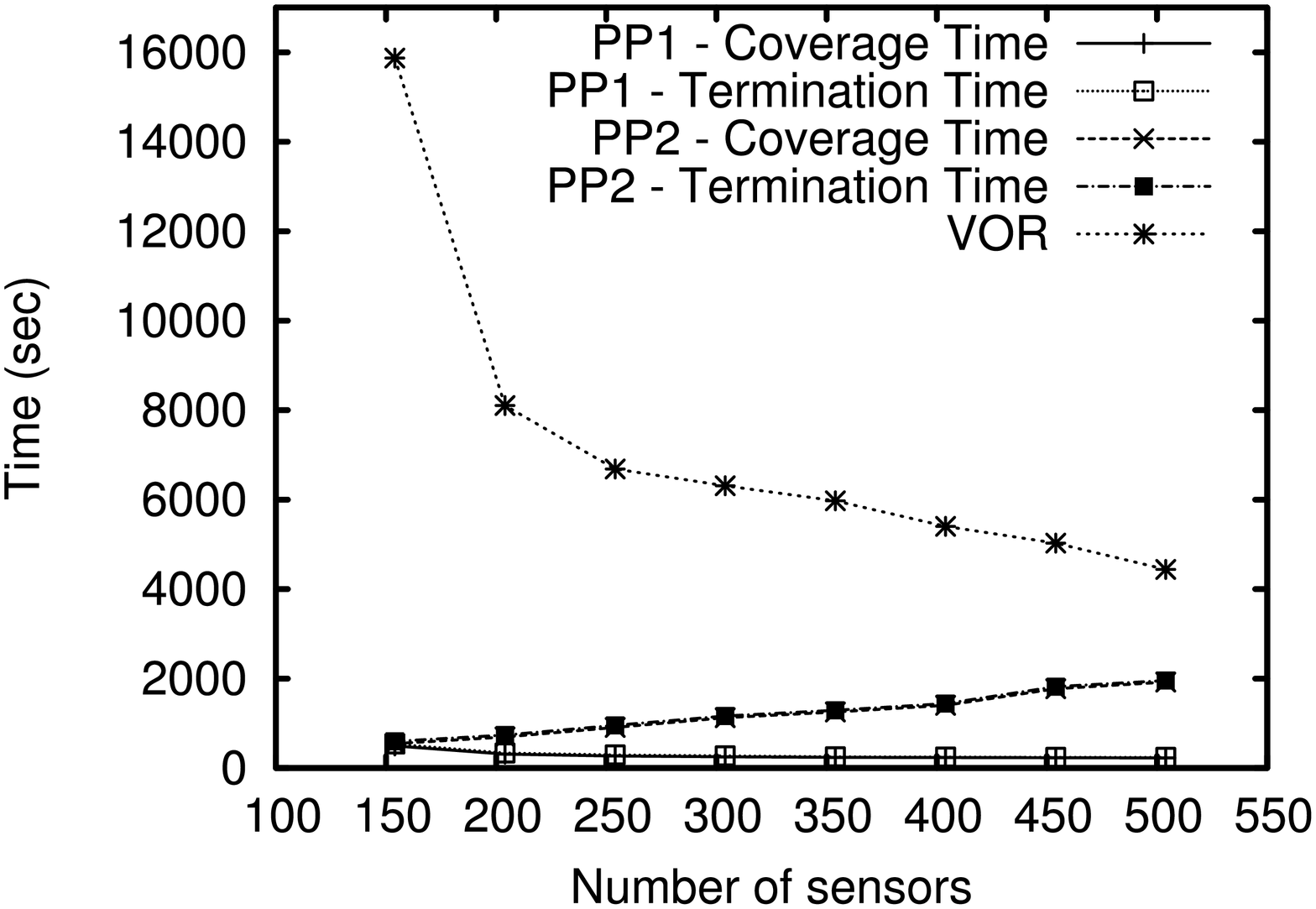}}}
\subfigure[]{\scalebox{0.32}{
\centering
\includegraphics[height = \textwidth, angle=-90]{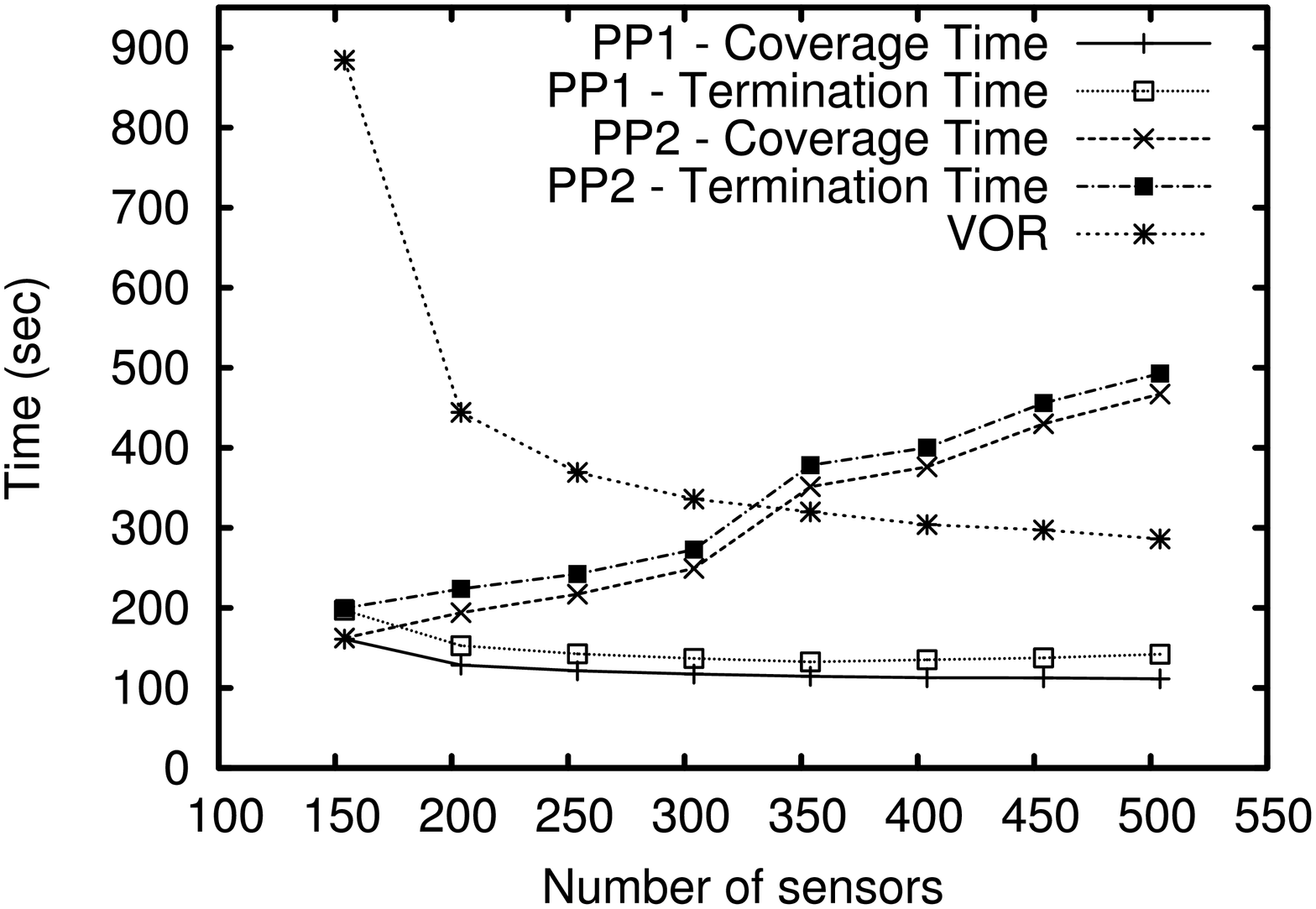}}}
\caption{Termination and coverage time with trail (a), safe location
(b) and central (c) initial deployment.} \label{fig:time}
\end{center}
\end{figure*}

PP1 presents a stable energy
consumption even when the number of sensors varies significantly. Indeed,
although only a fixed number of them are snapped, all sensors are
involved in the push and pull activities, thus improving the
coverage density and uniforming the energy consumption.

PP2 instead, shows that the consumed energy increases as the number of sensors grows.
Indeed, the more numerous are the sensors, the finer is the grid adopted by PP2. Therefore, in order to reach
their destination, the slaves traverse more hexagons, and  are
involved in a higher number of push activities than in the case of PP1. This
increases the number of starting/braking actions as shown in Figure
\ref{fig:startstop} and also increases the consumed energy consequently .

VOR$_\texttt{MM}$ consumes more energy than PP1 and PP2, when the
number of sensors is close to the tight value. Although sensors do
not traverse long distances (as shown in Figure \ref{fig:distance}),
the limit on the maximum moving distance per round required by VOR$_\texttt{MM}$
increases the number of starting/braking actions (see Figure
\ref{fig:startstop}), thus resulting in a high energy consumption.
This effect is particularly evident in the case of the safe location scenario
shown in Figure \ref{fig:energy}(b).

The average energy
consumption of  VOR$_\texttt{MM}$  decreases when increasing the number of sensors.
Notice that this is not
due to a better behavior of the algorithm but to the fact that a greater and greater fraction of sensors
do not move at all. This implies that a considerable number of sensors
consume a large amount of energy to move from overcrowded
regions toward uncovered areas. As soon as all the coverage holes
are eliminated, VOR$_\texttt{MM}$ stops, leaving some zones with very low density
coverage. These zones are prone to the occurrence of coverage holes in case of failures,
as the sensor density is very scarce and
the only sensors located in proximity have already consumed much energy during the
network deployment.

Although
PP2 consumes more energy when the number of the available sensors grows,
it guarantees a more uniform coverage with respect to VOR$_\texttt{MM}$ and PP1.
Moreover, the regularity of the final deployment enables the
use of topology control algorithms \cite{Poduri2007} that permit
a selective sensor activation, saving energy during the operative
phase which follows the deployment.
\subsubsection{Coverage completion and termination time}

Figure \ref{fig:time} shows the coverage and termination time for the three algorithms.
Notice that for VOR$_\texttt{MM}$ the termination and coverage completion times coincide, while
for \HC\ some more movements are still executed even after the coverage completion.

In the three considered scenarios, if the number of sensors available is close to the minimum needed to cover the AoI, VOR$_\texttt{MM}$ requires a very long time to complete the coverage, while \HC\ terminates  much earlier.
When the number of available sensors grows, VOR$_\texttt{MM}$ has a shorter termination time,
which instead remains stable under PP1.
On the contrary, the termination time of PP2 grows when the  number of available sensors increases.
In particular, VOR$_\texttt{MM}$ generally requires more time than PP1 to achieve its final coverage.
Only in the case of the central initial deployment, and for a high number of available sensors (N greater than 320) VOR$_\texttt{MM}$ terminates in a shorter time if compared with PP2 (see Figure \ref{fig:time}(c)).
This is due
to the fact that the termination time of PP2 is  delayed by the numerous hole triggers generated by the  pull activity.

It is worth noting that as already discussed, the safe location deployment,
constitutes a critical scenario for VOR$_\texttt{MM}$ as this algorithm works at
its best for more uniform initial sensor distributions. Indeed,
Figure \ref{fig:time}(b) shows that VOR$_\texttt{MM}$ requires much more time than
in the other sets of experiments, (a) and (c), to achieve its final deployment
(16000 sec in the case of safe location vs. 1400 sec in the case of trail, and
900 sec in the case of central initial deployment).

\section{Related Work}\label{sec:related_work}
There is an impressively growing interest in
self-managing systems, starting from
several industrial initiatives from IBM \cite{IBM-manifesto}, Hewlett Packard
\cite{HP-design-principles} and Microsoft  \cite{Microsoft_DSI}.
Various approaches have been proposed to self-deploy mobile sensors
although few of them can be actually considered autonomic.
The majority of these works are either based on the virtual force approach (VFA) or
on computational geometry techniques.

The virtual force approach (VFA)
\cite{Zou2003,Heo2005,Chen03} models the
interactions among sensors as a combination
of attractive and repulsive forces. This approach
requires the definition of thresholds to
determine the magnitude of the force one sensor exerts on another.
As shown in \cite{Chen03}, the VFA presents oscillatory sensor behavior.
This problem is addressed by defining further arbitrary
thresholds as stopping conditions. The tuning of
such thresholds is laborious and relies on an
off-line configuration.
In addition, it influences the resulting deployment, the
overall energy consumption and the convergence
rate.
Moreover, this approach does not guarantee the coverage in presence of
narrows.
A variation of the VFA is presented in \cite{Suckme2004}
where the introduction of two virtual forces guarantees better uniformity
by providing at least $K$ neighbors to each sensor.
Other approaches are inspired by physics as
well, such as \cite{Pac2006} and \cite{Kerr2004}.
In \cite{Pac2006}  the sensors are modelled as
particles of a compressible fluid and regulates
their movement mimicking a diffusive behavior. In
\cite{Kerr2004}  two approaches that make use of
gas theory to model sensor movements in presence of obstacles are proposed.
However the last three approaches still suffer
from oscillatory sensor behavior.
The work \cite{Chiasserini07} introduces a unified solution for
sensor deployment and relocation which also makes use of the virtual force approach.
This proposal
deals with a rather different problem with respect to ours. Indeed this work is designed for
an open environmental setting, namely where the target area is not determined prior to the deployment.

By contrast, the  techniques based on computational geometry, model
the deployment problem in terms of Voronoi diagrams or Delaunay triangulations.

The Voronoi approach (VOR$_\texttt{MM}$) is detailed in \cite{LaPorta06}.
According to this proposal, each sensor iteratively calculates its own
Voronoi polygon, determines the existence of coverage holes
and moves to a better position if necessary.
In this approach the relationship between the
transmission and the sensing range influences the
obtained performances by
either moving sensors toward already covered
positions or reducing the resulting covered area.
Furthermore, this approach is not designed to
improve the uniformity of an already complete
coverage.
According to \cite{Yang2007} each sensor makes a rough evaluation of the local density and calculates the movements needed to reach a final position that is as close
as possible to the points of a hexagonal tiling.
This is done by locally constructing the Delaunay
triangulation determined by the current sensor placement.
This approach suffers from similar limitations
to the VFA and does not guarantee
oscillation avoidance if proper threshold parameters
are not set.

In \cite{LaPorta04} the authors analyze the problem of sensor deployment
in a hybrid scenario, with both mobile and fixed sensors in the same environment.
They introduce the
general concept of logical movements.
Instead of moving iteratively, sensors
calculate their target locations based on a distributed iterative
algorithm, move logically, and exchange new logical
locations with their new logical neighbors. Actual movement
only occurs when sensors determine their final
locations, thus sparing energy by avoiding
zig-zag motions
at the expense of some more messaging activity.

A different approach is proposed in  \cite{Tan08}, which introduces a technique for sensor deployment for operative settings where the sensing radius is relatively large, hence
coverage does not necessarily imply connectivity.
These operative settings are not addressed by our paper which instead deals with the most common
types of devices for which the relation between the sensing and the transmission radius is such that
the achievement of a complete coverage  also guarantees network connectivity.

\section{Conclusions and future work} \label{sec:conclusions}

We proposed an original algorithm for mobile sensor self deployment named \HC.
According to our proposal, sensors autonomously coordinate their movements in order to achieve a complete and uniform coverage with moderate energy consumption.
The execution of \HC\
does not require any prior knowledge of the operating conditions
nor any manual tuning of key parameters, as sensors adjust their positions on the basis of
locally available information.
The proposed algorithm leads to a guaranteed final static and uniform coverage,
provided that there is a sufficient number of
sensors.
As experiments show, \HC\ outperforms previously proposed
approaches thanks to its ability to cover target
areas of even irregular shape.
Mechanisms for obstacle detection and avoidance are being investigated and considered as future extensions of this work.

\bibliographystyle{IEEEtran}

\bibliography{IEEEabrv,bibliografia}

\end{document}